\documentclass[a4paper, onecolumn, 12pt, accepted=2024-01-26]{quantumarticle}
\pdfoutput=1

\usepackage{framed}
\usepackage{mdframed}
\mdfsetup{frametitlealignment=\center}
\usepackage{amsmath, amssymb}
\usepackage{amsthm}

\makeatletter
\newtheorem*{rep@theorem}{\rep@title}
\newcommand{\newreptheorem}[2]{%
	\newenvironment{rep#1}[1]{%
		\def\rep@title{#2 \ref{##1}}%
		\begin{rep@theorem}}%
		{\end{rep@theorem}}}
\makeatother
\newreptheorem{theorem}{Theorem}
\usepackage{thmtools}
\usepackage{thm-restate}

\usepackage{tikz}
\usepackage{subcaption}
\usetikzlibrary{decorations.markings,decorations.pathreplacing,calc, shapes, arrows}
\usetikzlibrary{arrows.meta}
\usetikzlibrary{positioning}
\usetikzlibrary{calc}
\tikzset{main node/.style={circle,fill=blue!20,draw,minimum size=1cm,inner sep=0pt},
            }
\usepackage[utf8]{inputenc}
\usepackage{soul}
\usepackage{comment}

\usepackage{braket}
\usepackage[hyphens]{url}
\usepackage[colorlinks = true, breaklinks = true]{hyperref}
\usepackage{cleveref}
\usepackage[numbers]{natbib}

\usepackage{complexity}
\newcommand{\parityL}{\parity\mathsf{L}}
 
\usepackage{epstopdf}
\theoremstyle{plain}
\newtheorem{theorem}{Theorem}
\newtheorem{claim}[theorem]{Claim}
\newtheorem{lemma}[theorem]{Lemma}

\newtheorem{corollary}[theorem]{Corollary}

\newtheorem{fact}[theorem]{Fact}

\newtheorem*{gssproblem}{Graph state simulation problem}

\newtheorem*{PCTNSproblem}{Planar Clifford tensor network simulation problem}
\theoremstyle{definition}
\newtheorem{definition}[theorem]{Definition}

\usepackage{graphicx}
\graphicspath{{figures/}}

\usepackage{color}

\usepackage{xcolor}
\definecolor{darkred}  {rgb}{0.5,0,0}
\definecolor{darkblue} {rgb}{0,0,0.5}
\definecolor{darkgreen}{rgb}{0,0.5,0}

\hypersetup{
  pdftitle = {Fast simulation of planar Clifford circuits},
}

\usepackage[noend]{algpseudocode}
\usepackage{algorithm}

\newcommand*\Let[2]{\State #1 $\gets$ #2}
\renewcommand{\algorithmicrequire}{\textbf{Input:}}

\usepackage{braket}

\usepackage{enumitem}
\setlist[enumerate]{topsep=1ex,itemsep=0ex,partopsep=1.5ex,parsep=1ex}
\setlist[itemize]{topsep=1ex,itemsep=0ex,partopsep=1.5ex,parsep=1ex}
\setlist[description]{topsep=1ex,itemsep=0pt,partopsep=1.5ex,parsep=1ex}

\usepackage{silence}
\usepackage{etoolbox}

\usepackage{tikz}
\usetikzlibrary{arrows, shapes}
\usetikzlibrary{quantikz}
\definecolor{introduce}{rgb}{.3,.6,.3}
\definecolor{forget}{rgb}{1,0,0}
\definecolor{merge}{rgb}{0,0,1}

\usepackage[mathscr]{euscript}
\newcommand{\circuit}{\mathscr{C}}

\newcommand{\vqwblue}[1]{
	\edef\start{\the\pgfmatrixcurrentrow-\the\pgfmatrixcurrentcolumn}
	\edef\end{\the\numexpr#1+\pgfmatrixcurrentrow\relax-\the\pgfmatrixcurrentcolumn}
	\expandafter\expandafter\expandafter\vqwexplicitblue\expandafter\expandafter\expandafter{\expandafter\start\expandafter}\expandafter{\end}
}
\newcommand{\vqwexplicitblue}[2]{
	\arrow[from=#1,to=#2,arrows,draw=blue] {}
}
\newcommand{\vqwred}[1]{
	\edef\start{\the\pgfmatrixcurrentrow-\the\pgfmatrixcurrentcolumn}
	\edef\end{\the\numexpr#1+\pgfmatrixcurrentrow\relax-\the\pgfmatrixcurrentcolumn}
	\expandafter\expandafter\expandafter\vqwexplicitred\expandafter\expandafter\expandafter{\expandafter\start\expandafter}\expandafter{\end}
}
\newcommand{\vqwexplicitred}[2]{
	\arrow[from=#1,to=#2,arrows,draw=red] {}
}

\usepackage{xspace}
\usepackage{relsize}

\newcommand{\lows}{\operatorname{lower}}
\newcommand{\pfont}[1]{\textnormal{\texttt{#1}}}
\newcommand{\diag}{\operatorname{diag}}

\newcommand{\CZ}{\operatorname{CZ}}

\title{Fast simulation of planar Clifford circuits}
\author{David Gosset}
\affiliation{Institute for Quantum Computing, University of Waterloo, Canada}
\affiliation{Department of Combinatorics and Optimization, University of Waterloo, Canada}
\affiliation{Perimeter Institute for Theoretical Physics, Waterloo, Canada}
\email{dngosset@gmail.com}
\author{Daniel Grier}
\affiliation{Institute for Quantum Computing, University of Waterloo, Canada}
\affiliation{Cheriton School of Computer Science, University of Waterloo, Canada}
\affiliation{Department of Computer Science and Engineering and Department of Mathematics, University of California, San Diego, US}
\email{dgrier@ucsd.edu}
\author{Alex Kerzner}
\affiliation{Institute for Quantum Computing, University of Waterloo, Canada}
\affiliation{Department of Combinatorics and Optimization, University of Waterloo, Canada}
\email{alexkerzner1@gmail.com}
\author{Luke Schaeffer}
\affiliation{Institute for Quantum Computing, University of Waterloo, Canada}
\affiliation{Department of Combinatorics and Optimization, University of Waterloo, Canada}
\affiliation{Joint Center for Quantum Information and Computer Science, College Park, Maryland, US}
\email{lrschaeffer@gmail.com}
\date{}
\begin{document}
\maketitle
\begin{abstract}
A general quantum circuit can be simulated classically in exponential time. If it has a planar layout, then a tensor-network contraction algorithm due to Markov and Shi has a runtime exponential in the square root of its size, or more generally exponential in the treewidth of the underlying graph. Separately, Gottesman and Knill showed that if all gates are restricted to be Clifford, then there is a polynomial time simulation. We combine these two ideas and show that treewidth and planarity can be exploited to improve Clifford circuit simulation. Our main result is a classical algorithm with runtime scaling asymptotically as $ n^{\omega/2}<n^{1.19}$ which samples from the output distribution obtained by measuring all $n$ qubits of a \textit{planar} graph state in given Pauli bases. Here $\omega$ is the matrix multiplication exponent. We also provide a classical algorithm with the same asymptotic runtime which samples from the output distribution of any constant-depth Clifford circuit in a planar geometry. Our work improves known classical algorithms with cubic runtime.

A key ingredient is a mapping which, given a tree decomposition of some graph $G$, produces a Clifford circuit with a structure that mirrors the tree decomposition and which emulates measurement of the corresponding graph state. We provide a classical simulation of this circuit with the runtime stated above for planar graphs and otherwise $nt^{\omega-1}$ where $t$ is the width of the tree decomposition. Our algorithm incorporates two subroutines which may be of independent interest. The first is a matrix-multiplication-time version of the Gottesman-Knill simulation of multi-qubit measurement on stabilizer states. The second is a new classical algorithm for solving symmetric linear systems over $\mathbb{F}_2$ in a planar geometry, extending previous works which only applied to non-singular linear systems in the analogous setting.
\end{abstract}

\section{Introduction and Overview}

Enormous resources are now being marshalled by quantum information scientists towards the goal of building a quantum computer: demonstration of a computation for which there are no known efficient classical algorithms \cite{arute2019quantum}; quantum machines that are accessible to the public \cite{IBMQ};  and validation of certain building blocks of quantum error correction \cite{reed2012realization, corcoles2015demonstration, ofek2016extending}.  A by-product of the progress towards building quantum machines has been a renewed interest in the classical simulation of quantum computers \cite{markov2008simulating, boixo2017simulation, bravyi2019simulation}. This is partly due to the fact that a direct comparison with classical simulation algorithms is an important measure of the performance of today's quantum computers \cite{pednault2017breaking, pednault2019leveraging, barak2020spoofing}.  A quantum computer in the real world is limited to a specific architecture, which may restrict both the native gate set and the number of qubits while also being subjected to noise from its environment. While in some cases it is possible to leverage quantum advantage from small quantum devices with limited capabilities,  one can also exploit these limitations in classical simulation algorithms. As a result, the search for a quantum advantage with near-term quantum computers is grounded in our understanding of the power of restricted models of quantum computing. Such models include constant-depth quantum circuits \cite{terhal2002adaptive}, IQP circuits \cite{bremner2011classical}, BosonSampling \cite{aaronson2011computational}, circuits with few non-Clifford gates \cite{gottesman1998heisenberg, bravyi2016improved, bravyi2019simulation}, and others. Their study reveals bright spots of quantum advantage even in the presence of harsh restrictions on the operation of quantum devices.

Here we focus on the restricted family of Clifford circuits, which are expressed as a product of gates from the set Hadamard, $S=\mathrm{diag}(1,i)$, and $\CZ=\mathrm{diag}(1,1,1,-1)$.  Gottesman and Knill showed that Clifford circuits can be simulated in polynomial time on a classical computer \cite{gottesman1998heisenberg}. The quantum state at any point during such a Clifford circuit is a so-called \textit{stabilizer state} that can be specified using a tableau of $O(n^2)$ classical bits, see Section \ref{sec:stab} for details. Updating this representation after the action of a $H,S,$ or $\mathrm{CZ}$ gate can be performed by modifying $O(n)$ bits of the tableau in a simple way. Measurements are more costly: the best known algorithm for simulating the outcome of a single-qubit measurement uses $O(n^2)$ time \cite{aaronson+gottesman:2004}. 

Despite the fact that they can be simulated in polynomial time, it has been shown that quantum computations based on Clifford circuits still enjoy a certain type of quantum advantage \cite{bravyi2018quantum, watts2019exponential, grier2020interactive, bravyi2020quantum}. In this paper we consider a problem which is closely connected to the tasks considered in Refs.\ \cite{bravyi2018quantum, watts2019exponential, grier2020interactive}, which demonstrate a quantum advantage for certain special cases. Those tasks can be viewed  as special cases of the problem of simulating measurements on a \textit{graph state}. For any graph $G=(V,E)$ the graph state $|G\rangle$ is defined as the $n:=|V|$ qubit state given by
\[
|G\rangle=\bigg(\prod_{\{u,v\}\in E} \CZ_{uv}\bigg)H^{\otimes n}|0^n\rangle.
\]
\noindent Equivalently, $\ket{G}$ is the $n$-qubit state satisfying
\[
\pfont{X}_v \bigg(\prod_{j: \{v,j\}\in E} \pfont{Z}_j \bigg)|G\rangle=|G\rangle \qquad \qquad \forall v\in V,
\]
where $\pfont{X}_v$ and $\pfont{Z}_v$ denote the Pauli $\pfont{X}$ and $\pfont{Z}$ operators acting nontrivially on qubit $v$.

In this paper, we will focus on the computational problem of simulating measurement of graph states, which (as described later) can be used as the basis for simulating more general types of Clifford circuits. In the simplest version of the problem, we are given a graph $G=(V,E)$ with $n=|V|$ vertices and a single-qubit Pauli measurement basis $P_v\in \{\pfont X, \pfont Y, \pfont Z\}$ for each vertex $v\in V$, and we are asked to sample from the probability distribution obtained by measuring each qubit of $|G\rangle$ in the given bases. In particular, letting $U_v$ be a single-qubit Clifford unitary that maps the $P_v$ basis to the computational (\pfont Z) basis (i.e., $U_v P_v U_v^\dagger= \pfont Z$ for each $v\in V$), we are asked to sample a string $z\in \{0,1\}^n$ from the probability distribution 
\begin{equation}
p(z)=|\langle z|\bigotimes_{v\in V} U_v |G\rangle|^2.
\label{eq:pz}
\end{equation}
Notice that because $\ket{G}$ is prepared by a Clifford circuit and each $U_v$ is also Clifford, this problem is an instance of Clifford circuit simulation. Conversely, it is also well known that simulation of general Clifford circuits can be reduced to this problem using, e.g., techniques from measurement-based quantum computing \cite{raussendorf2001one}, where the size of the graph state will increase with the depth of the circuit.  We refer the reader to \Cref{sec:appl} for an example of this reduction, as well as a non-exhaustive list of applications of the problem of measuring graph states in Pauli bases.

In the general version of the problem, we also allow some of the qubits to be postselected. That is, in addition to the graph and measurement bases, we are given a subset of qubits $\mathcal{P} \subseteq V$ and measurement outcomes $m \in \{ 0, 1 \}^{\mathcal{P}}$ as input. The problem is then to sample the remaining qubits $\mathcal{S} := V \backslash \mathcal{P}$ after having postselected on measurement outcomes $m$ for the qubits of $\mathcal{P}$. In particular, let 
\begin{equation}
p_\mathcal{P}(m)=\langle G| \bigotimes_{v\in V} U^{\dagger}_v |m\rangle\langle m|_\mathcal{P} \bigotimes_{v\in V} U_v |G\rangle
\label{eq:pbm}
\end{equation}
be the probability of obtaining outcomes $m$ for the qubits in $\mathcal{P}$. If $p_\mathcal{P}(m)$ is nonzero, then the output is a binary string $x\in \{0,1\}^{|\mathcal{S}|}$ sampled from the conditional distribution $p_\mathcal{S}(x|m)$ over measurement outcomes in $\mathcal{S}$ after postselecting on outcomes $m$ for the qubits in $\mathcal{P}$:
\begin{equation}
p_\mathcal{S}(x|m)=\frac{1}{p_\mathcal{P}(m)}|\langle x|_\mathcal{S}\langle m|_\mathcal{P}\bigotimes_{v\in V} U_v |G\rangle|^2.
\label{eq:condpa}
\end{equation}
Otherwise, the output is a flag indicating that $p_\mathcal{P}(m)=0$, since it is meaningless to postselect on an impossible outcome. To summarize, we now state the problem in full.

\begin{gssproblem}
The input to the problem is a graph $G=(V,E)$ with $n=|V|$ vertices, a measurement basis $P_v\in \{\pfont X,\pfont Y, \pfont Z\}$ for each vertex $v\in V$, a partition $[n]=\mathcal{S}\cup \mathcal{P}$ of the vertices, and a binary string $m\in \{0,1\}^{|\mathcal{P}|}$. If the marginal probability $p_\mathcal{P}(m)$ from Eq.~\eqref{eq:pbm} is zero, the output is an error flag. Otherwise, the output is a binary string $x\in \{0,1\}^{|\mathcal{S}|}$ sampled from the conditional distribution $p_\mathcal{S}(x|m)$ defined in Eq.~\eqref{eq:condpa}.
\end{gssproblem}

We shall refer to the special case $\mathcal{P}=\varnothing$ as the graph state simulation problem \textit{without postselection}.  We note that although this special case without postselection is of particular interest, the classical algorithms we develop in this paper are capable of solving the more general graph state simulation problem with the same asymptotic runtime.

Let us now compare quantum and classical algorithms for the graph state simulation problem.  When there is no postselection, a quantum computer can solve the problem by preparing the graph state $|G\rangle$ and then measuring all qubits in the given bases. The total runtime of this procedure is linear in the size of the graph (number of edges plus number of vertices) and is therefore upper bounded as $O(n^2)$ in the general case, and as $O(n)$ if $G$ is sparse. The graph state simulation problem can also be solved efficiently on a classical computer via the Gottesman-Knill theorem since all unitaries and measurements to be simulated are Clifford operations. To our knowledge the best previously-known algorithm for the graph state simulation problem is based on the standard Gottesman-Aaronson simulation method \cite{aaronson+gottesman:2004} and has runtime $O(n^3)$. Our first contribution is to reduce the graph state simulation problem to matrix multiplication, improving the runtime in the general case to $O(n^{\omega})$ where  $\omega < 2.373$ \cite{alman+williams:2020} is the matrix multiplication exponent. 
\begin{theorem}
\label{thm:fastmm}
There is a classical algorithm with runtime $O(n^{\omega})$ which solves the graph state simulation problem.\footnote{As with many problems solved with fast matrix multiplication, we incur an extra $O(\log n)$ factor if $\omega = 2$.  For full details, see the proof of \Cref{thm:fastmm} in \Cref{sec:fast_measurement}.}
\end{theorem}

The proof of Theorem \ref{thm:fastmm} is obtained from a general algorithm for simulating multi-qubit measurements on stabilizer states which may be of independent interest, see Section \ref{sec:matrix_mult} for details. 

Guan and Regan \cite{guan2019stabilizer} have recently established a similar result for the problem of computing output probabilities; they show that given a binary string $z\in \{0,1\}^n$, there is a classical algorithm which computes the output probability $p(z)$ in Eq.~\eqref{eq:pz} in $O(n^{\omega})$ time. However, the result of Guan and Regan does not provide a classical algorithm for the graph state simulation problem considered here. An important distinction is that the latter problem provides a compelling candidate for a quantum speedup. Indeed, Theorem \ref{thm:fastmm} leaves open the possibility that quantum computers may have a polynomial speedup in terms of the total runtime (number of one- and two-qubit gates) required to solve the graph state simulation problem without postselection. For sparse graphs the comparison is more dramatic---the linear time quantum algorithm \textit{which bears no resemblance to classical linear algebraic methods} beats Theorem \ref{thm:fastmm} by a quadratic factor even if the matrix multiplication exponent is shown to meet its lower bound $\omega\geq 2$. Unfortunately we lack a suitable technique for proving that classical algorithms are more costly in terms of total runtime. On the other hand, if we turn our attention to other metrics such as circuit depth, Refs.\ \cite{bravyi2018quantum, grier2020interactive} prove that quantum computers beat classical machines in solving the graph state simulation problem even for the special case of a (subgraph of a) 2D grid graph.  These works establish a quantum advantage in the following sense:

\begin{itemize}[leftmargin=0pt]
\item[]\emph{Quantum advantage for relation tasks (\cite{bravyi2018quantum, watts2019exponential})}\hspace{1ex}

Let $G$ be a subgraph of the $\ell \times \ell$ grid, and consider a relaxation of the graph state simulation problem (without postselection) in which it suffices to output \emph{any} string of measurement outcomes $z\in \{0,1\}^n$ that occurs with nonzero probability. (Equivalently, the problem can be expressed in terms of properties of binary quadratic forms, without reference to quantum graph states, see Ref.\cite{bravyi2018quantum} for details). This problem can be trivially solved by a constant-depth quantum circuit since since the graph has constant degree.  In contrast, it cannot be solved by any constant-depth classical circuit composed of unbounded fan-in AND, OR, and NOT gates (i.e., $\AC^0$ circuits) \cite{bravyi2018quantum, watts2019exponential}. 

\item[]\emph{Quantum advantage for interactive tasks (\cite{grier2020interactive})}\hspace{1ex}
Simulation of low-depth quantum circuits can be shown to be even harder for classical devices when they are required to answer the graph state simulation problem in two rounds.  Again, here we consider some subgraph $G=(V,E)$ of the $\ell \times \ell$ grid.  The first round input is a set of Pauli measurement bases for each of the qubits in some given subset $\mathcal{P}$ of the qubits, and the desired output is any measurement outcome $m\in \{0,1\}^{|\mathcal{P}|}$ such that $p_\mathcal{P}(m)\neq 0$. The second round input is a set of measurement bases for the remaining qubits in $\mathcal{S}=V\setminus \mathcal{P}$, and the desired outcome is any string $x$ such that $p_\mathcal{S}(x|m)\neq 0$. Of course, this problem can be solved by a constant-depth quantum computation which prepares the graph state and in each round measures the corresponding subset of qubits in the given Pauli bases. Nevertheless, any classical algorithm which succeeds at this interactive task could be used to solve any problem in the class $\parityL$. This provides a strong lower bound on the type of classical circuits that can solve the interactive task, as the class $\parityL$ is known to be strictly more powerful than the constant-depth circuit family $\AC_0$.\footnote{We note that the postselected graph state simulation problem was already known to be $\parityL$-hard \cite{aaronson+gottesman:2004}. While the interactive task described in Ref. \cite{grier2020interactive}  has a similar hardness guarantee, it can be solved by a constant-depth quantum computation.}

\end{itemize}

Motivated by the above results, we first observe that Theorem \ref{thm:fastmm} can be improved for the special case of 2D grid graphs.  

\begin{theorem}[Special case of \Cref{thm:planargss}]
\label{thm:grid_warmup}
There is an $\widetilde{O}(n^{\omega/2})$-time classical algorithm which solves the graph state simulation problem for any subgraph of a $\sqrt{n}\times \sqrt{n}$ grid.
\end{theorem}

The seemingly innocuous quadratic improvement for the graph state simulation problem on grid graphs provides a striking conclusion about the problems used for quantum advantage with low-depth circuits \cite{bravyi2018quantum, watts2019exponential, grier2020interactive}.  Although quantum circuits based on measuring grid graph states are superior in terms of depth, \Cref{thm:grid_warmup} leaves open the possibility that there exist classical circuits with a similar number of elementary gates (e.g., if $\omega = 2$, then a classical circuit of $\widetilde{O}(n)$ gates is possible, matching the quantum circuit with $O(n)$ gates). This leaves their advantage in terms of gate complexity on shakier footing.

Although \Cref{thm:grid_warmup} follows directly from \Cref{thm:planargss} (discussed below), we now present a simpler, self-contained algorithm for grid graphs which serves as a warm-up for the more general case.  Consider graph state simulation on the $\ell \times \ell$ grid. For simplicity we consider the case without postselection ($\mathcal{P}=\varnothing$) although the algorithms we consider below work in the more general case with only superficial modifications. We shall also ignore fast matrix multiplication for the moment. Our starting point is the very simple classical simulation of the process of preparing the graph state---by starting with the all-zeros computational basis state, applying Hadamards to each qubit, and applying $\CZ$ gates between qubits at vertices connected by an edge---and then sequentially measuring the $\ell^2$ vertices in the given Pauli bases. Simulating all $O(\ell^2)$ gates in the circuit uses a runtime of $O(\ell^4)$. Since the state has $O(\ell^2)$ qubits, simulating each single-qubit measurement also uses a runtime of $O(\ell^4)$ \cite{aaronson+gottesman:2004}, and the total runtime of this algorithm is $O(\ell^6)$.  An improvement was noticed in Ref.\ \cite{bravyi2018quantum}, and is based on simulating a quantum algorithm that ``sweeps'' from the left to the right, processing at most $2\ell$ qubits at a time.  The key idea is that the measurements on any given column commute with the CZ gates two columns away. Concretely, one can prepare all qubits in the first two columns of the grid, measure the qubits in the first column in the given Pauli basis, and then initialize the qubits in the third column and apply the CZ gates between the second and third column.  Continuing in this way, the algorithm sweeps from left to right, storing an \textit{active set} of at most $2\ell$ qubits at any given time. Since we are now measuring a stabilizer state of only $O(\ell)$ qubits, the cost of each single-qubit measurement is reduced to $O(\ell^2)$. Thus, the classical simulation algorithm runs in $O(\ell^4)$ time. Perhaps surprisingly, this is \emph{not} the ordering of CZ gates and measurements that leads to the most efficient classical algorithm.  In the box containing \Cref{fig:2d}, we describe a recursive method for ordering the measurements in which only one active set of size $O(\ell)$ dominates the runtime, leading to the bound given by Theorem \ref{thm:grid_warmup}.  We stress that these advantages are not merely theoretical.  \Cref{fig:grid_simulation} compares the runtimes of software implementations of the na\"ive, left-to-right sweep, and recursive algorithms. The results show a clear asymptotic advantage for the recursive algorithm, and an absolute advantage beyond a $50 \times 50$ grid (where the algorithm finished in under a hundredth of a second anyway). Using the recursive algorithm we were able to simulate a grid of 529 million qubits in under 11 hours on a desktop computer.

\begin{figure}
\centering
\input{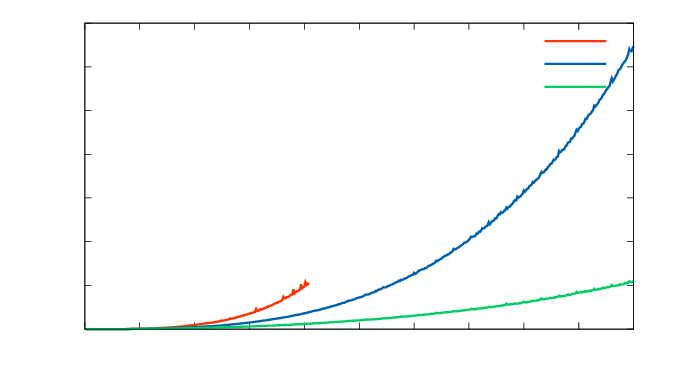}
\caption{Simulation time in seconds vs.\ grid length $\ell$ of an $\ell \times \ell$ grid. The data is taken as an average over 7 random trials from $\ell = 2$ to $\ell = 1000$ by $2$'s.  We cut off the na\"ive algorithm at $\ell=408$ due to memory limitations.  The simulations were conducted on a desktop computer using a modified version of the CHP Clifford simulator \cite{aaronson+gottesman:2004, gridCHPpp}. Our simulations do not use fast matrix multiplication. See \Cref{sec:software_implementation} for the implementation details.}
\label{fig:grid_simulation}
\end{figure}

The recursive, divide-and-conquer algorithm for the grid has its roots in the \textit{nested dissection} algorithm of George for solving positive semidefinite linear systems of equations, where the sparsity pattern of the matrix matches the adjacency of the grid graph \cite{george1973nested}. This result was later generalized to obtain fast algorithms for solving nonsingular linear systems of equations on planar graphs \cite{lipton1979generalized, alon2010solving}, and likewise our main result generalizes the graph state simulation algorithm from the grid to arbitrary planar graphs.
\begin{restatable}{theorem}{planargss}
There is an $\widetilde{O}(n^{\omega/2})$-time classical algorithm which solves the graph state simulation problem for planar graphs.
\label{thm:planargss}
\end{restatable}
Our work shares inspiration and methodology with the literature on planar linear system solving algorithms. Moreover, one of the key subroutines from our algorithm for graph state simulation provides a new method for solving linear systems in a planar geometry. In particular, unlike previous works \cite{lipton1979generalized, alon2010solving}, our algorithm can be applied to singular linear systems.

\begin{theorem}[Special case of \Cref{thm:linearsolver}]
There exists an algorithm to find a uniformly random solution to $Ay=b \pmod 2$ (or report that no solution exists) in time $\widetilde{O}(n^{\omega/2})$ where $A \in \mathbb F_2^{n \times n}$ is the adjacency matrix of a planar graph.
\label{thm:specialcase}
\end{theorem}

While the simulation algorithm from \Cref{thm:planargss} does use a divide-and-conquer strategy, we stress that it is not simply obtained using the active set approach---based on a clever ordering of measurements and CZ gates---which we used in the case of the 2D grid.\footnote{In fact, some features of our proof---such as the duplication of qubits---are mirrored in existing planar linear system solvers where the adjacency matrix is sparsified by adding additional rows and columns \cite{alon2010solving}.}

Instead, the divide-and-conquer algorithm from \Cref{thm:planargss} is based on a data structure known as the \emph{tree decomposition} of the graph.  The tree decomposition consists of nodes--each representing a subset of vertices of the original graph--which are connected together in a certain way to form a tree.  Loosely speaking, the more your original graph already resembled a tree, the smaller you can make the subsets at each node (for a formal definition of the tree decomposition, see \Cref{sec:td}).  The size of the largest subset of vertices at a node in the tree decomposition is referred to as the \emph{treewidth}.  Planar graphs are special due to the fact that they admit tree decompositions of width $O(\sqrt{n})$ that can be computed in linear time using the planar separator theorem \cite{Lipton1979}.  

More generally, we ask how the runtime of algorithms for graph state simulation scale with the treewidth of the graph. Here we take our inspiration from pioneering work of Markov and Shi who showed that a general (universal) quantum circuit can be simulated on a classical computer in time exponential in the treewidth of a graph derived in a simple way from the circuit \cite{markov2008simulating}. The simulation is based on viewing the quantum circuit as a tensor network and choosing an ordering to contract the tensor network based on a tree decomposition.  Tensor-network contraction based algorithms along these lines are especially well-suited to the simulation of near-term quantum computers with a two-dimensional qubit architecture \cite{aaronson2016complexity,boixo2017simulation,chen2018classical,villalonga2020establishing}. Of course, our setting is somewhat different as we are interested in polynomial-time algorithms and the special case of Clifford circuits. We describe a parameterized algorithm for the graph state simulation problem for low-treewidth graphs with the following runtime guarantee. 
\begin{restatable}{theorem}{treealg}
There is a classical algorithm which, given an instance of the graph state simulation problem along with a tree decomposition $T$ of $G$ of treewidth $t$ and at most $O(n)$ nodes, outputs a solution in $\widetilde{O}(nt^{\omega-1})$ time.
\label{thm:treegss}
\end{restatable}

We expect \Cref{thm:treegss} to be most useful for families of graphs where a tree decomposition can be computed efficiently. Computing the treewidth of a graph is an \NP-hard problem \cite{arnborg1987complexity}, and all known classical algorithms for finding a tree decomposition whose width is within a constant factor of optimal require exponential time in the worst case \cite{bodlaender1993tourist}.\footnote{Markov and Shi's algorithm does not require a tree decomposition to be given because its runtime is exponential in the treewidth, and constructing a suitable tree decomposition from scratch incurs no additional asymptotic cost \cite{bodlaender}.  In contrast, our algorithms are polynomial in the treewidth.}

\Cref{thm:planargss} and \Cref{thm:treegss} are both obtained as special cases of a general classical simulation algorithm, which is described in detail in \Cref{sec:gss}. At the core of this algorithm is a tree-decomposition-to-circuit mapping which takes a graph $G$ along with a tree decomposition $T$ and outputs a Clifford circuit that solves the graph state simulation problem on $G$ and has a structure similar to that of $T$. This structure allows us to simulate the circuit in a lazy way, only initializing qubits and applying gates as needed. Just like in the active set approach described above, this reduces the (dominant) cost of measurements in the classical simulation algorithm, such that when we measure a given qubit, the stabilizer state we are measuring has as few qubits as possible. The circuit begins with a quantum state which is a tensor product of few-qubit graph states, one for each leaf node of the tree decomposition. Since a given vertex of the graph may appear in multiple nodes of its tree decomposition, this initial state may have several replica qubits corresponding to each data qubit of the graph state to be simulated. This replication of data qubits is an essential difference between our approach and the active set approach. We will exploit the fact that it is possible to replicate data qubits and operate on the replicas before merging them later in the simulation. During the course of the circuit, gates are applied, qubits are measured, and replica qubits are merged together as needed using a gadget composed of Clifford gates and postselection.   We exhibit a classical simulation of this Clifford circuit with a runtime which is a simple function of the tree decomposition. We show that this runtime is upper bounded by \Cref{thm:treegss} in the general case, and we also provide a sharper upper bound as stated in \Cref{thm:planargss} for planar graphs.

There are several open questions related to our work.  As far as we know, the graph state simulation problem without postselection on general bounded-degree graphs remains a compelling candidate for a polynomial quantum speedup. Can the classical algorithms for the latter problem be improved, or conversely is it possible to strengthen the evidence for a quantum speedup? Another direction for future work is to use our techniques in conjunction with stabilizer-rank based classical simulation algorithms for universal quantum circuits \cite{bravyi2016trading,bravyi2016improved, bravyi2019simulation}. Roughly speaking, such algorithms are based on representing the quantum state of a given non-Clifford circuit as a superposition of stabilizer states which can each be simulated using the Gottesman-Knill theorem. For this application one requires a stronger form of classical simulation---going beyond the standard tableau representation \cite{aaronson+gottesman:2004}---which keeps track of the global phase of the stabilizer states, since they appear in superposition. Such phase-sensitive Clifford simulators are described in Refs.~\cite{nest2008classical,bravyi2016improved,bravyi2019simulation}. Recent work \cite{kerznerthesis} shows that our results are compatible with stabilizer rank methods, however the need for phase-sensitivity precludes the use of the fast matrix multiplication methods in Section \ref{sec:matrix_mult}.

The remainder of the paper is organized as follows. In \Cref{sec:appl} we describe applications of \Cref{thm:planargss} to the simulation of constant-depth Clifford circuits on planar graphs and an extension to more general Clifford tensor networks. In \Cref{sec:prelim} we discuss subroutines for the classical simulation of stabilizer states using the standard tableau representation \cite{aaronson+gottesman:2004}, and we describe improved subroutines for multi-qubit measurements using fast matrix multiplication. This section also reviews relevant background information regarding tree decompositions and treewidth. We describe the tree-decomposition-to-circuit mapping and prove  \Cref{thm:treegss} in \Cref{sec:gss}. Finally, in \Cref{sec:plan} we give a sharper analysis of the runtime of our algorithm when applied to tree decompositions that are computed using the planar separator theorem \cite{Lipton1979} and we prove \Cref{thm:planargss}. The proof of \Cref{thm:fastmm} is provided in Section \ref{sec:matrix_mult}.

\begin{center}
\parbox{\textwidth}{
\begin{mdframed}[linewidth=0.75pt, skipabove=2cm, skipbelow=2cm, frametitle={Example: recursive algorithm for 2D grid graphs \label{ex:grid_warmup}}]

Here we describe a simple recursive classical algorithm which attains the runtime bound claimed in \Cref{thm:grid_warmup} for the graph state simulation problem on an $\ell\times \ell$ grid graph. The \textit{main subroutine} measures (or postselects) qubits in the interior of the subgrid in the given bases, and outputs the measurement outcomes along with the post-measurement state on an ``active set" consisting of the $4\ell-4$ perimeter vertices. The key insight is that the main subroutine can be performed by recursively computing the outcomes and post-measurement states for 4 copies of the  $\lfloor \ell/2\rfloor-1 \times \lfloor \ell/2\rfloor-1 $ grid which are completely contained in the interior of the  $\ell\times \ell$ grid, separately preparing the stabilizer state on the $O(\ell)$ remaining qubits (including the perimeter), applying the $O(\ell)$ $\CZ$ gates which connect them together, and then measuring all qubits except those on the perimeter. The algorithm is illustrated schematically in \Cref{fig:2d}. Letting $f(\ell)$ denote the runtime of this algorithm, we have $f(1), f(2)=O(1)$ and $f(\ell) = 4f(\ell/2) + \widetilde{O}(\ell^\omega)$ for $\ell\geq 3$. Here the first term contains the runtime of the four recursive calls, while the second term contains the runtime required for the remaining steps of the algorithm which only involve preparation and measurement of stabilizer states of $O(\ell)$ qubits. The solution to the recurrence is $f(\ell) = \widetilde{O}(\ell^\omega) = \widetilde{O}(n^{\omega/2})$. In the very last stage of the algorithm the qubits on the perimeter are measured, incurring an additional $O(\ell^\omega)$ cost.
\vspace{1ex}

\captionsetup{type=figure}
\begin{center}
\begin{tikzpicture}[scale=0.5]

  \foreach \x in {1,...,3}
{
    \foreach \y in {1,...,2}  
{
      \draw[thick] (\x,\y)--(\x,\y+1) (\y,\x)--(\y+1,\x) ;
}
}

  \foreach \x in {5,...,7}
{
    \foreach \y in {1,...,2}  
{
      \draw[thick]  (\x,\y)--(\x,\y+1) (\y,\x)--(\y+1,\x) ;
}
}

  \foreach \x in {1,...,3}
{
    \foreach \y in {5,...,6}  
{
      \draw[thick]  (\x,\y)--(\x,\y+1) (\y,\x)--(\y+1,\x) ;
}
}

  \foreach \x in {5,...,7}
{
    \foreach \y in {5,...,6}  
{
      \draw[thick]  (\x,\y)--(\x,\y+1) (\y,\x)--(\y+1,\x) ;
}
}

  \foreach \x in {0,4,8}
{
    \foreach \y in {0,...,7}  
{
      \draw[thick]  (\x,\y)--(\x,\y+1) (\y,\x)--(\y+1,\x) ;
}
}
\foreach \i in {0,...,8}
{
\foreach \j in {0,...,8}
{
	\node[circle,fill=red!55!white, inner sep=0pt, minimum size=6]
        () at (\i,\j) {};
}
}
\node[circle,fill=white!40!blue, inner sep=0pt, minimum size=6]() at (2,2){};
\node[circle,fill=white!40!blue, inner sep=0pt, minimum size=6]() at (6,2) {};
\node[circle,fill=white!40!blue, inner sep=0pt, minimum size=6]() at (6,6) {};
\node[circle,fill=white!40!blue, inner sep=0pt, minimum size=6]() at (2,6) {};
\end{tikzpicture}
\qquad
\begin{tikzpicture}[scale=0.5]

  \foreach \x in {0,...,8}
{
    \foreach \y in {0,...,7}  
{
      \draw[thick] (\x,\y)--(\x,\y+1) (\y,\x)--(\y+1,\x) ;
}
}

\foreach \i in {0,...,8}
{
\foreach \j in {0,...,8}
{
	\node[circle,fill=red!55!white, inner sep=0pt, minimum size=6]
        () at (\i,\j) {};
}
}
\node[circle,fill=white!40!blue, inner sep=0pt, minimum size=6]() at (2,2){};
\node[circle,fill=white!40!blue, inner sep=0pt, minimum size=6]() at (6,2) {};
\node[circle,fill=white!40!blue, inner sep=0pt, minimum size=6]() at (6,6) {};
\node[circle,fill=white!40!blue, inner sep=0pt, minimum size=6]() at (2,6) {};
\end{tikzpicture}
\qquad 
\begin{tikzpicture}[scale=0.5]

  \foreach \x in {0,...,8}
{
    \foreach \y in {0,...,7}  
{
      \draw[thick] (\x,\y)--(\x,\y+1) (\y,\x)--(\y+1,\x) ;
}
}

\foreach \i in {1,...,7}
{
\foreach \j in {1,...,7}
{
	\node[circle,fill=white!40!blue, inner sep=0pt, minimum size=6]
        () at (\i,\j) {};
}
}

\foreach \i in {0,...,8}
{
\foreach \j in {0,8}
{
	\node[circle,fill=red!55!white, inner sep=0pt, minimum size=6]
        () at (\i,\j) {};
}
}

\foreach \i in {0,8}
{
\foreach \j in {1,...,7}
{
	\node[circle,fill=red!55!white, inner sep=0pt, minimum size=6]
        () at (\i,\j) {};
}
}
\end{tikzpicture}

\end{center}

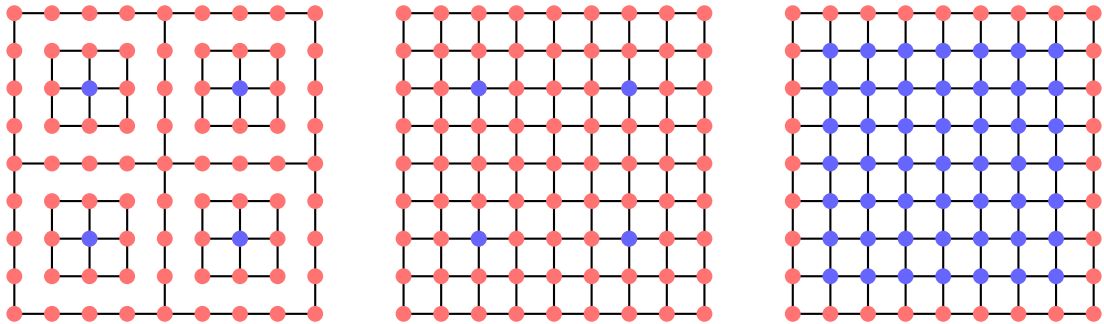
\captionof{figure}{Depiction of the recursive algorithm for graph state simulation on a 2D grid. Blue vertices are measured and red are unmeasured. (a) Quantum graph states corresponding to the 5 connected components of this graph are prepared and the qubits at the blue vertices are measured. (b) $\CZ$ gates are applied to connect them together. (c) Single-qubit measurements are performed on all qubits corresponding to red vertices not on the perimeter.} 
\label{fig:2d}
\end{mdframed}
}
\end{center}

\subsection{Applications \label{sec:appl}}

In this section, we describe how our algorithm for the graph state simulation problem can be leveraged to solve other more traditional circuit simulation tasks (proofs and more precise statements in \Cref{app:applications}).  In the most general setting, we prove an analog of the Markov and Shi tensor network contraction algorithm for Clifford tensor networks with planar geometry:

\begin{theorem}[informal]
\label{thm:clifford_tensor}
A random non-zero element from a planar Clifford tensor network can be sampled in $\widetilde{O}(n^{\omega/2})$ time.
\end{theorem}

Any $n$-qubit quantum state $\ket{\psi}$ can be represented by a multidimensional matrix (or \emph{tensor}) indexed by $\{0,1\}^n$, whose $(x_1,\ldots,x_n)$th entry is $\braket{\psi|x_1 \ldots x_n}$. Such tensors are sometimes written as a \emph{tensor network}, where individual tensors are combined together by a generalization of matrix multiplication to describe the global tensor. Sampling from a tensor network refers to sampling from the state that the tensor network represents. A \emph{Clifford tensor network} (see \Cref{app:applications} for a formal definition) is a tensor network whose individual tensors all represent stabilizer states. The tensors making up a tensor network often have operational meanings such as applying a gate or projecting onto a subspace. These operational meanings illustrate the expressiveness of Clifford tensor networks, which capture unitary operations (such as composition of Clifford gates) and non-unitary operations (such as projection onto stabilizer states). In fact, our theorem for Clifford tensor networks immediately gives an algorithm for the simulation of planar Clifford circuits:

\begin{theorem}[informal]
\label{thm:clifsim}
Let $C$ be an $n$-qubit, depth-$d$ Clifford circuit with qubits located at the vertices of a planar graph and two-qubit gates acting on edges of the graph. A random $n$-bit measurement outcome of the state $C \ket{0^n}$ can be sampled in time $\widetilde{O}(n^{\omega/2} d^\omega)$.
\end{theorem}

While this result can be obtained as a corollary to \Cref{thm:clifford_tensor}, in \Cref{sec:planar_circuit_application} we provide a proof which directly reduces the Clifford circuit simulation task to the graph state simulation problem \textit{without postselection}. Although our algorithms for graph state simulation with and without postselection have the same asymptotic running time, their practical running times are likely to be quite different due to the increased complexity associated with postselection. Indeed, in \Cref{sec:gss} we will provide a self-contained algorithm for graph state simulation without postselection which avoids many of the linear algebra subroutines used in the more general case. This allows one to simulate constant-depth Clifford circuits in a planar geometry in a much more straightforward manner.

We note that the two theorems described above are not an exhaustive list of applications of graph state simulation.  There are many possible variations of Clifford simulation problems one might want to solve that would follow from tweaking the proofs of the theorems above.  For example, \Cref{thm:clifsim} applies when the qubits of a Clifford circuit are arranged on the vertices of a planar graph, but one could get a similar result for circuits where the \emph{gates} are placed at the vertices of a planar graph with edges between gates if the output of one gate is fed into another gate (e.g., a Clifford circuit with gates acting only on neighboring qubits).\footnote{In more detail, the underlying graph is as follows:  each gate/wire of the circuit is a vertex; and there is an edge between two vertices when the wire is an input/output of the gate.}  Furthermore, planar graphs are not the only graphs which admit efficiently computable tree decompositions.  For example, the graph state sampling problem for the $n^{\frac{1}{3}} \times n^{\frac{1}{3}} \times n^{\frac{1}{3}}$ cube can be solved in $\widetilde O(n^{\frac{2}{3}\omega})$ time\footnote{In fact, such an algorithm follows from the same sort of recursive algorithm used in the proof of \Cref{thm:grid_warmup}: split the cube into 8 identical subcubes, recursively measure their interiors, and then measure the boundary.  This leads to the recursion $f(n) = 8f(n/8) + \widetilde{O}(n^{\frac{2}{3}\omega})$.  Once again, the computation at the root of the recursion dominates, so we get $f(n) = \widetilde{O}(n^{\frac{2}{3}\omega})$ runtime.} using a tree decomposition for the cube of width $O(n^{2/3})$, and so we could obtain similar applications for Clifford circuits where the gates/qubits are arranged in a 3D grid.

Finally, we note that the graph state simulation algorithm for planar graphs still suffices when the graph is close to planar in some sense.  To this end, we introduce the notion of coarse-graining a graph (see \Cref{fig:coarse_graining}).  We say that a graph $G'$ can be coarse-grained to another graph $G$ if there exists a map $\varphi \colon V(G')\rightarrow V(G)$ such that for any edge $\{u,v\}\in E(G')$ we have either (a) $\{\varphi(u),\varphi(v)\}\in E(G)$, or (b) $\varphi(u)=\varphi(v)$.\footnote{In other words, a graph homomorphism modulo loops, since graph states are typically defined only on loopless graphs.}
We say that $\varphi$ is an \emph{$r$-coarse-graining} from $G'$ to $G$ if each vertex of $G$ is the image of at most $r$ vertices of $G'$.

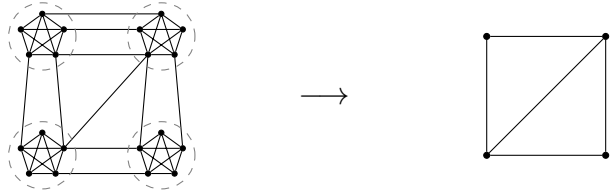
\begin{figure}[h!]
\centering
\begin{subfigure}[r]{.3\textwidth}
\centering
\begin{tikzpicture}
\pgfmathsetmacro{\dotradius}{.3}
\pgfmathsetmacro{\circleradius}{9}
\pgfmathsetmacro{\height}{45}
\pgfmathsetmacro{\width}{45}
\tikzstyle{dot}=[circle, fill=black, draw, inner sep=.7pt]
\tikzstyle{graycircle}=[draw, dashed, circle, color=gray, inner sep=\circleradius]

    \node[dot] at ({\dotradius * cos((pi/2 + 0*2*pi/5) r)}, {\dotradius * sin((pi/2 + 0*2*pi/5) r)}) (a1) {};
    \node[dot] at ({\dotradius * cos((pi/2 + 1*2*pi/5) r)}, {\dotradius * sin((pi/2 + 1*2*pi/5) r)}) (a2) {};
    \node[dot] at ({\dotradius * cos((pi/2 + 2*2*pi/5) r)}, {\dotradius * sin((pi/2 + 2*2*pi/5) r)}) (a3) {};
    \node[dot] at ({\dotradius * cos((pi/2 + 3*2*pi/5) r)}, {\dotradius * sin((pi/2 + 3*2*pi/5) r)}) (a4) {};
    \node[dot] at ({\dotradius * cos((pi/2 + 4*2*pi/5) r)}, {\dotradius * sin((pi/2 + 4*2*pi/5) r)}) (a5) {};
    
    \node[dot, yshift=\height] at ({\dotradius * cos((pi/2 + 0*2*pi/5) r)}, {\dotradius * sin((pi/2 + 0*2*pi/5) r)}) (b1) {};
    \node[dot, yshift=\height] at ({\dotradius * cos((pi/2 + 1*2*pi/5) r)}, {\dotradius * sin((pi/2 + 1*2*pi/5) r)}) (b2) {};
    \node[dot, yshift=\height] at ({\dotradius * cos((pi/2 + 2*2*pi/5) r)}, {\dotradius * sin((pi/2 + 2*2*pi/5) r)}) (b3) {};
    \node[dot, yshift=\height] at ({\dotradius * cos((pi/2 + 3*2*pi/5) r)}, {\dotradius * sin((pi/2 + 3*2*pi/5) r)}) (b4) {};
    \node[dot, yshift=\height] at ({\dotradius * cos((pi/2 + 4*2*pi/5) r)}, {\dotradius * sin((pi/2 + 4*2*pi/5) r)}) (b5) {};

    \node[dot, xshift=\width] at ({\dotradius * cos((pi/2 + 0*2*pi/5) r)}, {\dotradius * sin((pi/2 + 0*2*pi/5) r)}) (c1) {};
    \node[dot, xshift=\width] at ({\dotradius * cos((pi/2 + 1*2*pi/5) r)}, {\dotradius * sin((pi/2 + 1*2*pi/5) r)}) (c2) {};
    \node[dot, xshift=\width] at ({\dotradius * cos((pi/2 + 2*2*pi/5) r)}, {\dotradius * sin((pi/2 + 2*2*pi/5) r)}) (c3) {};
    \node[dot, xshift=\width] at ({\dotradius * cos((pi/2 + 3*2*pi/5) r)}, {\dotradius * sin((pi/2 + 3*2*pi/5) r)}) (c4) {};
    \node[dot, xshift=\width] at ({\dotradius * cos((pi/2 + 4*2*pi/5) r)}, {\dotradius * sin((pi/2 + 4*2*pi/5) r)}) (c5) {};
    
    \node[dot, xshift=\width, yshift=\height] at ({\dotradius * cos((pi/2 + 0*2*pi/5) r)}, {\dotradius * sin((pi/2 + 0*2*pi/5) r)}) (d1) {};
    \node[dot, xshift=\width, yshift=\height] at ({\dotradius * cos((pi/2 + 1*2*pi/5) r)}, {\dotradius * sin((pi/2 + 1*2*pi/5) r)}) (d2) {};
    \node[dot, xshift=\width, yshift=\height] at ({\dotradius * cos((pi/2 + 2*2*pi/5) r)}, {\dotradius * sin((pi/2 + 2*2*pi/5) r)}) (d3) {};
    \node[dot, xshift=\width, yshift=\height] at ({\dotradius * cos((pi/2 + 3*2*pi/5) r)}, {\dotradius * sin((pi/2 + 3*2*pi/5) r)}) (d4) {};
    \node[dot, xshift=\width, yshift=\height] at ({\dotradius * cos((pi/2 + 4*2*pi/5) r)}, {\dotradius * sin((pi/2 + 4*2*pi/5) r)}) (d5) {};

    \draw (a1)--(a2)--(a3)--(a4)--(a5)--(a1)--(a3)--(a5)--(a2)--(a4)--(a1);
    \draw (b1)--(b2)--(b3)--(b4)--(b5)--(b1)--(b3)--(b5)--(b2)--(b4)--(b1);
    \draw (c1)--(c2)--(c3)--(c4)--(c5)--(c1)--(c3)--(c5)--(c2)--(c4)--(c1);
    \draw (d1)--(d2)--(d3)--(d4)--(d5)--(d1)--(d3)--(d5)--(d2)--(d4)--(d1);
    
    \draw (a5)--(b4);
    \draw (a2)--(b3);
    \draw (a5)--(d3);
    \draw (a5)--(c2);
    \draw (a4)--(c3);
    \draw (c2)--(d3);
    \draw (c5)--(d4);
    \draw (b1)--(d1);
    \draw (b5)--(d2);
    \draw (b4)--(d3);
    
    \node[graycircle] {};
    \node[graycircle, yshift=\height] {};
    \node[graycircle, xshift=\width] {};
    \node[graycircle, xshift=\width, yshift=\height] {};
    
\end{tikzpicture}
\end{subfigure}
$\longrightarrow$
\begin{subfigure}[c]{.3\textwidth}
\centering
\begin{tikzpicture}
\pgfmathsetmacro{\height}{45}
\pgfmathsetmacro{\width}{45}
\tikzstyle{dot}=[circle, fill=black, draw, inner sep=.8pt]

\node[dot] (a) {};
\node[dot, yshift=\height] (b) {};
\node[dot, xshift=\width] (c) {};
\node[dot, xshift=\width, yshift=\height] (d) {};

\draw (a)--(c)--(d)--(b)--(a)--(d);
    
\end{tikzpicture}
\end{subfigure}
\caption{A 5-coarse-graining from a nonplanar graph to a planar graph.}
\label{fig:coarse_graining}
\end{figure}

The following theorem, which we prove in \Cref{sec:plan}, shows that the graph state sampling algorithm can be extended to graphs which can be coarse-grained to planar.

\begin{restatable}[Extension of Theorem \ref{thm:planargss}]{theorem}{coarse}
Let $G'$ be a graph with $n$ vertices and let $G$ be a planar graph. Suppose we are given an $r$-coarse-graining from $G'$ to $G$. There is a classical algorithm with runtime upper bounded as $\widetilde{O}(n^{\omega/2} r^\omega)$ which solves the graph state simulation problem for $G'$.
\label{thm:coarse}
\end{restatable}

\section{Preliminaries}
\label{sec:prelim}

In this section we define stabilizer states and tree decompositions and describe some of their properties. In Section \ref{sec:stab} we review the tableau representation of stabilizer states and describe the runtime of classical algorithms for updating this representation under Clifford operations. In Section \ref{sec:td} we review the definition of a tree decomposition of a graph, and we describe simple algorithms which convert a given tree decomposition into one with certain desirable features.

\subsection{Stabilizer states \label{sec:stab}}

For $a \in \{0,1\}^n$, let $\pfont X^a$ be the $n$-qubit Pauli operator $\pfont{X}^{a_1} \otimes \cdots \otimes \pfont X^{a_n}$.  Define $\pfont Z^a$ similarly.  A general $n$-qubit Pauli operator $P$ is of the form $i^\alpha (-1)^{\beta} \pfont X^{a}\pfont Z^b$ for some $n$-bit strings $a,b$ and bits $\alpha,\beta\in \{0,1\}$ describing the global phase.

Recall that an $n$-qubit stabilizer state $\psi$ can be specified (up to a global phase) by a set of independent, commuting $n$-qubit Pauli operators $\{P_1,P_2,\ldots, P_n\}$ such that $P_j|\psi\rangle=|\psi\rangle$. Here independence means that none of the operators $P_k$ is proportional to a product of the other ones. Equivalently, an $n$-qubit stabilizer state is of the form $|\psi\rangle=Q|0^n\rangle$ where $Q$ is an $n$-qubit Clifford unitary, i.e., a unitary which can be expressed as a product of single-qubit Hadamard, phase $S=\mathrm{diag}(1,i)$, and two-qubit $\CZ=\mathrm{diag}(1,1,1,-1)$ gates.

The Gottesman-Knill theorem establishes that stabilizer states can be stored and manipulated efficiently on a classical computer \cite{gottesman1998heisenberg}. Throughout this work, we use a classical representation of stabilizer states and Clifford operations which is a variant of the tableau representation described by Aaronson and Gottesman \cite{aaronson+gottesman:2004}. In particular, we represent both the $n$-qubit stabilizer state $|\psi\rangle=Q|0^n\rangle$ (up to a global phase) and the Clifford unitary $Q$ as a tableau that lists the images of $\pfont X^{e_j}$ and $\pfont Z^{e_j}$ under conjugation by $Q$ for all $j$, where $e_j \in \{0,1\}^n$ is the $n$-bit vector with a 1 in position $j$ and 0's everywhere else. Let 
\[
Q \pfont X^{e_j} Q^\dag = i^{\alpha_j} (-1)^{\beta_j} \pfont X^{a_j} \pfont Z^{b_j} \qquad \text{and} \qquad Q \pfont Z^{e_j} Q^\dag = i^{\delta_j} (-1)^{\epsilon_j} \pfont X^{c_j} \pfont Z^{d_j}.
\]
We arrange this conjugation information into a $2n \times 2n$ matrix $M$, a phase vector $p \in \{0,1\}^{2n}$, and a sign vector $s \in \{0,1\}^{2n}$:
\begin{equation}
M = \left[\begin{array}{c | c}
A & B \\ \hline C & D
\end{array}\right]
\hspace{10pt}
p = \left[\begin{array}{c} \alpha \\ \hline \delta \end{array}\right]
\hspace{10pt}
s = \left[\begin{array}{c} \beta \\ \hline \epsilon \end{array}\right]
\label{eq:tableau}
\end{equation}
where $A, B, C,$ and $D$ are the $n \times n$ binary matrices whose elements correspond to the Pauli decompositions. We will refer to this as the \emph{tableau representation} of the state/circuit.

It is well known that the tableau representation can be updated efficiently if we apply a Clifford gate or measure a qubit of the stabilizer state. For example, one can easily see that the runtime of updating the tableau under the application of a $1$- or $2$-qubit Clifford gate is $O(n)$. Aaronson and Gottesman \cite{aaronson+gottesman:2004} provide an $O(n^2)$ algorithm for measuring a qubit in the computational basis, improving upon a na\"{i}ve $O(n^3)$ algorithm. In this work we are interested in measuring batches of qubits all at once. An application of Aaronson-Gottesman would upper bound the runtime of measuring $k$ qubits in the computational basis as $O(n^2 k)$. We now describe how this can be improved.

The Aaronson-Gottesman simulation algorithm can be viewed as an application of linear algebra over $\mathbb F_2$. In fact, Aaronson and Gottesman show that a variant of the Clifford circuit simulation problem\footnote{Specifically, they consider the task of determining whether measuring the first qubit of a given Clifford circuit's output state yields result $\ket{1}$ with probability 1.} is $\parityL$-complete, putting it in the same class as matrix inversion, iterated matrix multiplication, determinant, and other linear algebra tasks over $\mathbb F_2$ \cite{damm90}. Since the discovery of fast matrix multiplication in the 1960s, many linear algebra tasks have been sped up via a reduction to matrix multiplication. The (asymptotically) fastest currently known algorithm for matrix multiplication has a runtime $O(n^{\omega})$ for $\omega < 2.3728596$ \cite{alman+williams:2020}. It seems quite natural to conjecture that simulating a depth-$n$ Clifford circuit on $n$ qubits, and then measuring all $n$ qubits could be accomplished in $O(n^{\omega})$ time. Yet surprisingly, we find no proof of this fact in the literature. To this end, in \Cref{sec:matrix_mult} we will give algorithms for fast stabilizer circuit composition (\Cref{thm:fast_composition}) and state measurement (\Cref{thm:fast_measurement}) using their tableau representations.  These operations are described in more detail in \Cref{table:clifford_asymptotics}, where we summarize all the subroutines for manipulating stabilizer states which are used in this work.

\begin{table}[ht]
\begin{framed}
\noindent \large{\textbf{Classical representation of stabilizer states}}
\normalsize

\noindent An $n$-qubit stabilizer state $|\psi\rangle$ is represented up to a global phase as a tableau Eq.~\eqref{eq:tableau} of $4n(n+1)$ bits. This representation can be manipulated as follows:
\begin{itemize}
    \item[] \textbf{Append a qubit:} $|\psi\rangle|0\rangle$ can be computed in time $O(n)$.
    \item[] \textbf{Apply a gate:} Any $1$- or $2$-qubit Clifford gate can be applied to $|\psi\rangle$ in time $O(n)$.
    \item[] \textbf{Append:} If $|\phi\rangle$ is an $\ell$-qubit stabilizer state, compute $|\psi\rangle|\phi\rangle$ in time $O((n+ \ell)^2)$.
    \item[] \textbf{Measure:} Measure a given subset of $k$ qubits (in the $\pfont Z$ basis) in time $\widetilde O(k^{\omega-2}n^{2})$.
	\item[] \textbf{Apply CZ gates:} Apply an arbitrary collection of CZ gates in $O(n^\omega)$ time. 
\end{itemize}
\end{framed}
\caption{Running times for basic Clifford operations. These follow from Aaronson and Gottesman \cite{aaronson+gottesman:2004} for na\"{i}ve matrix multiplication (i.e., for $\omega = 3$). In \Cref{sec:matrix_mult}, we show how to achieve these running times for $2 \leq \omega < 3$ assuming fast matrix multiplication in time $O(n^{\omega})$. See Theorems \ref{thm:fast_composition} and \ref{thm:fast_measurement} specifically. }
\label{table:clifford_asymptotics}
\end{table}

We will also need the following simple fact about stabilizer states:

\begin{claim}
Suppose $\psi$ is an $n$-qubit stabilizer state and $y \in \{0,1\}^n$ satisfies $\langle y|\psi\rangle \neq 0$. If $R$ is a uniformly random stabilizer of $\psi$, then the bit string $z \in \{0,1\}^n$ such that $|z\rangle \propto R|y\rangle$ is drawn from the distribution $p(z)=|\langle z|\psi\rangle|^2$.
\label{claim:simplestab}
\end{claim}

\begin{proof}
Let $\mathrm{Stab}(\psi)$ denote the stabilizer group of $\psi$. The probability distribution of the random variable $z$ is, by definition,
\begin{equation}
p(z)= \frac{|\Omega|}{|\mathrm{Stab}(\psi)|}=\frac{|\Omega|}{2^n} \qquad \quad \Omega=\left\{P\in \mathrm{Stab}(\psi): \langle z|P|y\rangle \neq 0\right\}.
\label{eq:pofy}
\end{equation}
Now suppose $P\in \Omega$, so $\langle y|P= \alpha_P \langle z|$ for some $\alpha_P \in \{\pm 1, \pm i\}$. Then $\langle y|P|\psi\rangle=\langle y|\psi\rangle=\alpha_P \langle z|\psi\rangle$. This shows that $\alpha_P$ does not in fact depend on $P$ and in particular 
\begin{equation}
\alpha_P=\alpha\equiv \frac{\langle y|\psi\rangle}{\langle z|\psi\rangle}.
\label{eq:alpha}
\end{equation}
Therefore
\[
\langle y|\psi\rangle\langle \psi|z\rangle=\frac{1}{2^n} \sum_{Q\in \mathrm{Stab}(\psi)} \langle y|Q|z\rangle=\frac{1}{2^n} \sum_{P\in \Omega} \langle y|P|z\rangle=\frac{\alpha}{2^n} |\Omega|.
\]
Using Eq.~\eqref{eq:alpha} gives $|\Omega|=2^n |\langle z|\psi\rangle|^2$ and combining with Eq.~\eqref{eq:pofy} completes the proof.
\end{proof}
Finally, we note that \Cref{claim:simplestab} has the following direct corollary:
\begin{claim}
Suppose $\psi$ is an $n$-qubit stabilizer state and $x,y \in \{0,1\}^n$ satisfy $\langle x|\psi\rangle \neq 0$ and $\langle y|\psi\rangle \neq 0$. Then there exists a stabilizer $P$ of $\psi$ such that $P|x\rangle \propto |y\rangle$.
\label{claim:existp}
\end{claim}

\subsection{Tree Decompositions \label{sec:td}}

One of our goals in this work is to give a parameterized algorithm for the graph state simulation problem, where the parameter is \emph{treewidth}. Let us introduce tree decompositions and treewidth. 
\begin{definition}\label{def:tree_decomposition}
Let $G = (V,E)$ be a graph. A \emph{tree decomposition} of $G$ is a tree $T$ such that each of its nodes is associated with a \emph{bag} of vertices $B_i \subseteq V$ such that the following hold:
\begin{enumerate}
    \item For all $v \in V$, there exists $B_i$ such that $v \in B_i$.  

    \item For all $\{u,v\} \in E$, there exists $B_i$ such that $u,v \in B_i$.  

    \item For all $v \in V$, the nodes $\{B_i : v \in B_i\}$ form a connected subtree of $T$.
\end{enumerate}
It is conventional to use the word ``node'' to describe vertices of the tree, to avoid confusion with vertices of the graph $G$. Similarly, we will use the term ``link'' to describe edges of the tree, to avoid confusion with the edges of $G$.
\end{definition}
\noindent The \emph{width} of a tree decomposition is defined as $t = \max_i |B_i| - 1$, the size of the largest bag minus one.  The \emph{treewidth} of a graph is the smallest $t$ for which there exists a tree decomposition of width $t$. \Cref{fig:graph_and_td} shows an example of a tree decomposition for a graph of treewidth 2.

Although the treewidth is a natural parameter for many graph algorithms, the performance of our algorithms are more precisely captured by the sum of the bag sizes raised to some power.  We introduce the notation 
$$
\| T \|_{p} := \left( \sum_i |B_i|^{p} \right)^{1/p}
$$
for the $p$-norm of the vector of bag sizes in a tree decomposition $T$ for real $p > 0$. We also define the $\infty$-norm
$$
\| T \|_{\infty} := \max_{i} |B_i|,
$$
and note that $t = \| T \|_{\infty} - 1$.

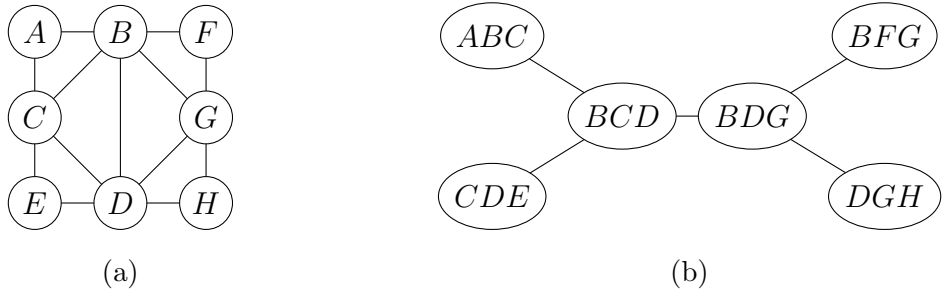
\begin{figure}
\centering
\begin{subfigure}[t]{0.45\textwidth}
\centering
\begin{tikzpicture}
\tikzstyle{vertex}=[draw, circle, minimum size=20pt, inner sep=0pt, outer sep=0pt]
\matrix[row sep=12pt, column sep=12pt] {
\node (A) [vertex] {$A$}; & \node (B) [vertex] {$B$}; & \node (F) [vertex] {$F$}; \\
\node (C) [vertex] {$C$}; & & \node (G) [vertex] {$G$}; \\
\node (E) [vertex] {$E$}; & \node (D) [vertex] {$D$}; & \node (H) [vertex] {$H$}; \\
};
\path[draw,-] 
(A)--(B)--(F)--(G)--(H)--(D)--(E)--(C)--(A)
(B)--(C)--(D)--(G)--(B)--(D);
\end{tikzpicture}
\subcaption{}
\end{subfigure}
\begin{subfigure}[t]{0.45\textwidth}
\centering
\begin{tikzpicture}
\tikzstyle{node}=[draw, ellipse, minimum width=30pt, minimum height=25pt, inner sep=0pt, outer sep=0pt]
\matrix[row sep=5pt, column sep=8pt] {
\node (ABC) [node] {$A B C$}; & & &  \node (BFG) [node] {$B F G$}; \\
& \node (BCD) [node] {$B C D$}; & \node (BDG) [node] {$B D G$}; \\
\node (CDE) [node] {$C D E$}; & & & \node (DGH) [node] {$D G H$}; \\
};
\path[draw,-] 
(ABC)--(BCD)--(BDG)--(BFG)
(CDE)--(BCD)
(BDG)--(DGH);
\end{tikzpicture}
\subcaption{}
\end{subfigure}
\caption{(a) Planar graph on 8 vertices.  (b) A tree decomposition of width 2 \cite{treed}.}
\label{fig:graph_and_td}
\end{figure}

We will often assume that the tree decomposition has a particularly nice form.  We caution that while the definition below is convenient for our purposes, it may diverge from other works.

\begin{definition}
A \emph{nice} tree decomposition \cite{bodlaender1991better} is a rooted tree decomposition where every node is either an \emph{introduce} node, a \emph{forget} node, or a \emph{merge} node: 
\begin{itemize}
    \item[] \textit{Introduce:} Nodes whose bag strictly contains its child's bag (or that have no children). These nodes are either leaves, or have one child.
    \item[] \textit{Forget:} Nodes that have exactly one child and whose bag is strictly contained in their child's bag. 
    \item[] \textit{Merge:} Nodes that have exactly two children and whose bag is the union of its children's bags. 
\end{itemize}

\end{definition}

The following theorem will allow us to work with nice tree decompositions containing relatively few nodes whenever it is convenient.
\begin{theorem}
\label{thm:to_nice}
Let $G=(V,E)$ be a graph with $n=|V|$ vertices, and let $T$ be a tree decomposition of $G$ of width $t$. We can compute a nice tree decomposition $T'$ of $G$ with $O(n/t)$ nodes and width $O(t)$ in $O(\|T\|_1)$ time.
\end{theorem}

We will prove this theorem from a series of lemmas. For a node $B$ with parent $A$ in a rooted tree decomposition, we say that the vertices $B\setminus A$ are \textit{forgotten} on the link from $A$ to $B$. If there is a link between nodes $A$ and $B$, \textit{contracting} the two nodes refers to the process of contracting the link between $A$ and $B$ and setting the combined node equal to $A \cup B$. If nodes $B$ and $S$ are siblings,\footnote{We say that two nodes are siblings if they share the same parent.} then \textit{merging} them refers to the following process: delete the nodes $B$ and $S$; create a new node $B \cup S$ with a link to parent shared by $B$ and $S$, as well as links to any children of $B$ and $S$. Note that the contraction and merge operations preserve the requirements for a tree decomposition set out in \Cref{def:tree_decomposition}.

\begin{lemma}
		\label{lem:forgot}
		Given a rooted tree decomposition of a graph, every vertex of the graph is either forgotten exactly once or belongs to the root node. 
\end{lemma}
\begin{proof}
    Consider an arbitrary vertex $v$ that is not in the root node. Suppose for the sake of contradiction that $v$ is forgotten along two distinct links (i.e. $v$ is in two distinct bags $A$ and $B$, but not in their respective parent(s)). However, $A$ and $B$ are connected by a unique path which goes through one or both of their parents. However, by \Cref{def:tree_decomposition}, the nodes containing $v$ are connected, contradicting our assumption that neither parent contains $v$. 
\end{proof}

\begin{lemma}
    \label{lem:forget_properties}
    Consider a rooted tree decomposition. Let $B$ be a node with child $C$, parent $A$, sibling $S$, and grandparent $P$.\footnote{We say that node $P$ is a grandparent of $B$ if $P$ is the parent of $B$'s parent.} Then:
    \begin{enumerate}
        \item If we contract $A$ and $B$ into a node $A \cup B$ then
        \begin{align*}
            C \backslash (A \cup B) &= C \backslash B, 
            & S \backslash (A \cup B)&= S \backslash A, 
            & (A \cup B) \backslash P &= (B \backslash A) \cup (A \backslash P).
        \end{align*}
        \item If we merge $B$ and $S$ into a new child $B \cup S$ of $A$ then 
        \begin{align*}
            C \backslash (B \cup S) &= C \backslash B, 
            & (B \cup S) \backslash A &= (B \backslash A) \cup (S \backslash A).
        \end{align*}
    \end{enumerate}  
    In other words, contract and merge operations (i) preserve the nodes forgotten by their children or siblings; and (ii) forget the (disjoint) union of the vertices forgotten by their constituent parts. 
\end{lemma}
\begin{proof}
		The proof consists of repeated application of the requirement that vertices appear in connected subsets of a tree decomposition. For the first four equalities, we will show that the two sets are equal by showing that they contain one another.
		
		Consider the parent-child contraction $A, B \mapsto A \cup B$ first. Clearly $C \backslash (A \cup B)\subseteq C \backslash B$. If there was some vertex $v \in C \backslash B$ but $v \notin C \backslash (A \cup B)$, then clearly $v \in C$ and $v \notin B$, but $v \in A \cup B$ and hence $v \in A$. This contradicts the fact that bags containing $v$ are connected. Therefore we must also have $C \backslash (A \cup B)\supseteq C \backslash B$, so we conclude $C \backslash (A \cup B) = C \backslash B$. Similarly, $S \backslash (A \cup B)$ is contained in $S \backslash A$, and if the containment was strict we would have $v \in B$ and $v \in S$, but $v \notin A$, a contradiction. 
		
		Next, we have $(A \cup B) \backslash P = (B \backslash P) \cup (A \backslash P) \supseteq (B \backslash A) \cup (A \backslash P)$, since any $v \in B \backslash A$ is also in $B \backslash P$. (Otherwise, $v \in B$ and $v \in P$ implies $v \in A$ and hence $v \notin B \backslash A$, a contradiction.) In the other direction, take an arbitrary $v \in (B \backslash P) \cup (A \backslash P)$. Clearly $v \notin P$, but it must belong to either $A$ or $B$. If $v \in A$ then $v \in (A \backslash P) \subseteq (B \backslash A) \cup (A \backslash P)$. Otherwise, $v \notin A$ and $v \in B$, so again $v \in (B \backslash A) \cup (A \backslash P)$. Thus, the two sets are equal. 
		
		Now consider the sibling merge $B, S \mapsto B \cup S$. Once again, clearly $C \backslash (B \cup S) \subseteq C \backslash B$. If $v \in C \backslash B$ then $v \in C$ and $v \notin B$. The path connecting $C$ and $S$ must pass through $B$, so using the fact that bags containing $v$ form a connected subtree, we conclude that $v \notin S$, and thus $v \in C \backslash (B \cup S)$. The final equality follows from the distributive law:
		\[
		(B \cup S) \backslash A = (B \cup S) \cap \overline{A} = (B \cap \overline{A}) \cup (S \cap \overline{A}) = (B \backslash A) \cup (S \backslash A).
		\] 
	\end{proof}

\begin{lemma}\label{lem:td1}
		Given a graph $G$, a rooted tree decomposition $T$ for $G$ of treewidth $t$, and an integer $c \geq 0$, we can construct a new rooted tree decomposition with treewidth at most $t+2c$ and at most $\tfrac{n}{c+1}+1$ nodes in $O(\|T\|_1)$ time.
		In particular for $c = \Theta(t)$ the tree decomposition has width $O(t)$ and $O(n/t)$ nodes. 
	\end{lemma}

\begin{proof}
        The algorithm proceeds recursively, moving from the leaf nodes up to the root node. For a bag $B$, we recursively process its children. Then we repeatedly merge pairs of children whenever there are two that both forget $\leq c$ vertices along the link to $B$. (This may result in more than two children being merged into the same node. For example, if $B$ has three children that each forget $c/3$ vertices along the link to $B$, we would merge the first two, and then merge the new node with the third.) Once this process is complete, there might still be one such child. If this is the case, we contract it into $B$.

        First, we prove that the resulting tree decomposition has at most $\tfrac{n}{c+1}+1$ nodes. We do so by claiming that when this process is complete, all nodes (except the root) forget at least $c+1$ vertices. Suppose to the contrary that some node $B$ forgets $\leq c$ vertices. By Lemma~\ref{lem:forget_properties}, the merges and contractions above $B$ in the tree do not change the set of vertices $B$ forgets, so $B$ forgot $\leq c$ vertices even at the recursive call on its parent. But then the algorithm would have merged it with a sibling, or if no sibling had $\leq c$ vertices, contracted it with its parent, which contradicts the fact that $B$ exists in $T$. Thus, all links forget at least $c+1$ vertices. By \Cref{lem:forgot}, every vertex is forgotten (somewhere in the tree) at most once, so there are at most $\frac{n}{c+1}$ links, and thus $\frac{n}{c+1} + 1$ nodes. 
		
		Next, we show that the resulting treewidth is at most $t + 2c$. Merging and contraction usually results in larger nodes, potentially beyond the original upper bound of $t+1$ vertices. However, when we merge two siblings forgetting $\leq c$ vertices each, the resulting node forgets $\leq 2c$ vertices (by \Cref{lem:forget_properties}), and is thus at most $2c$ vertices larger than its parent. The parent is (at this point in the algorithm) untouched from the original tree, so it has at most $t+1$ vertices, and thus the merged node has $\leq t+1+2c$ vertices. Likewise, contraction only occurs when there is a child forgetting $\leq c$ vertices, so we add $\leq c$ vertices to an untouched parent with $\leq t+1$ vertices for at total size of $\leq t+1+c$. The size of a bag never exceeds $t + 2c + 1$ during transformation, and so the resulting treewidth is at most $t + 2c$.

        Finally, we claim that the runtime can be made to be $O(\|T\|_1)$. The main hurdle is that if we merge and contract bags by explicitly computing the union of the two bags every time, then our runtime might not necessarily be $O(\|T\|_1)$. Instead, we will wait until the very end before explicitly computing the contents of each bag. Observe that each bag in the new tree decomposition is a union of bags from $T$, since that's all that the merge and contraction operations do to the contents of bags. Throughout the algorithm, for each bag we will store a set of pointers to the bags of $T$ that that bag is the union of, and wait until the very end to compute all unions. In order to do so, we need to show that the number of vertices forgotten along a link can be computed using this pointer-based approach (rather than explicitly knowing the contents of the child's bag). Recall that by \Cref{lem:forget_properties}, when we merge two nodes $B$ and $S$ with common parent $A$, the number of vertices $|(B \cup S)\setminus A|$ forgotten along the new link to $A$ is given by $|B\setminus A| + |S\setminus A|$. \Cref{lem:forget_properties} also tells us that when contracting $B$ and $A$, the number of vertices $|(A\cup B)\setminus P|$ forgotten along the new link to $A$'s parent $P$ is given by $|B\setminus A|+ |A\setminus P|$. These simple rules allow us to keep a running tally of how many vertices are forgotten along each link, without having to explicitly compute bag contents every time. Using these rules, whenever we merge or contract two nodes, we can update the number of vertices forgotten along the link from the resulting node to its parent in time $O(1)$.
        
        Before we begin processing $T$, we compute the number of vertices forgotten along each link. For a bag $B$ with children $C_1,\ldots,C_k$, the number of vertices forgotten along each link to $B$ can be computed in time $|B| + |C_1| + \cdots + |C_k|$ by enumerating vertices of $C_1, \ldots, C_k$ and checking each against $B$. Repeating this for each bag takes time $O(\|T\|_1)$, since each bag is used at most once as a parent and once as a child. Then, while performing the transformation on the tree decomposition, whenever we merge or contract two nodes, we can update the number of vertices forgotten along the link from the resulting node to its parent using the rules discussed above in time $O(1)$. Next, observe that each bag from $T$ is represented in exactly one bag of the new tree decomposition. Therefore, if for each bag in the new tree decomposition, we have pointers to the bags of $T$ that the new bag is the union of, then we can explicitly compute the contents of all bags in a total time of $O(\|T\|_1)$, since each bag of $T$ will be processed by the ``union'' operation exactly once. This shows that the total runtime is $O(\|T\|_1)$.
	\end{proof}

\begin{lemma} \label{lem:td2}
Given a width-$t$ tree decomposition $T$ with $m$ nodes, we can construct a width-$t$ nice tree decomposition $T'$ with at most $\frac{7}{2}m$ nodes in $O(mt)$ time.
\end{lemma}
\begin{proof}
We shall describe the transformation in two steps. 

In the first step, for each node with bag $B$ and children with bags $C_1, \ldots, C_k$, add nodes $B \cap (C_1 \cup \cdots\cup C_k)$ and $B \cap C_1, \ldots, B \cap C_k$ as shown in \Cref{fig:to_nice}. Contract links where the parent and child have the same bag. 

It is clear that for each link in the tree, the bags of the endpoints are comparable with either $\subseteq$ or $\supseteq$. Classify nodes with zero or one children as introduce nodes or forget nodes by the direction of this containment.

The second step of the transformation is as follows. For any node with more than two children, note that it is the union of its children. It will eventually become a merge node, but first we need to reduce the number of children. To do this, we pick two children arbitrarily, give them a new merge node as a parent, and make it a child of the original parent. Of course, the new merge node bag will be the union of its children's bags (as is required for a merge node). Repeat until all nodes have at most two children. 

\Cref{fig:to_nice} depicts an example of this transformation. Note that the transformation only creates nodes that are intersections of others, so the maximum bag size cannot increase. When we create merge nodes from two children, we take the union of two nodes, but both are already contained in their parent, so the new bag cannot be bigger than its parent. Thus, treewidth does not increase. 

Observe that since $T$ is a tree with $m$ nodes, it has $m-1$ links. The first step of the transformation creates at most one node per link, and an additional node per multi-child internal node. There are at most $m-1$ links, and at most $\frac{m-1}{2}$ multi-child internal nodes (since each one is associated with at least two links), so we add at most $\frac{3}{2}m$ nodes in this step. For a bag $B$ with children $C_1, \ldots, C_k$, this step takes time $O\left(|B| + \sum_{i=1}^k |C_i|\right)$, which leads to a total runtime of $O(||T||_1)$ for this first step, which is in $O(mt)$.

If a node of the tree has $d > 2$ children, then in the second step of the transformation we create new merge nodes which blow it up into a binary tree with $d$ leaves. The leaves and root of the binary tree already exist, but we create $d-2$ additional nodes. The total number of children across all nodes is $m-1$, so we add at most $m-3$ additional nodes. Therefore in total, the process introduces at most $\frac{5}{2}m$ new nodes. Because each new node is obtained by taking unions and intersections of at most three nodes (whose sizes are at most $t+1$), the time to compute each new node is $O(t)$. Combining this with the fact that this transformation only creates $O(m)$ new nodes, we find that the runtime of this second step is $O(mt)$.
\end{proof}

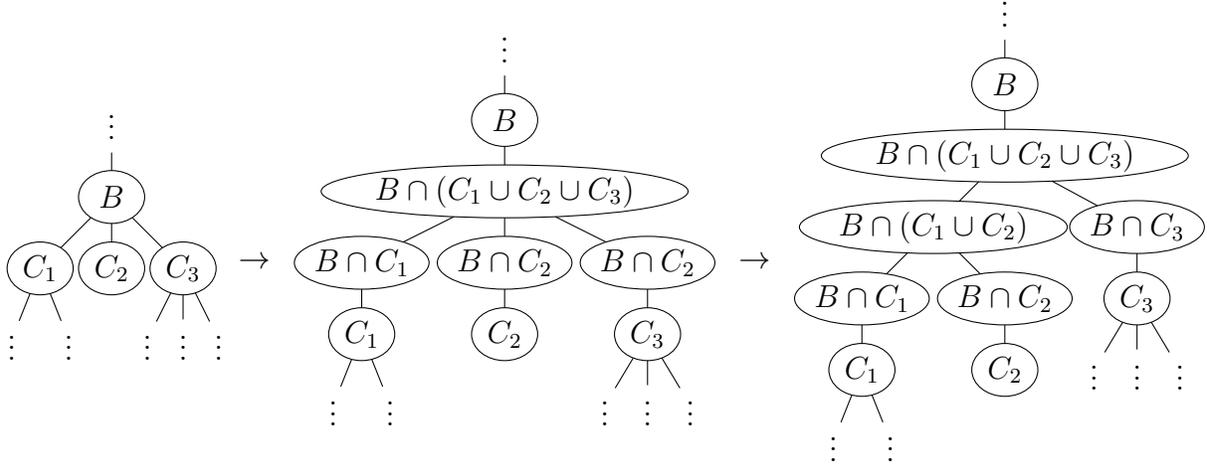
\begin{figure}
\tikzset{
  treenode/.style = {align=center, inner sep=0pt, text centered, font=\sffamily},
  bag/.style = {draw, ellipse, minimum width=25pt, minimum height=20pt, inner sep=0pt, outer sep=0pt},
  etc/.style = {circle, inner sep=0pt, outer sep=0pt},
}
\centering

\begin{tikzpicture}[-,>=stealth',level/.style={sibling distance = 1cm,level distance = 1cm}, scale=.95]
]
\node[etc, yshift=-1.5cm] (T1) {\raisebox{5pt}{$\vdots$}}
	child{ node[bag] {$B$}
    	child{ node[bag] {$C_1$} 
    		child[sibling distance=8mm]{ node[etc] {$\vdots$} }
    		child[sibling distance=8mm]{ node[etc] {$\vdots$} }
    	}
    	child{ node[bag] {$C_2$} }
    	child{ node[bag] {$C_3$}
        		child[sibling distance=5mm]{ node[etc] {$\vdots$} }
        		child[sibling distance=5mm]{ node[etc] {$\vdots$} }
        		child[sibling distance=5mm]{ node[etc] {$\vdots$} }
	}
}
; 

\node at (2cm, -3.5cm) {$\rightarrow$};

\node[etc] at (5.5cm, -.5cm) (T2) {\raisebox{5pt}{$\vdots$}}
	child{ node[bag] {$B$}
	child { node[bag] {$B \cap (C_1 \cup C_2 \cup C_3)$}
		child[sibling distance=2cm]{ node[bag] {$B \cap C_1$}
    			child{ node[bag] {$C_1$} 
    				child[sibling distance=8mm]{ node[etc] {$\vdots$} }
    				child[sibling distance=8mm]{ node[etc] {$\vdots$} }
    			}
		}
		child[sibling distance=2cm]{ node[bag] {$B \cap C_2$}
    			child{ node[bag] {$C_2$} 
			}
		}
		child[sibling distance=2cm]{ node[bag] {$B \cap C_3$}
    			child{ node[bag] {$C_3$}
        				child[sibling distance=6mm]{ node[etc] {$\vdots$} }
        				child[sibling distance=6mm]{ node[etc] {$\vdots$} }
        				child[sibling distance=6mm]{ node[etc] {$\vdots$} }
			}
		}
	}
}
;
\node at (9cm, -3.5cm) {$\rightarrow$};

\node[etc] at (12.5cm, 0) (T3) {\raisebox{5pt}{$\vdots$}}
child{ node[bag] {$B$}
    child{ node[bag] {$B \cap (C_1 \cup C_2 \cup C_3)$}
        child[sibling distance=2cm]{ node[bag] {$B \cap (C_1 \cup C_2)$} 
            child[sibling distance=2cm]{ node[bag] {$B \cap C_1$}
                child{ node[bag] {$C_1$}
                    child[sibling distance=8mm]{ node[etc] {$\vdots$} }
                    child[sibling distance=8mm]{ node[etc] {$\vdots$} }
                }
            }
            child[sibling distance=2cm]{ node[bag] {$B \cap C_2$}
                child { node[bag] {$C_2$} }
            }                            
        }
        child[sibling distance=3.7cm]{ node[bag] {$B \cap C_3$}
            child { node[bag] {$C_3$} 
                child[sibling distance=6mm]{ node[etc] {$\vdots$} }
                child[sibling distance=6mm]{ node[etc] {$\vdots$} }
                child[sibling distance=6mm]{ node[etc] {$\vdots$} }
            }
    	}
    }
    }
; 
\end{tikzpicture}
\caption{Depiction of the transformation used in \Cref{lem:td2} which maps a tree decomposition to a nice tree decomposition.}
\label{fig:to_nice}
\end{figure}

We now prove \Cref{thm:to_nice}.

\begin{proof}[Proof of \Cref{thm:to_nice}]
    Set $c = \Theta(t)$ in \Cref{lem:td1} to contract/merge down to a tree decomposition $T'$ with $O(n/t)$ nodes and width $O(t)$ in $O(\|T\|_1)$ time. Then use \Cref{lem:td2} on $T'$ to get a \emph{nice} tree decomposition $T''$ in $O(\tfrac{n}{t} \cdot t) = O(n)$ time, still with $O(n/t)$ nodes and width $O(t)$. Clearly $n = O(\|T\|_1)$ since $T$ contains every vertex, so the result follows. 
\end{proof}

Suppose that $\varphi \colon V(G) \rightarrow V(H)$ is a coarse-graining from a graph $G$ to another graph $H$. The following Lemma allows us to compute a tree decomposition of $G$, given one of $H$. In particular, suppose $T_H$ is a tree decomposition of $H$. Define a tree decomposition $T_G := \varphi^{-1}(T_H)$ to have the same tree as $T_H$ but each bag $B$ in $T_H$ is replaced with its preimage  
$$
\varphi^{-1}(B) := \{ v \in V : \varphi(v) \in B \}. 
$$
We now prove that $T_G$ is indeed a tree decomposition of $G$.
\begin{lemma}[Tree Decompositions are closed under coarse-graining]
\label{lem:tree_decomp_homo}
Let $\varphi \colon G \to H$ be a coarse-graining from a graph $G = (V,E)$ to graph $H$. Given a tree decomposition $T_H$ of the graph $H$, the preimage $T_G := \varphi^{-1}(T_H)$ is a tree decomposition for $T_G$. 
\end{lemma}
\begin{proof}
We check the defining properties of a tree decomposition explicitly:
\begin{enumerate}
    \item For any vertex $v \in V$, its image $\varphi(v)$ appears in some bag $B$ of $T_H$, and therefore $v$ appears in $\varphi^{-1}(B)$ in $T_G$.
    \item For each edge $\{ u, v \} \in E$, either $\varphi(u) = \varphi(v)$ appears in some bag $B$ and therefore $\{ u, v \}$ appears in $\varphi^{-1}(B)$, or $\varphi(u) \neq \varphi(v)$ but $\{ \varphi(u), \varphi(v) \}$ is an edge appearing in some bag $B$ of $T_H$, and thus $\{ u, v \}$ appears in $\varphi^{-1}(B)$.
    \item Take an arbitrary vertex $v \in V$. It appears in nodes of $T_G$ corresponding to nodes of $T_H$ that contain $\varphi(v)$. Since this is a connected subtree of $T_H$ (by the tree decomposition property of $T_H$), it is also connected in $T_G$. 
\end{enumerate}
\end{proof}

\section{Algorithm for graph state simulation}
\label{sec:gss}

In this section we describe our algorithm for the graph state simulation problem which makes use of a tree decomposition of the graph. We first recall the definition of the problem:

\begin{gssproblem}
The input to the problem is a graph $G=(V,E)$ with $n=|V|$ vertices, a measurement basis $P_v\in \{\pfont X,\pfont Y, \pfont Z\}$ for each vertex $v\in V$, a partition $[n]=\mathcal{S}\cup \mathcal{P}$ of the vertices, and a binary string $m\in \{0,1\}^{|\mathcal{P}|}$. If the marginal probability $p_\mathcal{P}(m)$ from Eq.~\eqref{eq:pbm} is zero, the output is an error flag. Otherwise, the output is a binary string $x\in \{0,1\}^{|\mathcal{S}|}$ from the conditional distribution $p_\mathcal{S}(x|m)$ defined in Eq.~\eqref{eq:condpa}.
\end{gssproblem}

Our goal in this section is to prove the following theorem:

\begin{theorem}
There is a classical algorithm which takes as input an instance of the graph state simulation problem along with a nice tree decomposition $T$ of the given graph $G$. The algorithm outputs a solution in time $\widetilde{O}(\| T \|_\omega^\omega)$.
\label{thm:sumofcubes}
\end{theorem}

\Cref{thm:treegss} from the introduction (restated below) follows immediately.
\treealg*
\begin{proof}
\Cref{thm:to_nice} allows us to map the given tree decomposition $T$ to a nice tree decomposition $T'$ with $O(n/t)$ nodes and treewidth $O(t)$ in time $O(nt)$. Clearly $\| T' \|_{\omega}^{\omega} = \sum_{i} |B_i|^{\omega} = O(n t^{\omega-1})$ since there are at most $O(n/t)$ bags in $T'$ of size at most $O(t)$. Applying Theorem~\ref{thm:sumofcubes} to $T'$ and $G$ finishes the proof. 
\end{proof}

Later, in \Cref{sec:plan} we will show how the graph state sampling algorithm for planar graphs (\Cref{thm:planargss}) also follows from \Cref{thm:sumofcubes}. Thus, \Cref{thm:sumofcubes} can be viewed as our main result in this work.

\paragraph{Algorithm overview}
Let us first describe the high level structure of the algorithm described in Theorem \ref{thm:sumofcubes}.  It is based on a quantum circuit which is constructed from the given tree decomposition. The circuit, which we denote by $\circuit$, is obtained by replacing each node of $T$ with a gadget based on the node type, i.e., introduce, forget, or merge. Since the circuit has $n=|V|$ data qubits labelled by vertices of the graph, and one qubit for each of the $n_a$ merge gadgets (we refer to these as \emph{merge ancillas}), $\circuit$ acts on a total of $n_t := n + n_a$ qubits. All qubits are initially prepared in the state $|+\rangle\equiv 2^{-1/2}(|0\rangle+|1\rangle)$. We will see that, if we apply the circuit and then project all merge ancillas into the all-zeros state, the resulting state of the $n$ data qubits is the graph state to be simulated:
\begin{equation}
(I\otimes \langle 0^{n_a}|) \circuit \left(|+^n\rangle\otimes |+^{n_a}\rangle\right) \propto |G\rangle.
\label{eq:postsel}
\end{equation}
Here the right-hand side is of course independent of the tree decomposition. To solve the graph state simulation problem it therefore suffices to classically simulate the following procedure: apply $\circuit$, project all merge ancillas onto the all-zeros state, project all data qubits in the set $\mathcal{P}$ according to the given measurement bases/outcomes, and measure all data qubits in the set $\mathcal{S}$ in the given Pauli bases. Our algorithm is a slight variant of this and is based on (a classical simulation of) of two steps: a sampling subroutine, followed by a correction subroutine.

In the first step, which we call the \textit{sampling subroutine}, we apply $\circuit$, measure all merge ancillas in the computational basis, and measure all data qubits in the given Pauli bases $\{P_v\}_{v\in V}$. This produces a bit string $y\in \{0,1\}^{n_t}$ drawn from a distribution
\begin{equation}
p(y)=|\langle y|(U_{\mathrm{bases}}\otimes I)\circuit|+^{n_t}\rangle|^2.
\label{eq:py}
\end{equation}
Here we write $U_{\mathrm{bases}}$ for the $n$-qubit Clifford unitary which locally changes the basis of each data qubit to the given Pauli bases $\{P_v\}$,  such that 
\[
P_v U_{\mathrm{bases}}|z\rangle=(-1)^{z_v} U_{\mathrm{bases}}|z\rangle \qquad \qquad z\in \{0,1\}^n.
\] 
In particular, $U_{\mathrm{bases}}$ is a tensor product of single-qubit identity gates, Hadamard gates, and gates 
\[
    \frac{1}{\sqrt{2}} \left( \begin{matrix} 1 & 1 \\ i & -i \end{matrix} \right)
\] 
which map the computational basis to itself, the $X$ basis, and the $Y$ basis respectively. A crucial feature of $\circuit$, as we will see in more detail below, is that its structure mirrors that of the tree decomposition. We show how to classically simulate the circuit in a ``lazy" way whereby gates are applied from the leaves up to the root of the tree and qubits are measured as soon as possible (after all operations on them have been applied). At any point in the simulation the state is separable under a partition reflecting the structure of the tree decomposition; we will show how this leads to a runtime of $\widetilde O(\|T\|_\omega^\omega)$ which can be far less than the na\"ive simulation based on the Gottesman-Knill theorem. 

The output $y\in \{0,1\}^{n_t}$ of the sampling subroutine will typically fail to have all the merge ancillas in the all-zeros state,
and thus will not directly provide a solution to the graph state simulation problem (since Eq.~\eqref{eq:postsel} will not be satisfied). In addition,  $y$ may disagree with the given measurement outcomes for the data qubits in the set $\mathcal{P}$ to be postselected.  We will fix both issues in the \textit{correction subroutine} of the algorithm.  The goal of the correction subroutine is to sample a uniformly random $n_t$-qubit Pauli $P_{\mathrm{cor}}$ such that
\begin{equation}
P_{\mathrm{cor}}(U_{\mathrm{bases}}\otimes I)\circuit|+^{n_t}\rangle \propto (U_{\mathrm{bases}}\otimes I)\circuit|+^{n_t}\rangle
\label{eq:pcor1}
\end{equation}
and 
\begin{equation}
P_{\mathrm{cor}} |y\rangle \propto |z\rangle\otimes |0^{n_a}\rangle \qquad z|_{\mathcal{P}}=m.
\label{eq:pcor2}
\end{equation}
Here the notation $z|_{\mathcal{P}}=m$ means that $z\in \{0,1\}^n$ matches the desired postselected outcomes (i.e., $z$ agrees with $m$ on the qubits in the set $\mathcal{P}$).  The runtime used in our algorithm to compute $P_{\mathrm{cor}}$ is $O(\|T\|_\omega^\omega)$.  For the special case where postselection on the data qubits is not needed (i.e., $\mathcal{P}=\varnothing$), we also provide a simplified algorithm with a faster $O(n_t+|E|)$ runtime.

Plugging Eqs.~(\ref{eq:pcor1}, \ref{eq:pcor2}) into Eq.~\eqref{eq:py} and using Eq.~\eqref{eq:postsel} we see that
\begin{equation}
p(y)=|\langle y|P_{\mathrm{cor}}(U_{\mathrm{bases}}\otimes I)\circuit|+^{n_t}\rangle|^2\propto |\langle z|U_{\mathrm{bases}}|G\rangle|^2.
\label{eq:correctequiv}
\end{equation}

Therefore, the binary string $z$ defined by $P_{\mathrm{cor}}|y\rangle\propto |z\rangle$ is a valid measurement outcome for the given  graph state simulation instance, i.e., it occurs with nonzero probability if the graph state is measured in the given bases, and matches the specified postselection outcomes. To show that $z$ is sampled from the correct probability distribution, which is uniform over all valid measurement outcomes, we appeal to \Cref{claim:existp}. Suppose $z'$ is another valid measurement outcome for this given graph state simulation instance. \Cref{claim:existp} states that there exists a Pauli $P$ in the stabilizer group of $(U_{\mathrm{bases}}\otimes I)\circuit|+^{n_t}\rangle$ such that $P|z\rangle\propto |z'\rangle$. Since $P_{\mathrm{cor}}$ is a uniformly random stabilizer subject to Eqs.~(\ref{eq:pcor1}, \ref{eq:pcor2}), the stabilizer $P P_{\mathrm{cor}}$ is also uniformly random subject to those constraints. By definition, sampling $P_{\mathrm{cor}}$ in the correction step leads to output $z$, whereas sampling $PP_{\mathrm{cor}}$ leads to output $z'$. Therefore any two valid measurement outcomes $z,z'$ are equally likely and we have shown that the output of the correction step is uniformly sampled from the set of valid measurement outcomes.

In the remainder of this section we fill in the details of the algorithm described above. In \Cref{sec:coft} we describe the circuit $\circuit$ and establish Eq.~\eqref{eq:postsel}. In \Cref{sec:samplingalg} we describe the sampling subroutine, show that it outputs a bit string $y$ sampled according to Eq.~\eqref{eq:py}, and establish the runtime bound $\widetilde O(\|T\|_\omega^\omega)$. Finally, in \Cref{sec:coralg} we describe the correction subroutine, and show that it outputs a random Pauli $P_{\mathrm{cor}}$ satisfying equations \eqref{eq:pcor1} and \eqref{eq:pcor2} using a runtime $\widetilde O(\|T\|_\omega^\omega)$.

\subsection{The quantum circuit \texorpdfstring{$\circuit$}{C} \label{sec:coft}}
Given a nice tree decomposition $T$ of a graph $G$, let us consider a quantum computation obtained by replacing each node with a gadget based on its type---introduce, forget, or merge. The gadgets depend on both the node and its children, so let $A$ be bag of the current node, and let $B_1, B_2$ be the bags of its children, or just $B$ for a single child, or $B = \varnothing$ if no children. The gadgets are depicted in \Cref{fig:circuit_gadgets} and described below:
\begin{description}
\item[Introduce:] Add a qubit in state $\ket{+}$ for each vertex in $A \backslash B$, and pass everything else through. 
\item[Forget:] Set aside qubits $B \backslash A$ (i.e., do not pass them to the parent), but first, we apply CZ gates for any edges with at least one endpoint in $B \backslash A$ and both endpoints in $B$.
\item[Merge:] For each vertex $v \in B_1 \cap B_2$, apply a CNOT from the qubit $v \in B_2$ to qubit $v \in B_1$, then measure the latter qubit, a \emph{merge ancilla}, in the $Z$-basis.
\end{description}
In other words, when a bag introduces a qubit, we add a new qubit in state $\ket{+}$. Just before a bag forgets a qubit, there is a ``last call'' to apply any CZ gates involving that qubit. Last, merge nodes apply the merge gadget to combine duplicate qubits.

See \Cref{fig:td_to_circuit} for an example of the entire construction. We define the quantum circuit $\circuit$ to be the \textit{unitary} part of the circuit constructed in this way, which contains only CZ and CNOT gates and does not contain any measurements. So the quantum computation described above and depicted in \Cref{fig:td_to_circuit} consists of preparing all qubits in the $|+^{n+n_a}\rangle$ state, applying $\circuit$ and then measuring all merge ancillas in the computational basis. 

Note that the total number of qubits in the circuit is upper bounded as
\begin{equation}
\text{Total number of qubits in } \circuit = n_t = n+n_a \leq \|T\|_1.
\label{eq:numq}
\end{equation}
The circuit contains a CZ gate for every edge in the graph and a CNOT gate for each merge ancilla. Therefore the total number of gates in $\circuit$ is 
\begin{equation}
\text{Total number of gates in } \circuit = |E|+ n_a.
\label{eq:numg}
\end{equation}

\begin{figure}
    \centering
    \sbox0{
    \begin{quantikz}[row sep=.8cm]
    \lstick[wires=3]{$B$} & \qw & \qw \rstick[wires=5]{$A$} \\
    & \qw & \qw \\
    & \qw & \qw \\
    & \lstick{$\ket{+}$} & \qw \\
    & \lstick{$\ket{+}$} & \qw
    \end{quantikz}
    }
    \sbox1{
    \begin{quantikz}[row sep=.7cm]
    \lstick[wires=5]{$B$} & \ctrl{1} & \ctrl{3} & \ctrl{4} & \qw &  \\
    & \control{} & \qw & \qw & \qw & \qw \rstick[wires=3]{$A$} \\
    & \control{} & \qw & \qw & \qw & \qw \\
    & \qw & \control{} & \qw & \qw & \qw  \\ 
    & \ctrl{-2} & \qw & \control{} & \qw & \\
    \end{quantikz}
    }
    \sbox2{
    \begin{quantikz}[row sep=.7cm]
    \lstick[wires=3]{$B_1$} & \qw & \qw & \qw & \qw \rstick{$B_1 \backslash B_2$} \\[-5pt]
    & \targ{} & \qw & \meter[scale=.8]{} \rstick[wires=2]{merge \\ ancillas} \\[-10pt]
    & \qw & \targ{} & \meter[scale=.8]{} \\
    \lstick[wires=3]{$B_2$} & \ctrl{-2} & \qw & \qw & \qw\rstick[wires=2]{$B_1 \cap B_2$} \\
    & \qw & \ctrl{-2} & \qw & \qw \\
    & \qw & \qw & \qw & \qw \rstick{$B_2 \backslash B_1$} \\
    \end{quantikz}
    }
    \begin{tabular}{ccc}
        \bf{Introduce} & \bf{Forget} & \bf{Merge} \\
        \raisebox{5pt}{\usebox0} & \raisebox{-5pt}{\usebox1} & \usebox2 
    \end{tabular}
    \caption{Examples of the three gadgets for translating tree decomposition $T$ to circuit $\circuit$.}
    \label{fig:circuit_gadgets}
\end{figure}
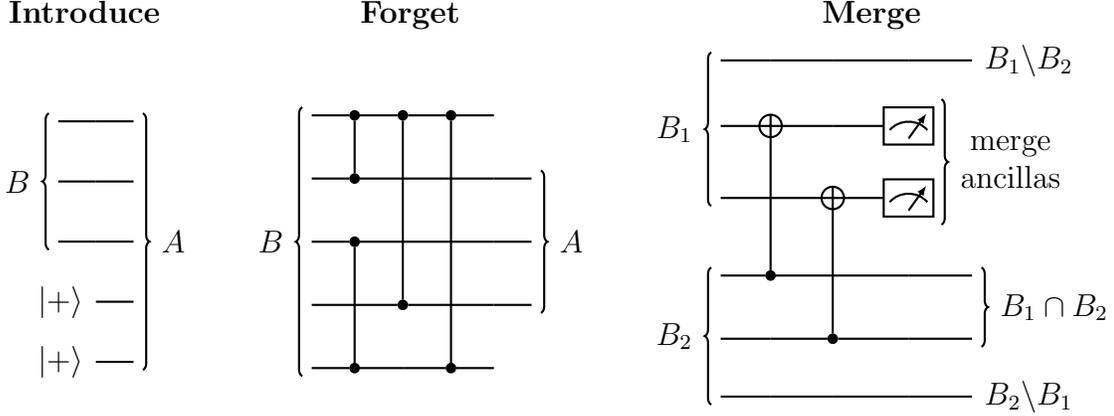

\begin{figure}
\begin{center}
\scalebox{.98}{\begin{tikzpicture}
\tikzstyle{bag}=[rectangle, rounded corners=8pt, minimum size=30pt, text width=30pt, align=center, fill opacity=0.1, text opacity=1]
\tikzstyle{introduce}=[bag, fill=introduce]
\tikzstyle{forget}=[bag, fill=forget]
\tikzstyle{merge}=[bag, fill=merge]
\matrix[row sep=17pt, column sep=23pt] {
\node (ABC) [introduce]{\phantom{$A B C$}}; \node (ABC_text) {$A B C$}; &  \node (BC) [forget]{\phantom{$B C$}};  \node (BC_text) {$B C$}; &  \\
& & \node (BCD) [merge]{\phantom{$B C D$}}; \node (BCD_text) {$B C D$}; & \node (BD) [forget]{\phantom{$B D$}};  \node (BD_text) {$B D$}; \\ 
\node (CDE) [introduce]{\phantom{$C D E$}}; \node (CDE_text) {$C D E$}; & \node (CD) [forget]{\phantom{$C D$}}; \node (CD_text) {$C D$}; & & & \node (BDG) [merge]{\phantom{$B D G$}}; \node (BDG_text) {$B D G$}; & \node (BG) [forget]{\phantom{$B G$}}; \node (BG_text) {$B G$};  & \node (BFG) [introduce]{\phantom{$B F G$}};  \node (BFG_text) {$B F G$}; & \node (EMPTY) [forget]{\phantom{$\varnothing$}};  \node (EMPTY_text) {$\varnothing$}; \\
& & \node (DGH) [introduce]{\phantom{$D G H$}};  \node (DGH_text) {$D G H$}; & \node (DG) [forget]{\phantom{$D G$}};  \node (DG_text) {$D G$}; \\
};
\newcommand{\roundec}{2.5pt}
\path[draw,-] (ABC) edge (BC)
(CDE) edge (CD)
($(BC.south east) - (\roundec,-\roundec)$) -- ($(BCD.north west) - (-\roundec,\roundec)$)
($(CD.north east) - (\roundec,\roundec)$) -- ($(BCD.south west) - (-\roundec,-\roundec)$)
(BCD) edge (BD)
(DGH) edge (DG)
($(BD.south east) - (\roundec,-\roundec)$) -- ($(BDG.north west) - (-\roundec,\roundec)$)
($(DG.north east) - (\roundec,\roundec)$) -- ($(BDG.south west) - (-\roundec,-\roundec)$)
(BDG) edge (BG)
(BG) edge (BFG)
(BFG) edge (EMPTY);
\tikzstyle{type}=[rectangle, rounded corners=2pt, fill opacity=.1]
\tikzstyle{legend}=[]
\node (INTRODUCE) [type, fill=introduce, above=80pt of BG]{};
\node (INTRODUCELABEL) [legend, right=5pt of INTRODUCE]{Introduce node};
\node (FORGET) [type, fill=forget, below=5pt of INTRODUCE]{};
\node (FORGETLABEL) [legend, right=5pt of FORGET]{Forget node};
\node (MERGE) [type, fill=merge, below=5pt of FORGET]{};
\node (MERGELABEL) [legend, right=5pt of MERGE]{Merge node};
\end{tikzpicture}}
\scalebox{.98}{\input{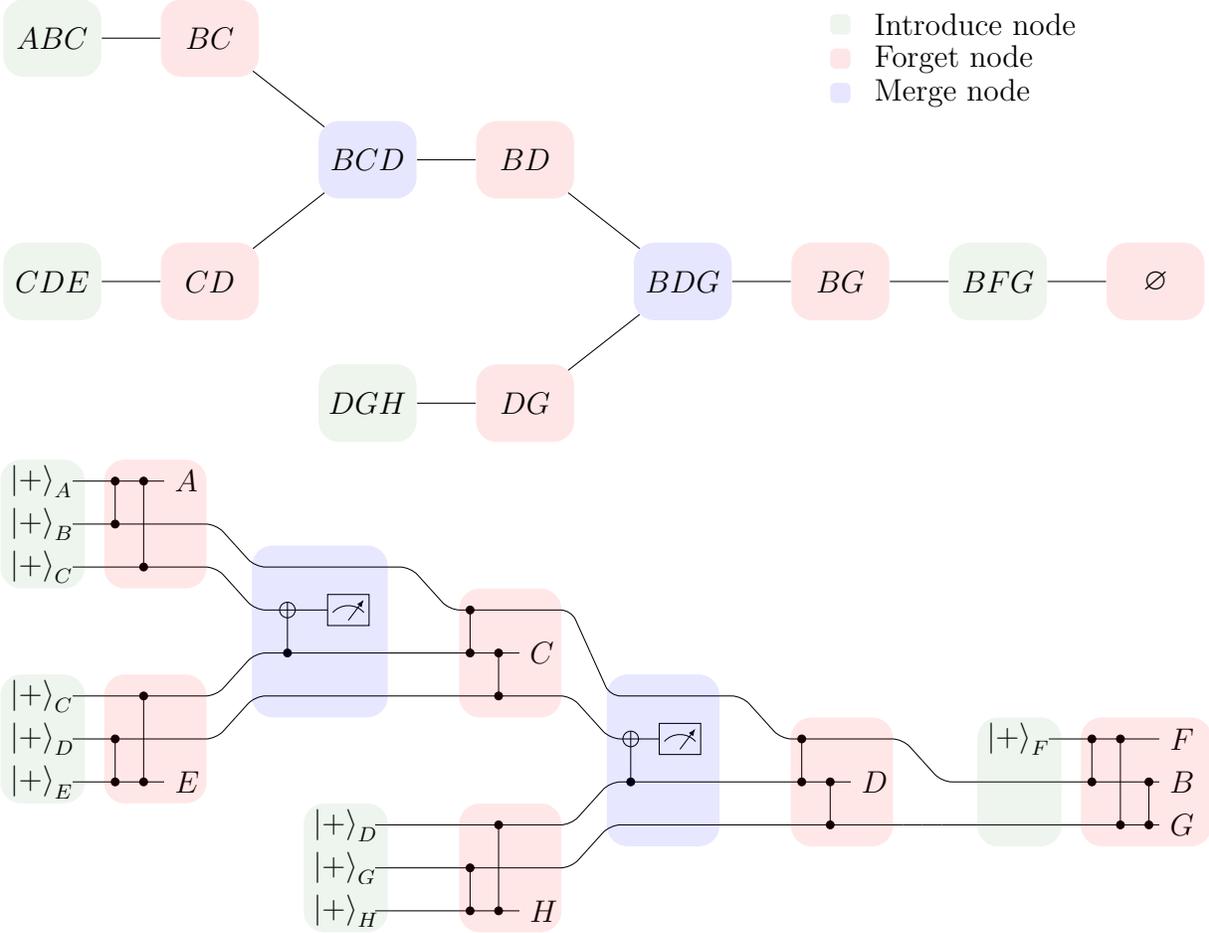}}
\end{center}
\caption{\textit{Top:}  Nice tree decomposition $T$ of graph in \Cref{fig:graph_and_td}.  \textit{Bottom:}  The quantum computation which prepares the state $\circuit|+^{n_t}\rangle$ and then measures all merge ancillas in the computational basis.}
\label{fig:td_to_circuit}
\end{figure}

We now establish Eq.~\eqref{eq:postsel}.

\begin{theorem}
\label{thm:graph_state_equality}
Let $G$ be a graph with a nice tree decomposition $T$. If all the merge ancillas are measured in the state $\ket{0}$, then the state on the remaining qubits is precisely $\ket{G}$.
\end{theorem}
\begin{proof}
We will prove the following more general statement:

\begin{claim}
The state prepared by  $\circuit$  is a graph state $|G'\rangle= \circuit|+^{n_t}\rangle$ where the graph $G'=(V',E')$ has the following properties:
\begin{itemize}
    \item $V'=V \cup M$ where $V$ is the vertex set of $G$ and $M$ contains a vertex for each merge ancilla in $\circuit$ (i.e., $|M| = n_a$).
    \item $E'=E \cup F$ where $E$ is the edge set of $G$ (with both endpoints in $V$) and $F$ only contains edges with at least one endpoint in $M$.
\end{itemize}
\label{claim:graphgprime}
\end{claim}

Theorem \ref{thm:graph_state_equality} follows directly from the claim, since (using an easily verified property of graph states)  $\left(I_V\otimes \langle 0|_M\right)|G'\rangle$ is proportional to the graph state for the induced subgraph of $G'$ on vertices in $V$. Claim~\ref{claim:graphgprime} states that this induced subgraph is $G$.

It remains to establish the claim. Consider the two circuit identities shown in Fig.~\ref{fig:circuitidentities} describing how CNOT moves past CZ. 

\begin{figure}[h!]
\centering
\begin{quantikz}
& \control[blue]{} \vqwblue{1} & \qw & \qw & && \qw & \control[blue]{} \vqwblue{1} & \control[blue]{} \vqwblue{2}\qw & \qw  \\
& \control[blue]{} & \targ{} & \qw & = && \targ{} & \control[blue]{} & \qw & \qw \\
& \qw & \ctrl{-1} & \qw & && \ctrl{-1} & \qw & \control[blue]{} & \qw
\end{quantikz}
\hspace{2cm}
\begin{quantikz}
& \qw & \targ{} & \qw & && \targ{} & \qw & \qw \\
& \control[red]{} \vqwred{1} & \ctrl{-1} & \qw & = && \ctrl{-1} & \control[red]{} \vqwred{1} & \qw \\
& \control[red]{} & \qw & \qw & && \qw & \control[red]{} & \qw 
\end{quantikz}
\caption{Two circuit identities.\label{fig:circuitidentities}}
\end{figure}

Putting these together, we see that we can push CZ gates to the right of CNOT gates that appear in a merge gadget, as shown in Fig.~\ref{fig:pushCNOTS}.
\begin{figure}[h!]
\centering

\begin{quantikz}
\lstick{$a$}   & \control[blue]{} \vqwblue{1} & \qw & \qw \rstick{$a$} \\
\lstick{$b_1$} & \control[blue]{} & \targ{} & \qw \rstick{$\text{Merge ancilla}$} \\
\lstick{$b_2$} & \control[red]{} & \ctrl{-1} & \qw \rstick{$b$} \\
\lstick{$c$}   & \control[red]{} \vqwred{-1} & \qw & \qw \rstick{$c$}
\end{quantikz}
\; =
\begin{quantikz}
\lstick{$a$}   & \qw       & \control[blue]{} \vqwblue{1}   & \control[blue]{} \vqwblue{2}   & \qw \rstick{$a$}  \\
\lstick{$b_1$} & \targ{}   & \control[blue]{} & \qw        & \qw \rstick{$\text{Merge ancilla}$} \\
\lstick{$b_2$} & \ctrl{-1} & \control[red]{} & \control[blue]{} & \qw \rstick{$b$} \\
\lstick{$c$}   & \qw       & \control[red]{} \vqwred{-1}  & \qw        & \qw \rstick{$c$}
\end{quantikz}
\caption{Pushing CZ gates to the right of CNOTS\label{fig:pushCNOTS}}
\end{figure}

We can use this rule to move all the CNOT gates in $\circuit$ to the left of all CZ gates. The CNOT gates can then be removed from the circuit altogether, as they have a trivial action on the initial state $|+^{n_t}\rangle$ (that is, $\mathrm{CNOT} \ket{++}=\ket{++}$). 

Once we have removed all the CNOTs from the circuit in this manner, we are left with a circuit composed of CZ gates acting on $|+^{n_t}\rangle$, that is, a graph state. The graph $G'$ has $n+n_a$ vertices as claimed. Using the circuit identities above we can also verify that the edge set of $G'$ has the desired properties. 
\end{proof}

\subsection{The sampling subroutine \label{sec:samplingalg}}

The goal of the sampling subroutine is to (classically simulate) the following process: initialize $n_t=n+n_a$ qubits in the $|+\rangle$ state, apply the circuit $\circuit$ (described in the previous section), measure all $n$ data qubits in the given Pauli bases $\{P_v\}_{v\in V}$, and measure all $n_a$ merge ancillas in the computational basis. The output of the sampling subroutine is a binary string $y\in\{0,1\}^{n_t}$ sampled from the distribution Eq.~\eqref{eq:py}. 

The sampling subroutine is a recursive algorithm which is described in Algorithm \ref{alg:basic_sampling}. In particular, the desired output $y\in \{0,1\}^{n_t}$ is obtained by running Algorithm \ref{alg:basic_sampling} with the given graph $G$, and tree decomposition $T$ (recall it has an empty root bag). The algorithm should be interpreted as a \textit{classical} algorithm in which all stabilizer states are represented using tableaux and all operations are performed using the update rules summarized in Table \ref{table:clifford_asymptotics}.  

\begin{algorithm}
	\caption{Sampling subroutine \label{alg:basic_sampling}}
	\begin{algorithmic}[1]
		\Require{Nice tree decomposition $T$}
		{
		\renewcommand{\algorithmicrequire}{\textbf{Global:}}
		\Require{Graph $G$ and Pauli measurement bases $\{P_v\}_{v\in V(G)}$.}
		}
		\Ensure{Print measurement outcomes of $\circuit$ for merge ancillas and each vertex \emph{not} in the root bag of $T$.  Return post-measurement state of vertices in root bag.}
		\Statex{}
		\Function{SampleGraph}{$T$}
 		\Statex{$(\triangleright)$ We denote the root bag of $T$ by $A$, and its child bags by $B_1$ and $B_2$ (or just $B$ if it only has one child). Similarly, let $S_1$ and $S_2$ be the left and right subtrees of $T$ (or $S$ if there is only one child).}
		\If{$T$ is a single node}
		\State \Return{$\bigotimes_{v \in A} \ket{+}_v$}
		\EndIf
		\If{root($T$).type = Introduce}
		\Let{$\ket{\psi}$}{\Call{SampleGraph}{$S$}}
		\State \Return{$\ket{\psi} \otimes \bigotimes_{v \in A \backslash B} \ket{+}_v$}
		\EndIf
		\If{root($T$).type = Forget}
		\Let{$\ket{\psi}$}{\Call{SampleGraph}{$S$}}
		\State Apply CZ to $\ket{\psi}$ for each edge $\{u,v\} \in E(G)$ with $u \in B$ and $v \in B \backslash A$
		\State For each $v \in B \backslash A$, measure corresponding qubit of $\ket{\psi}$ in $P_v$ basis
		\State \Return{post-measurement state}
		\EndIf
		\If{root($T$).type = Merge}
		\Let{$\ket{\psi_1}$}{\Call{SampleGraph}{$S_1$}}
		\Let{$\ket{\psi_2}$}{\Call{SampleGraph}{$S_2$}}
		\State For each $v \in B_1 \cap B_2$, apply CNOT controlled by qubit $v$ in $\ket{\psi_2}$, to qubit $v$ in $\ket{\psi_1}$. 
		\State For each such $v$, measure the corresponding qubit of $\ket{\psi_1}$ in $Z$ basis
		\State \Return{post-measurement state}
		\EndIf
		\EndFunction
	\end{algorithmic}
\end{algorithm}

\begin{theorem}
\label{thm:basic_analysis}
\Cref{alg:basic_sampling} runs in $\widetilde O(\|T\|_\omega^\omega)$ time.
\end{theorem}
\begin{proof}

\Cref{alg:basic_sampling} recursively simulates $\circuit$ gadget by gadget up the tree decomposition, so it suffices to upper bound the runtime for simulating each gadget using the stabilizer simulation algorithms summarized in \Cref{table:clifford_asymptotics}. Let $A$ be the bag of the current node, and $B_1, B_2$ be the bags of its children (or $B$ if there is only one child). 
\begin{description}
\item[Introduce node:] Initializing $|A \backslash B|$ many qubits in the $\ket{+}$ state uses a runtime $O(|A|^2)$.
\item[Forget node:] First we apply all CZ gates for edges which touch forgotten vertices in $B$, which takes $O(|B|^\omega)$ time. We then measure qubits at vertices in $B \backslash A$, requiring $\widetilde O(|B\backslash A|^{\omega-2} |B|^{2})\leq \widetilde O(|B|^{\omega})$ time.
\item[Merge node:] We apply CNOT gates between $O(|B_1 \cap B_2|)$ pairs of qubits, and a similar number of measurements. We bound both as costing at most
\[
\widetilde O((|B_1| + |B_2|)^\omega)  = \widetilde{O}(|A|^{\omega}).
\]
\end{description}
In each of the above cases, we can conservatively upper bound the cost as the sum of the $\omega$th power of all the bags involved. Since each bag is involved in at most two gadgets (once as a child, once as a parent), the cost for the entire circuit is at most $\widetilde O(\|T\|_\omega^\omega)$ as claimed.
\end{proof}

\subsection{The correction subroutine \label{sec:coralg}}

The sampling subroutine from the previous section outputs a binary string $y\in \{0,1\}^{n_t}$ sampled from the distribution Eq.~\eqref{eq:py}.  The correction subroutine takes this $y$ as input and computes a uniformly random $n_t$-qubit Pauli $P_{\mathrm{cor}}$ such that Eqs.~(\ref{eq:pcor1},\ref{eq:pcor2}) hold, or reports that no such Pauli exists. In the former case, the output of the correction subroutine is the binary string $z\in \{0,1\}^n$ from Eq.\eqref{eq:pcor2} which solves the graph state simulation problem, i.e., it is sampled from the conditional distribution obtained by measuring all qubits of the graph state $|G\rangle$ in the given Pauli bases $\{P_v\}_{v\in V}$ and postselecting on the given outcomes $m$ for qubits in $\mathcal{P}$.

In Section \ref{sec:corgeneral} we provide a description of the correction subroutine, which completes our description of the classical algorithm for graph state simulation and the proof of Theorem \ref{thm:sumofcubes}. In some applications, such as the Clifford circuit simulation task described in \Cref{sec:planar_circuit_application}, one may be interested in the special case of graph state simulation without postselection.  For this special case we provide a simplified correction subroutine with an improved runtime, in Section \ref{sec:withoutpost}.

\subsubsection{Correction subroutine: general case \label{sec:corgeneral}}
There are two types of qubits for which our correction algorithm has strict requirements: postselected qubits (i.e., those qubits in the set $\mathcal{P}$) and merge ancillas.  For the merge ancillas, we enforce that their measurement outcomes are $0$.  For postselected qubits, we enforce that their measurement outcomes coincide with some desired values specified by the input to the problem.  Nevertheless, in our analysis below we will treat each type of qubit equally since our correction subroutine will not need any properties that distinguish the two cases. Our goal is to compute a uniformly random Pauli $P_{\mathrm{cor}}$ that stabilizes the state $(U_{\mathrm{bases}}\otimes I)\circuit|+^{n_t}\rangle$ (up to a sign) and whose $\pfont X$-component either flips or leaves unchanged all bits of $y$ corresponding to merge ancillas and postselected qubits, depending on its desired measurement result (see Eq. \eqref{eq:pcor2}).

Let us represent $n$-qubit Pauli operators using the homomorphism mapping $\alpha \pfont X^{a} \pfont Z^{b}$ to the $2n$-bit vector $(a, b)$. The phase $\alpha$ is not needed because it will not affect measurement. Under this homomorphism, the Pauli group maps to $\mathbb F_2^{2n}$, a subgroup of the Paulis maps to a linear subspace of $\mathbb F_2^{2n}$, and a coset of a Pauli subgroup maps to an \emph{affine} subspace of $\mathbb F_2^{2n}$.

When searching for a Pauli stabilizer with particular $\pfont X$- and $\pfont Z$-components, we represent the desired components by a \emph{pattern} $\pi = (\pi^{(X)},\pi^{(Z)}) \in \{0,1,*\}^{2n}$, where $*$ indicates that the $\pfont{X}$- or $\pfont{Z}$-component can be either $0$ or $1$.  We associate the pattern $\pi$ with the set of all allowable outcomes
\[
\Pi := \{(a,b) \in \{0,1\}^{2n} : a_i = \pi^{(X)}_i \text{ if } \pi^{(X)}_i \in \{0,1\} \text{ and } b_j = \pi^{(Z)}_j \text{ if } \pi^{(Z)}_j \in \{0,1\} \},
\]
which forms an affine subspace of $\{0,1\}^{2n}$. We say that a Pauli $\pfont{X}^a \pfont{Z}^b$ \emph{respects} $\pi$ if $(a,b)\in\Pi$. The rest of this section is devoted to proving the following theorem:
\begin{restatable}{theorem}{powerfulcorrection}
\label{thm:powerful_correction}
Let a nice tree decomposition $T$ of some $n$-vertex graph $G$ be given.  Let $\mathrm{Stab}\subseteq \mathbb{F}_2^{2n_t}$ be the linear subspace describing the stabilizer group of $(U_{\mathrm{bases}}\otimes I)\circuit|+^{n_t}\rangle$, and let $\Pi\subseteq \mathbb{F}_2^{2n_t}$ be the affine subspace describing a given pattern $\pi$. There is an $O(\lVert T \rVert_\omega^\omega)$ time algorithm that samples a uniformly random element of the intersection $\mathrm{Stab}\cap \Pi$, or reports that no such element exists.
\end{restatable}

While our algorithm will make use of the structure of the tree decomposition, most of the required operations are in fact independent of $T$. Consider then any $n$-qubit stabilizer state $\ket{\psi}$. We will define subsets of the stabilizer group of $\ket{\psi}$ that are given by affine subspaces $\mathcal A = \{ \begin{psmallmatrix} C \\ D \end{psmallmatrix} x + \begin{psmallmatrix} c \\ d \end{psmallmatrix} : x \in \{0,1\}^N \} \subseteq \{0,1\}^{2n}$ for matrices $C, D \in \{0,1\}^{n \times N}$, vectors $c, d \in \{0,1\}^n$, and $N \le 2n$ (note that if the latter condition did not hold then some of the columns of $\begin{psmallmatrix} C \\ D \end{psmallmatrix}$ would be linearly dependent).  Here, $C$ and $c$ correspond to the $\pfont X$-component of the stabilizer, while $D$ and $d$ correspond to the $\pfont Z$-component. For example, the affine subspace for the entire stabilizer group of $\ket{\psi}$ is such that each column of $\begin{psmallmatrix} C \\ D \end{psmallmatrix}$ is the bit representation of a stabilizer generator of $\ket{\psi}$ and $c, d$ are both zero. We say that $\mathcal{A}$ is an \emph{affine stabilizer subspace}. For convenience, we sometimes refer to the elements of $\mathcal A$ directly as Pauli operators (modulo signs) using the correspondence described above.

To find the elements of an affine stabilizer subspace $\mathcal{A}$ that also respect the pattern $\pi$ we must find all solutions $x$ to
\[
\begin{pmatrix}C\\D\end{pmatrix} x + \begin{pmatrix}c\\d\end{pmatrix} \in \Pi.
\]
We can write this more conventionally as the rectangular system of linear equations
\[
\begin{pmatrix}C'\\D'\end{pmatrix} x = \begin{pmatrix}c'\\d'\end{pmatrix} + \pi'
\]
where $C'$, $D'$, $c'$, $d'$ and $\pi'$ are obtained by removing row $i$ from $\begin{psmallmatrix}
C\\D
\end{psmallmatrix}$, $\begin{psmallmatrix}
c\\d
\end{psmallmatrix}$, and $\pi$ whenever $\pi_i = *$.  This system can be solved using the following lemma:
\begin{lemma}
\label{lem:fast_linear_equations}
Let $E \in \{0,1\}^{k \times \ell}$ and $\zeta \in \{0,1\}^k$ be given for some $k$ and $\ell$.  There is an $O(k^\omega + \ell^\omega)$-time algorithm which determines if the system of linear equations $Ex = \zeta$ has a solution.  If so, the algorithm reports its affine solution space, that is, a matrix $E' \in \{0,1\}^{\ell \times \ell}$ and vector $\zeta' \in \{0,1\}^\ell$ such that
\[
\{ x \in \{0,1\}^\ell : Ex = \zeta \} = \{E'z + \zeta' : z \in \{0,1\}^\ell \}.
\]
\end{lemma}
\begin{proof}
First construct a generalized inverse\footnote{A generalized inverse of a $k \times \ell$ matrix $E$ is an $\ell \times k$ matrix $E^g$ such that $E E^g E = E$.} $E^g$ of $E$ in $O(k^{\omega-1}\ell + k\ell^{\omega - 1})$ time using an algorithm of Ibarra, Moran, and Hui \cite{ibarra+moran+hui:1982}. This algorithm can be used directly if $k\leq\ell$; otherwise, if $k>\ell$, we may find a generalized inverse of $E^T$ and then use the fact that $E^T(E^T)^g E^T = E^T$ implies $E ((E^T)^g)^T E = E$. We have that $Ex = \zeta$ has a solution iff $E E^g \zeta = \zeta$, and furthermore, if a solution exists, then the solution space is given by \cite{ben2003generalized}
\begin{equation}\label{eq:solution_space}
\{ (I - E^g E)z + E^g \zeta : z \in \{0,1\}^\ell \}.
\end{equation}
Setting $E' := I - E^g E$ and $\zeta' := E^g \zeta$, which can be computed in time $O(k^\omega + \ell^\omega)$ and $O(k \ell)$, respectively, will give the lemma. All that remains to show is that $E^g E$ can be computed in $O(k^\omega + \ell^\omega)$ time.  If $\ell \le k$, then the matrices $E^g$ and $E$ can be partitioned into $\lceil k/\ell \rceil$ $\ell \times \ell$ matrices (if $\ell$ doesn't evenly divide $k$, just pad the matrix). To compute $E^g E$, we must then compute $\lceil k/\ell \rceil$ square $\ell \times \ell$ matrix multiplications, taking time $\lceil k/\ell \rceil \ell^\omega = O(k \ell^{\omega - 1})$. However, since we assumed $\ell \le k$ this running time is also $O(k^\omega)$. A similar analysis shows that when $\ell \ge k$, the running time is $O(\ell^\omega)$.
\end{proof}

We arrive at the following important fact: enforcing the pattern $\pi$ onto an affine stabilizer subspace returns another affine stabilizer subspace. In the following lemma, we summarize other important operations that share this property.

\begin{lemma}
\label{claim:affine_properties}
Let $\ket{\psi}$ be an $n$-qubit stabilizer state. Let $\mathcal A \subseteq \{0,1\}^{2n}$ be an affine stabilizer subspace representing a subset of stabilizers of $\ket{\psi}$. There exist algorithms to calculate the affine stabilizer subspace $\mathcal A'$ arising from each of the following operations:
\begin{itemize}
\item \emph{Enforcing pattern $\pi \in \{0,1,*\}^{2n}$ on $\mathcal A$}:  $\mathcal A' \subseteq \mathcal A$ contains exactly those Paulis that are both in $\mathcal A$ and respect pattern $\pi$. (Runtime $O(n^\omega)$.)
\item \emph{Applying a Clifford unitary $U$ to $\ket{\psi}$}:  $\mathcal A' = U \mathcal A U^\dag$.  
(Runtime $O(n^\omega)$.)
\item \emph{Restricting $\mathcal{A}$ to qubits in some subset $S \subseteq [n]$}:  $\mathcal A'$ is the set of Paulis $\bigotimes_{i \in S} P_i$ for which there exists a Pauli $\bigotimes_{i\in [n]}P_i$ in $\mathcal A$. (Runtime $O(n^\omega)$.)
\end{itemize}
Letting $\ket{\psi_1}$ and $\ket{\psi_2}$ be $n_1$- and $n_2$-qubit stabilizer states with affine stabilizer subspaces $\mathcal A_1 \subseteq \{0,1\}^{2n_1}$ and $\mathcal A_2 \subseteq \{0,1\}^{2n_2}$, we have 
\begin{itemize}
\item \emph{Taking tensor products}:  $\mathcal A' = \mathcal A_1 \oplus \mathcal A_2$ contains exactly those Paulis which are tensor products of elements of $\mathcal A_1$ and $\mathcal A_2$. (Runtime $O((n_1+n_2)^2)$.)
\end{itemize}
\end{lemma}
\begin{proof}
Let $\mathcal A = \{ \begin{psmallmatrix} C \\ D \end{psmallmatrix} x + \begin{psmallmatrix} c \\ d \end{psmallmatrix} : x \in \{0,1\}^N \} \subseteq \{0,1\}^{2n}$ be the affine stabilizer subspace for $\ket{\psi}$. We'll consider the four operations starting with pattern enforcement. Supposing the pattern $\pi$ has $k$ non-$*$ entries, this requires solving a system of linear equations $Ex = \zeta$ where $E \in \{0,1\}^{k \times N}$ and $\zeta \in \{0,1\}^{k}$. By \Cref{lem:fast_linear_equations}, the solution space is affine and can be computed in $O(n^\omega)$ time since $N \le 2n$.

For unitary application, recall that stabilizers of $\ket{\psi}$ are mapped to stabilizers of $U \ket{\psi}$ under conjugation by $U$:
\[
P \ket{\psi} = \ket{\psi} \iff (UPU^\dag) U\ket{\psi} = U \ket{\psi}
\]
 Therefore, the affine stabilizer subspace $\mathcal A$ is mapped to the affine stabilizer subspace $U \mathcal A U^\dag$ for the state $U \ket{\psi}$ by ``conjugating'' the columns of $\begin{psmallmatrix} C \\ D \end{psmallmatrix}$ and the displacement vector $\begin{psmallmatrix} c \\ d \end{psmallmatrix}$, i.e., associating each vector with a Pauli element, conjugating by $U$, and then reverting back to the bit representation. To do this, we simply multiply $\begin{psmallmatrix} C \\ D \end{psmallmatrix}^T$ and $\begin{psmallmatrix} c \\ d \end{psmallmatrix}^T$ by the tableau for $U$, which takes $O(n^{\omega})$ time (this is equivalent modulo signs to the procedure described in \Cref{thm:tableaux_composition}). 

Suppose now that we want to restrict $\mathcal A$ to some subset of the qubits of $\ket{\psi}$. We do this by tracing out the unwanted qubits. To trace out the $i$th qubit from $\mathcal A$, we simply remove the $i$th row of $C$, $D$, $c$, and $d$.  After tracing out $m$ qubits, we obtain the affine space:
\[
\mathcal A' = \{ \begin{psmallmatrix} C' \\ D' \end{psmallmatrix} x + \begin{psmallmatrix} c' \\ d' \end{psmallmatrix} : x \in \{0,1\}^N \} \subseteq \{0,1\}^{2(n-m)}
\]
for matrices $C', D' \in \{0,1\}^{(n-m)\times N}$ and vectors $c', d' \in \{0,1\}^{n-m}$.  Unfortunately, $N$ may now be quite large in comparison to the total number of qubits $n-m$.  If $N>2(n-m)$ we apply an algorithm of Ibarra, Moran, and Hui to find a maximal set of independent columns of $\begin{psmallmatrix} C' \\ D' \end{psmallmatrix}$ in time $O(n^\omega)$, so that our representation of $\mathcal A'$ becomes manageable once again \cite{ibarra+moran+hui:1982}.

Finally, consider the tensor product of two states $\ket{\psi_1}$ and $\ket{\psi_2}$ with affine stabilizer subspaces $\mathcal A_1$ and $\mathcal A_2$, respectively. Notice that Pauli operator $P_1\otimes P_2$ stabilizes $\ket{\psi_1}\otimes\ket{\psi_2}$ (up to signs) iff $P_1\in\mathcal{A}_1$ and $P_2\in\mathcal{A}_2$. In other words, the new stabilizer affine space is given by the direct product of the two previous spaces:
\begin{equation}\label{eq:affine_product}
\mathcal A' = \mathcal A_1 \oplus \mathcal A_2 = 
\{ \begin{psmallmatrix} C_1 & 0 \\ D_1 & 0 \\ 0 & C_2 \\ 0 & D_2 \end{psmallmatrix} x + 
\begin{psmallmatrix} c_1 \\ d_1 \\ c_2 \\ d_2 \end{psmallmatrix} : x \in \{0,1\}^{N_1 + N_2} \}.
\end{equation}
\end{proof}
\Cref{claim:affine_properties} shows that tracing out qubits of a stabilizer state still results in an affine subspace. Importantly, one can check that all of the other operations (applying gates, enforcing patterns, and taking tensor products) in the lemma can also be performed on subsets of qubits of a stabilizer state using only the affine space for that restricted subset. This is due to the fact that these operations commute with the partial trace operator. Because of this, we will continue to use the phrase ``affine stabilizer subspace'' to refer to these restricted subspaces to emphasize their connection to stabilizers of some larger state.

We are now ready to prove the main theorem for this section (restated below):
\powerfulcorrection*
\begin{proof}
Let $\circuit' = (U_\mathrm{bases}\otimes I)\circuit$ for ease of notation.
A na{\"i}ve way to do this would be to compute $\mathrm{Stab}$ in its entirety, enforce $\pi$ using Lemma \ref{lem:fast_linear_equations}, and randomly sample an element of the solution space. This $O(n_t^\omega)$-time approach is far too slow. Instead, we use the tree structure of $\circuit'$ to compute the stabilizers of $\circuit'\ket{+^{n_t}}$ while simultaneously enforcing $\pi$. In \Cref{sec:affine_example} we work through an example of the algorithm described below.

We begin by defining an affine space $\mathcal{A}_B$  for each bag $B$ of the given tree decomposition $T$. Consider the subtree $T_B$ of $T$ rooted at $B$. Let us write $\circuit|_B$ and $\circuit|_{T_B}$ to denote the portions of $\circuit$ corresponding to $B$ and $T_B$, respectively. For example, for the root node $R$ of the tree decomposition in \Cref{fig:td_to_circuit}, $\circuit|_R$ consists of three wires and three CZ gates; and for the leftmost merge node $M = \{ B, C, D \}$, the circuit $\circuit|_{T_M}$ consists of six wires, four CZ gates, and one CNOT gate. Define $\circuit'|_B$ and $\circuit'_{T_B}$ analogously. Let $\ket{\psi_B}=\circuit'|_{T_B}\ket{+^{n_{T_B}}}$, where $n_{T_B}$ is the number of qubits in $\circuit'|_{T_B}$. First consider an affine space $\mathcal{Y}_B$ which describes the stabilizers of $\ket{\psi_B}$ whose $\pfont X$- and $\pfont Z$-components agree with $\pi$ on all qubits that appear \textit{only} in $\circuit'|_{T_B}$. These are the qubits which are merged or forgotten in  $\circuit'|_{T_B}$ and are therefore untouched by all operations that occur later in the circuit $\circuit'$. Now define $\mathcal{A}_B$ to be the restriction of $\mathcal{Y}_B$ to coordinates of qubits in $\circuit'|_B$. For example, $\mathcal{Y}_R=\textrm{Stab}\cap\Pi$ and $\mathcal{A}_R\subseteq \mathbb{F}_2^6$.

We construct $\mathcal{A}_B$ inductively, working our way up the tree. Each node inherits an affine subspace from its children (unless it is a leaf node), and the first step will always be to restrict that affine subspace to the qubits of $\circuit'|_B$. For an introduce node in which $N$ qubits are initialized, we create the affine space $\mathcal I = \{ \begin{psmallmatrix} I \\ 0 \end{psmallmatrix} x +\begin{psmallmatrix} 0 \\ 0 \end{psmallmatrix}: x \in \{0,1\}^{N} \}$ for the space of stabilizers of the $\ket{+^N}$ state.  If $B$ has a child $B'$, then $\mathcal{A}_B = \mathcal{A}_{B'} \oplus \mathcal I$; otherwise, $\mathcal{A}_B = \mathcal I$. For a forget node, we apply the CZ gates and local operators of $\circuit'|_B$, and then enforce $\pi$ on all qubits being forgotten in $B$. For a merge node we take the direct product of the affine spaces for each child, apply the CNOT gates of $\circuit'|_B$, and enforce $\pi$ on the merged qubits.

We continue up the tree until we reach the root node $R$.  If at any point during this process we encounter an affine space that is empty, we report that a stabilizer matching $\pi$ does not exist. Otherwise by definition $\mathcal{A}_R$ contains the restriction of such a stabilizer to coordinates of qubits in $\circuit'|_R$. All that remains is to reconstruct the remaining coordinates using the other affine spaces. To do so, we process the bags in reverse order.

Let $R'$ be the child of $R$. By definition, a stabilizer $P_\textrm{cor}$ of $\circuit'\ket{+^{n_t}}$ respects $\pi$ iff its restriction to qubits in $\circuit'|_R$ is in $\mathcal{A}_R$. With this in mind, we choose a uniformly random element $(a,b)$ from $\mathcal{A}_R$ and set $\pfont{X}^a\pfont{Z}^b$ to be the tensor element of $P_\textrm{cor}$ corresponding to qubits of $R$. To reconstruct the remaining components of $P_\textrm{cor}$, we push $(a,b)$ backward through $\circuit'|_R$ by computing $\pfont{X}^{a'} \pfont{Z}^{b'} \propto (\circuit'|_R)^\dagger \pfont{X}^a\pfont{Z}^b (\circuit'|_R)$; randomly sample an element $(a'',b'')$ of $\mathcal{A}_{R'}$ respecting the pattern $(a',b')$; and set the $j$th tensor element of $P_\textrm{cor}$ to be $\pfont{X}^{a''_j}\pfont{Z}^{b''_j}$, for each qubit $j$ being merged or forgotten in $\circuit'|_{R'}$. We continue in this way down the tree, always randomly sampling an element of the affine stabilizer subspace of the child node that preserves choices made higher up in the tree.

The number of qubits that we handle at any given moment is never more than the size of the bag being operated upon plus that of its children. The only operations needed to construct $\mathcal{A}_B$ are those in \Cref{claim:affine_properties}. The time needed for each of these operations is at most the size(s) of the bag(s) involved raised to the $\omega$th power. All bags contribute to the runtime at most twice. Therefore all affine spaces can be constructed in time $O(\lVert T \rVert_\omega^\omega)$. When reconstructing $P_{\mathrm{cor}}$ from $(a,b)$, the only operations needed are those from \Cref{claim:affine_properties} together with the ability to sample from an affine subspace. The latter operation can be done in time at most the square of the corresponding bag's size, by randomly sampling a bit string of length $N$, where $N$ is the subspace dimension, and then multiplying that string by a rank-$N$ matrix and adding a displacement vector corresponding to the subspace. Each bag is processed just once during the reconstruction of $P_{\mathrm{cor}}$ from $(a,b)$, and therefore this process also takes time $O(\lVert T \rVert_\omega^\omega)$.

Finally, we must show that $P_{\mathrm{cor}}$ is chosen uniformly from $\textrm{Stab} \cap \Pi$.  Let us start with a general fact about sampling affine spaces.  Consider some affine space $\mathcal B \subseteq \mathbb F_2^n$ and let $\mathcal B_{S} \subseteq \mathbb F_2^{|S|}$ be the restriction of $\mathcal B$ onto the coordinates in $S \subseteq [n]$.  Now consider the following iterative process over subsets $S_1, S_2, \ldots$ such that $\bigcup_i S_i = [n]$:  sample $\beta_1$ uniformly at random from $\mathcal B_{S_1}$; define the affine space $\mathcal B^{(1)} := \{x \in \mathcal B : x \text{ agrees with } \beta_1 \text{ on } S_1\}$; then sample $\beta_2$ uniformly from $\mathcal B^{(1)}_{S_2}$; and so on.    Let $\beta$ be obtained by combining all the $\beta_i$ into a single string of length $n$.  Since each $\beta_i$ term was chosen such that it never conflicted with $\beta_j$ for $j < i$, the choice of $\beta$ is unambiguous. 

We claim that $\beta$ is a uniformly random element of $\mathcal B$.  After the first iteration, we only want to sample strings of $\mathcal B$ where the coordinates in $S_1$ are restricted to $\beta_1$.  As we've seen before, this imposes a linear constraint on $\mathcal B$, and so this new space ($\mathcal B^{(1)}$) is still affine by \Cref{lem:fast_linear_equations}.  In fact, we see from that lemma that the size of the solution space is independent of our choice of $\beta_1$ (provided, of course, that $\beta_1 \in \mathcal B_{S_1}$).  In the second iteration, we sample uniformly from $\mathcal B_1$ restricted to the coordinates in $S_2$.  And so on.  It follows that the number of different $\beta_j$ that occur with nonzero probability in the $j$th step is independent of the specific choices of $\beta_1, \ldots, \beta_{j-1}$.  Therefore, the final string $\beta$ chosen from $\mathcal B$ has the same probability of being chosen as any other string $\beta' \in \mathcal B$.

We now only have left to show that this is actually what is occurring in our correction protocol.  At the root node, we sample uniformly from the affine space $\mathcal A_R$ which, by construction, is equal to $\textrm{Stab} \cap \Pi$ restricted to qubits which are forgotten in the root bag.  Further down in the tree, when nodes are forgotten or merged, we only sample strings which are consistent with measurements made higher up in the tree. Recall that the qubits that are measured in a particular bag are never touched again higher in the tree.  Therefore, the affine space for a bag $B$ (i.e., the space $\mathcal A_B$) restricted to the qubits that are measured in that bag is equal to the affine space for the entire circuit $\textrm{Stab} \cap \Pi$ restricted to the qubits that are measured in that bag.  For the qubits in $\mathcal A_B$ that are not measured, we restrict them to the measurement results sampled higher in the tree.  In some cases, these measurement results are obtained \emph{after} applying some unitary.  However, the application of that unitary is equivalent to applying an invertible linear transformation on the coordinates of $\mathcal{A}_B$ corresponding to unmeasured qubits.  Applying the inverse of this transformation to the outcome sampled higher in the tree gives the unique string over the unmeasured qubits of $\mathcal A_B$ consistent with higher measurements.  Since we sample uniformly from $\mathcal A_B$ restricted to this string, the theorem follows.
\end{proof}

We close this subsection with the observation that \Cref{thm:powerful_correction} can be used to solve symmetric linear systems modulo $2$. That is, suppose we are given a system $Ay = b$ over $\mathbb F_2$ where $A \in \mathbb F_2^{n \times n}$ is the (symmetric) adjacency matrix of a graph $G$, and we are given a tree decomposition $T$ for the graph $G$. We claim that $Ay = b$ has a solution if and only if $\circuit \ket{+^{n_t}}$ has a stabilizer Pauli respecting the pattern $(*^{n} \, 0^{n_a}, b \,\, *^{n_a})$, and moreover the solution $y$ can be read off of the stabilizer. This yields a fast algorithm for planar or limited-treewidth linear systems.
\begin{restatable}{theorem}{linearsolver}
There exists an algorithm to find a uniformly random solution to $Ay=b$ (or report that no solution exists) in time $O(\| T \|_{\omega}^{\omega})$ given a graph $G$ with adjacency matrix $A \in \mathbb F_2^{n \times n}$, a vector $b \in \mathbb F_2^{n}$, and a nice tree decomposition $T$ for $G$.
\label{thm:linearsolver}
\end{restatable}

Conversely, we also claim such an algorithm for solving linear systems implies the correction procedure in \Cref{thm:powerful_correction}. See Appendix~\ref{sec:lss} for the proof of Theorem \ref{thm:linearsolver} and the full details of this equivalence.

In Section \ref{sec:plan} we show that, given a planar graph, we can compute a nice tree decomposition $T$ in $O(n\log n)$ time, such that $\|T\|_\omega^{\omega}=\tilde{O}(n^{\omega/2})$ (see Lemmas \ref{lem:tdnorm}, \ref{lem:tdruntime}). Thus we obtain Theorem \ref{thm:specialcase} from the introduction as a corollary.

We believe the new linear system solver stated in Theorem \ref{thm:linearsolver} may be of independent interest because, to our knowledge, existing solvers exploiting planarity and/or treewidth--such as Alon and Yuster \cite{alon2010solving} or Lipton, Rose, and Tarjan \cite{lipton1979generalized}--require $A$ to have full rank. Our algorithm makes no such assumption; in fact, it can sample a uniformly random solution when the solution is not unique. 

\subsubsection{Simplified correction subroutine for graph state simulation without postselection\label{sec:withoutpost}}
The correction procedure for the graph state simulation problem \emph{without postselection} can be simplified to pushing Paulis backwards and forwards through the circuit, removing the need for explicit linear system solving. This improves the runtime to linear in the size of the circuit (i.e., $O(|E| + n_t)$).

\begin{theorem}
\label{thm:ancilla_fix}
The state $\circuit|+^{n_t}\rangle$ has stabilizers with arbitrary $\pfont X$-component. Moreover, given a desired $\pfont X$-component vector $a \in \{0,1\}^{n_t}$, there is an $O(|E|+n_t)$-time algorithm that outputs a stabilizer $Q(a)=\pfont X^a \pfont Z^b$ (up to sign) for some $b \in \{0,1\}^{n_t}$.
\end{theorem}
\begin{proof}
Since $Q(a) = \pfont X^{a} \pfont Z^{b}$ is a stabilizer for $\circuit \ket{+^{n_t}}$ (up to sign), it follows that $\circuit^{\dag} \pfont X^{a} \pfont Z^{b} \circuit$ stabilizes $\ket{+^{n_t}}$ (up to sign), and is therefore of the form $\pm \pfont X^{a'}$. That is,
\begin{equation}
\circuit^{\dag} Q(a) \circuit = \circuit^{\dag} \pfont X^{a} \pfont Z^{b} \circuit  = \pm \pfont X^{a'}
\label{eq:solve}
\end{equation}
for some $a' \in \{ 0, 1 \}^{n_t}$. 

Recall that $\circuit$ is composed of CZ and CNOT gates, both of which normalize the group generated by $\pfont Z$-Paulis, 
\[
\mathcal{Z} := \{ \pfont Z^{b'} : b' \in \{ 0, 1 \}^{n_t} \}.
\]
It follows that we can mod out Eq.~\eqref{eq:solve} by $\mathcal{Z}$ to get $\pm \pfont X^{a'} \equiv \circuit^{\dag} \pfont X^{a} \circuit \pmod{\mathcal{Z}}$. We compute $\circuit^{\dag} \pfont X^{a} \circuit$ by commuting $\pfont X^{a}$ backwards through $\circuit$, one gate at a time. This gives us a Pauli with an irrelevant $\pfont Z$-component (and sign) times $\pfont X^{a'}$. Finally, we push $\pfont X^{a'}$ \emph{forwards} through $\circuit$ to obtain $\circuit \pfont{X}^{a'} \circuit^{\dag}$ which is the desired stabilizer $Q(a)$ up to an overall $\pm 1$ sign.

The cost of the algorithm is in commuting $\pfont X^{a}$ and $\pfont X^{a'}$ through the circuit, which involves updating a $2n_t$-bit string once for each of the $|E| + n_a$ gates, leading to an overall runtime of $O(|E| + n_t)$.
\end{proof}

In particular, the Pauli $P_{\mathrm{cor}}$ that corrects the output $y \in \{0,1\}^{n_t}$ of the sampling routine is computed as
\begin{equation}
P_{\mathrm{cor}}=\left(U_{\mathrm{bases}}\otimes I\right) Q(y) \big(U^{\dagger}_{\mathrm{bases}}\otimes I\big).
\label{eq:conjU}
\end{equation}
Since $Q(y)$ stabilizes $\circuit|+^{n_t}\rangle$, we get that $P_{\mathrm{cor}}$ stabilizes $(U_{\mathrm{bases}} \otimes I)\circuit|+^{n_t}\rangle$, that is, $P_{\mathrm{cor}}$ satisfies Eq.~\eqref{eq:pcor1}.  Furthermore, notice that the $\pfont X$-component of $P_{\mathrm{cor}}$ on the merge ancillas is equal to $y$ on the merge ancillas, that is, $P_{\mathrm{cor}}$ satisfies Eq.~\eqref{eq:pcor2}.  Therefore, $P_{\mathrm{cor}}$ is a valid correcting Pauli, computable in $O(|E|+n_t)$ time.

Unlike the more general postselection procedure, we have not generated a uniformly random correcting Pauli.  Of course, the measurement outcome $z \in \{0,1\}^n$ obtained in this way is still valid. So, by \Cref{claim:simplestab}, we can obtain a uniformly random sample in $O(n)$ time by choosing a random stabilizer $R$ of the state $|\psi\rangle=U_{\mathrm{bases}}|G\rangle$.  A random stabilizer $R$ of $\psi$ can be computed straightforwardly in time $O(n+|E|)$.\footnote{Note that $\psi$ has a stabilizer generator for each vertex $v$ of $G$, with Pauli weight $d_v+1$. We obtain $R$ by choosing a random subset of vertices and multiplying the corresponding generators, which takes time $O(\sum_{v\in V} (d_v+1))=O(n+|E|)$.}

\section{Improved algorithm for planar graphs}
\label{sec:plan}

In this section we describe an improved algorithm for the graph state simulation problem in the special case of planar graphs. In particular, we prove \Cref{thm:planargss} and its slight generalization \Cref{thm:coarse}, restated below.

\planargss*

\coarse*

The above results are obtained as a consequence of our more general algorithm for graph state simulation (Theorem \ref{thm:sumofcubes}). That result allows us to solve the graph state simulation problem for $G$, given a tree decomposition $T$ of $G$, using a runtime $\widetilde{O}(\|T\|_\omega^{\omega})$. It is known that planar graphs have treewidth $t=O(\sqrt{n})$, and there are linear-time algorithms to construct a width-$O(\sqrt{n})$ tree decomposition of a planar graph. This fact, combined with Theorem \ref{thm:sumofcubes}, leads directly to a worst-case $\widetilde{O}(nt^{\omega-1})=\widetilde{O}(n^{(\omega+1)/2})$ algorithm for planar graph state simulation. In this section we present a more careful analysis which shows that on the tree decompositions that we construct (recursively), the algorithm from Theorem \ref{thm:sumofcubes} actually runs in $O(n^{w'/2})$ for any $w'>w$. The intuition is that any subgraph of a planar graph is again planar, so an $m$-vertex subgraph of an $n$-vertex graph has treewidth $O(\sqrt{m})$, not $O(\sqrt{n})$. Thus, we expect bags to get smaller as we descend the tree. This leads to a better upper bound on the runtime $\widetilde{O}(\|T\|_\omega^{\omega})$.

We shall use tree decompositions which are constructed recursively from planar graph \emph{separators}.
\begin{definition}
Given a graph $G = (V,E)$, a \emph{separator} is a subset of the vertices $S \subseteq V$ which divides the remaining vertices $V \backslash S$ into two disconnected parts, $A$ and $B$.
\end{definition}
It is known that any planar graph has a separator of size $|S|=O(\sqrt{n})$, and Lipton and Tarjan showed how to find one in linear time. 
\begin{theorem}[Lipton and Tarjan \cite{Lipton1979}]
Given an $n$-vertex planar graph $G$, there is an $O(n)$-time deterministic algorithm to compute a separator $S$ of size $|S|=\beta \sqrt{n}$ that divides the remaining vertices into two disconnected parts $A$ and $B$, of size $|A|, |B|\leq \alpha n$, where $\alpha = \frac{2}{3}$ and $\beta = 2 \sqrt{2}$. 
\end{theorem}
 It is well known that one can compute a tree decomposition of width $O(\sqrt{n})$ using a runtime $O(n \log n)$ by recursively applying these separators, as in \Cref{alg:separator_td} (\textsc{ComputeTD}). This can be improved to linear time with clever data structures \cite{goodrich:1995}, though we will not use this improvement here.

The rest of this section is devoted to analyzing the runtime of \textsc{ComputeTD}, the properties of the tree it creates, and the performance of our graph state simulation algorithm when given this tree. 

\begin{algorithm}
	\caption{Separator to Tree Decomposition Algorithm\label{alg:separator_td}. \newline Makes use of a subroutine \textsc{Separator} which takes as input a planar graph and outputs a planar separator with parameters $\alpha, \beta$.}
	\begin{algorithmic}[1]
		\Require{Graph $G$, subset $U \subseteq V(G)$}
		\Ensure{A tree decomposition of $G$ with $U$ contained in the root bag}
		\Statex
		\Let{$n_0$}{$\left\lceil \frac{\beta^2}{(1-\alpha)^2}\right \rceil$}
		\Function{ComputeTD}{$G$, $U$}
		\If{$|V(G)| \leq n_0$}
		\State \Return{introduce node $V(G)$ below a forget node $U$}
		\Else
		\Let{$(A,S,B)$}{\Call{Separator}{$G$}}
		\Let{$T_A$}{\Call{ComputeTD}{$G|_{A \cup S}$, $S \cup (A \cap U)$}}
		\Let{$T_B$}{\Call{ComputeTD}{$G|_{B \cup S}$, $S \cup (B \cap U)$}}
		\State \Return{forget node $U$ above a merge node $S \cup U$ with children $T_A$ and $T_B$}
		\EndIf
		\EndFunction
	\end{algorithmic}
\end{algorithm}

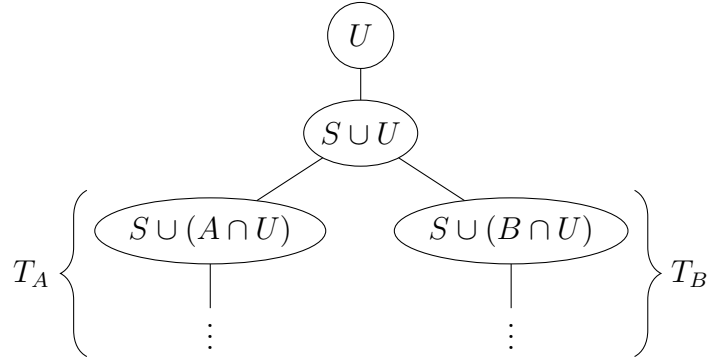
\begin{figure}
\tikzset{
  treenode/.style = {align=center, inner sep=0pt, text centered, font=\sffamily},
  bag/.style = {draw, ellipse, minimum width=25pt, minimum height=25pt, inner sep=0pt, outer sep=0pt},
  etc/.style = {circle, inner sep=0pt, outer sep=0pt},
}
\centering
\begin{tikzpicture}[-,>=stealth',level/.style={sibling distance = 4cm,level distance = 1.3cm}]
\node[bag] (U) {$U$}
child { node[bag]  (SU) {$S\cup U$} 
	child { node[bag] (SAU) {$S\cup(A\cap U)$} 
		child { node[etc] (dotsa) {$\vdots$} }
		}
	child { node[bag] (SBU) {$S\cup(B\cap U)$} 
		child { node[etc] (dotsb) {$\vdots$} }
		}
	}
;

\draw [decorate,decoration={brace,amplitude=10pt,mirror,raise=4pt}]
([shift={(-1.5,0.1)} ] SAU.north) -- ([shift={(-1.5,-0.1)}] dotsa.south) node [midway, xshift=-25pt] {$T_A$};

\draw [decorate,decoration={brace,amplitude=10pt,mirror,raise=4pt}]
([shift={(+1.5,-0.1)}] dotsb.south) -- ([shift={(+1.5,0.1)} ] SBU.north) node [midway,xshift=25pt] {$T_B$};

\end{tikzpicture}
\caption{The tree decomposition constructed by \textsc{ComputeTD}(G,U) in the case where $|V(G)|\geq n_0$. \label{fig:recursivecase}}
\end{figure}

We begin by verifying that \textsc{ComputeTD} is correct, and upper bounding the width of its output.
\begin{lemma}
Given a planar graph $G$ and a subset $U$ of its vertices, the tree $T$ computed by $\textsc{ComputeTD}(G, U)$ is a tree decomposition of $G$ with root bag $U$ and width $O(|U| + \sqrt{n})$. 
\label{thm:tdsize}
\end{lemma}
\begin{proof}
We claim by induction that $\textsc{ComputeTD}(G,U)$ produces a valid tree decomposition where the root bag is $U$. We leave the base case for the reader to check. In the recursive case, the tree decomposition has the form depicted in \Cref{fig:recursivecase} and we see that 
\begin{itemize}
    \item every vertex appears in some bag, since every vertex belongs to $A \cup S$ or $B \cup S$, and by induction the two subtrees contain those vertices,
    \item every edge appears in some bag, since all the edges in $A \cup S$ appear in $T_A$, all the edges in $B \cup S$ appear in $T_B$, and there are no edges between $A$ and $B$ because of the separator property,
    \item nodes containing $v \in V$ form a connected component because
    \begin{itemize}
        \item they are connected in $T_A$ and $T_B$, 
        \item if $v$ appears in both $T_A$ and $T_B$, then it follows that $v \in S$ and appears in the merge node above $T_A$ and $T_B$, and
        \item if $v$ belongs to the merge node $S \cup U$, then it also appears at the root of $T_A$ (or $T_B$) if it appears in that subtree at all.  
    \end{itemize}
\end{itemize}
It is also clear by construction that $U$ is the root node of the new tree, finishing the induction. 

Notice that the recursive calls are on graphs $G|_{A\cup S}$ and $G|_{B\cup S}$ of size 
$$|A \cup S|, |B \cup S| \leq \alpha n + \beta \sqrt{n} \leq \gamma n,$$
for some $\gamma$ such that
\begin{equation}
\gamma < \alpha + \frac{\beta}{\sqrt{n_0}} \leq 1,
\label{eq:epschoice}
\end{equation}
where $n_0 = \frac{\beta^2}{(1-\alpha)^2}$ (as specified in the algorithm) guarantees $\gamma < 1$. It follows that the graph input to a recursive call at depth $\ell$ will have at most $\gamma^{\ell} n$ vertices.

The bag size can only increase from parent to child by at most  $|S|$, so a bag in the $\ell$th level of recursion has size which has increased by at most the sizes of $\ell$ planar separators. Since the size of each planar separator is at most $\beta$ times the square root of the number of vertices in the graph, the number of elements gained from $\ell$ levels of separators is therefore at most
$$
|\text{Bag} \backslash U| \leq \sum_{j=0}^{\ell} \beta \sqrt{n} \gamma^{j/2} \leq \beta \sqrt{n} \sum_{j=0}^{\infty} \gamma^{j/2} = \frac{\beta \sqrt{n}}{1-\sqrt{\gamma}} = O(\sqrt{n}).
$$
We conclude that every bag has at most $O(|U| + \sqrt{n})$ vertices. 
\end{proof}

We now bound the runtime $\widetilde{O}(\| T \|_{\omega}^{\omega})$ which appears in \Cref{thm:sumofcubes}, for tree decompositions $T$ produced by \textsc{ComputeTD}. 
\begin{lemma}
Let $G$ be a planar graph with $n$ vertices, let $U\subseteq V(G)$, and let $T$ be the tree decomposition computed by $\textsc{ComputeTD}(G, U)$. For all $c > 2$ we have
$$
\| T \|_{c}^{c} = \widetilde{O}(n^{c/2} + |U|^c).
$$
\label{lem:tdnorm}
\end{lemma}
\begin{proof}
It is a property of $p$-norms that $\| T \|_{p} \leq \| T \|_{q}$ whenever $p \geq q > 0$. In particular, since $c > 2$ we have $\| T \|_{c} \leq \| T \|_2$. It suffices to show that $\| T \|_{2}^{2} = \widetilde{O}((\sqrt{n} + |U|)^2)$, because then 
$$
\| T \|_{c}^{c} \leq \| T \|_{2}^{c} = \left( \| T \|_{2}^{2} \right)^{c/2} = \widetilde{O}((\sqrt{n} + |U|)^{c}) = \widetilde{O}(n^{c/2} + |U|^c),
$$
as desired.

The crux of the proof is to reduce to the case where $U = \varnothing$. That is, consider the tree $T' \gets \textsc{ComputeTD}(G, \varnothing)$ that we get by passing $\varnothing$ instead of $U$ as the second argument. The tree structure $T$ and $T'$ will be identical except for the content of the nodes, since they will get the same separators, recurse on the same graphs, and stop at the same base cases (in other words, the control flow of \textsc{ComputeTD} does not depend on $U$). However, each node in $T'$ is contained in the corresponding node of $T$, i.e., some vertices of nodes in $T$ are missing from some nodes $T'$. We say $T$ and $T'$ \emph{disagree} on these vertices at these nodes.

We will upper bound $\| T \|_2^2 - \| T' \|_2^2$ as
\begin{align*}
    \| T \|_2^{2} - \| T' \|_2^{2} = \sum_{i} | B_i|^2 - |B_i'|^2 &= \sum_{i} \left( |B_i| - |B_i'| \right) \left( |B_i| + |B_i'| \right) \\
    &\leq \left( \sum_{i} |B_i| - |B_i'| \right) (\max_{j} |B_j| + \max_{j} |B_j'|).
\end{align*}
Recall from \Cref{thm:tdsize} that the size of each bag is at most $O(|U| + \sqrt{n})$, which upper bounds the second factor in the product above. It remains to bound the first factor (essentially the total number of missing vertices across all nodes), and use the resulting bound on $\| T \|_{2}^{2} - \| T' \|_{2}^{2}$ to draw a conclusion about $\| T \|_{2}^{2}$.

Note that $T$ and $T'$ will disagree on a vertex $u \in U$ at the root, and continue to disagree as we descend the tree until either
\begin{itemize}
    \item the separator divides $u$ from $T$, i.e., the current subtree corresponds to a subgraph of $G$ that does not include $u$ anymore, or
    \item $u$ is in the separator at some level of the recursion, and since $T$ and $T'$ have the same separators, it will appear in both trees.
\end{itemize}
In either case, if $T$ and $T'$ agree on a vertex $u$ at node $X$, then they also agree on $u$ in all descendants of $X$. Furthermore, $T$ and $T'$ cannot disagree about $u$ in both children, $Y_1$ and $Y_2$, of some node $X$. Either $u$ belongs to the separator at node $X$, and thus appears in both children in both $T$ and $T'$, \emph{or} $u$ is in one half of the partition and thus absent from the other half in both $T$ and $T'$. The take away is that $u \in U$ can be missing from at most one node per level of the tree. Otherwise, if $u$ is missing from nodes $Z_1$ and $Z_2$ in the same level, then it is also missing from their ancestors (since if $T$ and $T'$ agree on an ancestor of $Z_i$, they agree on all its descendants, including $Z_i$), and in particular from their least common ancestor $Z$ and its two children. Since $T$ and $T'$ cannot disagree on both children, this is a contradiction.

We conclude that any particular $u \in U$ is missing from at most $O(\text{depth}) = O(\log n)$ nodes, and so there are a total of $O(|U| \log n)$ missing nodes across the whole tree. That is, $\sum_{i} |B_i| - |B_i'| = O(|U| \log n)$, and thus 
$$
\| T \|_{2}^{2} - \| T' \|_{2}^{2} = O(|U| (|U| + \sqrt{n}) \log n).
$$
We will now show that $\| T' \|_2^{2} = \widetilde{O}(n)$ to complete the proof.

Let $f(n)$ be the maximum value of $\| T' \|_{2}^{2}$ for a tree $T' \gets \textsc{ComputeTD}(G, \varnothing)$ over all graphs $G$ on at most $n$ vertices. In the recursive case (where $n > n_0$), we have two recursive calls which are unfortunately \emph{not} of the form $\textsc{ComputeTD}(G', \varnothing)$; they have non-empty second arguments. However, by the same argument as above, setting the second argument from $U'$ to $\varnothing$ changes the squared $2$-norm by at most $O(|U'| (|U'|+ \sqrt{n}) \log n) = O(n \log n)$, since here $|U'| \leq |S| = O(\sqrt{n})$. It follows that we can upper bound $f(n)$ by $f(n_1) + O(n \log n)$ and $f(n_2) + O(n \log n)$ for the two subtrees, plus $O(|S|^2) = O(n)$ for the root node, giving
\begin{equation}\label{eq:source_of_gamma}
f(n) \leq f(n_1) + f(n_2) + O(n \log n),
\end{equation}
where $n_1, n_2 \leq \alpha n + \beta \sqrt{n}$ and $n_1 + n_2 \leq n + \beta \sqrt{n}$. We show by induction that $f(n) \leq k n \log^2 n$ for all $n \geq 2$, where the constant $k$ is to be determined.\footnote{The base of the logarithm is irrelevant since it can be absorbed into $k$.} As long as $n \geq 2/\alpha$, both parts will be size $\geq \alpha n \geq 2$, so we can apply the induction hypothesis. Let us take $k$ large enough that the induction hypothesis holds for the remainder, $2 \leq n < 2/\alpha$. Then taking $\gamma$ to be the constant from the $O(n\log n)$ term in \Cref{eq:source_of_gamma}, we have
\begin{align*}
f(n) &\leq f(n_1) + f(n_2) + \gamma n \log n \\
&\leq k (n_1 \log^2 n_1 + n_2 \log^2 n_2) + \gamma n \log n \\
&\leq k (n_1 + n_2) \log^2 \max(n_1, n_2) + \gamma n \log n \\
&\leq k (n + \beta \sqrt{n}) \log^2 \max(n_1, n_2) + \gamma n \log n \\
&= kn \log^2 \max(n_1, n_2) + k \beta \sqrt{n} \log^2 \max(n_1, n_2) + \gamma n \log n.
\end{align*}
Clearly $\max(n_1, n_2) \leq n$, and $\sqrt{n} \log^2 n = o(n \log n)$, so for any $\beta > 0$ and $\epsilon > 0$ we can take $N$ large enough that 
$$
\beta \sqrt{n} \log^2 \max(n_1, n_2) \leq \epsilon n \log n
$$
for all $n > N$. For the other term, we note that $n_1, n_2 \leq \alpha n + \beta \sqrt{n} \leq \alpha' n$ for some constant $\alpha < \alpha' < 1$ when $n$ is sufficiently large. 
\begin{align*}
f(n) &\leq kn (\log n)(\log n + \log \alpha') + (\gamma + k\epsilon) n \log n \\
&= k n\log^2 n + (\gamma + k\epsilon + k\log \alpha') n \log n
\end{align*}
We wish to show this is at most $k n \log^2 n$. At this point we fix $\alpha'$, assume $N$ is large enough that $\epsilon < \log(1/\alpha')$, and then take any $k > \frac{\gamma}{\log(1/\alpha)-\epsilon}$ so that the second term is negative, and hence $f(n) \leq k n \log^2 n$. This completes the induction, and we conclude that $f(n) = O(n \log^2 n)$.
\end{proof}

The running time for the algorithm follows immediately, as long as \textsc{ComputeTD} is not too expensive. Let us bound that next.
\begin{lemma}
Let $G$ be a planar graph with $n$ vertices.  $\textsc{ComputeTD}(G, U)$ runs in time $O(n \log n)$ given that $\textsc{Separator}$ runs in linear time. 
\label{lem:tdruntime}
\end{lemma}
\begin{proof}
Let $F(n)$ be the worst-case running time of $\textsc{ComputeTD}$ on graphs of at most $n$ vertices. Using the recursive structure of the algorithm, we see that for sufficiently large $n$, 
$$
f(n) \leq f(n_1) + f(n_2) + O(n)
$$
for some $n_1, n_2$ such that $n_1 + n_2 \leq n + \beta \sqrt{n}$ and $n_1, n_2 \leq \alpha n + \beta \sqrt{n}$. Similar to the previous lemma, an inductive argument shows this recurrence is solved by $f(n) = O(n \log n)$.
\end{proof}

We are now ready to prove \Cref{thm:planargss}.
\begin{proof}[Proof of \Cref{thm:planargss}]
We will apply the algorithm from \Cref{thm:sumofcubes} to the tree $T$ output by $\textsc{ComputeTD}(G, \varnothing)$. The graph state simulation algorithm runs in $\widetilde{O}(\| T \|^{\omega}_{\omega})$ by \Cref{thm:sumofcubes}, and $\|T \|_{\omega}^{\omega}$ is $\widetilde{O}(n^{\omega/2})$ by \Cref{lem:tdnorm}. The cost of constructing the tree decomposition is $O(n \log n)$ by \Cref{lem:tdruntime}. Thus, the total running time is $\widetilde{O}(n^{\omega/2})$.
\end{proof}

We also immediately get \Cref{thm:coarse}.

\begin{proof}[Proof of \Cref{thm:coarse}]
Let $\varphi \colon G' \to G$ be the mapping witnessing that $G'$ is $r$-coarse-grained planar. That is, $G$ is planar and $\varphi$ maps at most $r$ vertices in $G'$ to one vertex in $G$. To solve an instance of the graph state simulation problem for $G'$, we run the following steps:
\begin{enumerate}
    \item Construct $T_{G} \gets \textsc{ComputeTD}(G, \varnothing)$, using the planar tree decomposition.
    \item Take the preimage $T_{G'} \gets \varphi^{-1}(T_G)$. That is, let $T_{G'}$ have the same tree structure, and replace a bag $B \subseteq V(G)$ with the preimage set $\varphi^{-1}(B) := \{ v \in V(G') : \varphi(v) \in B \}$.
    \item\label{step:simulation} Run the graph state simulation algorithm from \Cref{thm:sumofcubes} for $G'$ with tree decomposition $T_{G'}$.
\end{enumerate}
Note that $T_{G'}$ is a tree decomposition for $G'$ by \Cref{lem:tree_decomp_homo}, and therefore the algorithm is correct.

The running time will be dominated by Step~\ref{step:simulation} since constructing $T_G$ is $O(n \log n)$ time by \Cref{lem:tdruntime}, and computing preimages is $O(n)$ time. Recall that the graph state simulation algorithm runs in time $O(\| T_{G'} \|_{\omega}^{\omega})$ by \Cref{thm:sumofcubes}. Each bag in $T_{G'}$ is at most $r$ times bigger than $T_G$, so it is clear that $\| T_{G'} \|_{\omega}^{\omega} \leq r^{\omega} \| T_G \|_{\omega}^{\omega}$. Finally, $\| T_G \|_{\omega}^{\omega} = \widetilde{O}(n^{\omega/2})$ by \Cref{lem:tdnorm}. Putting these facts together, we see that the sampling algorithm is $\widetilde{O}(n^{\omega/2} r^{\omega})$ time. 
\end{proof}

For example, consider the $k \times k \times k$ cubic grid $G'$. There is clearly a $k$-coarse graining onto the $k \times k$ grid $G$. Since the $k \times k$ grid is only $n = k^2$ qubits, sampling only takes $\widetilde{O}( (k^2)^{\omega/2} k^{\omega}) = \widetilde{O}(k^{2 \omega})$ time. This matches (up to log factors) the natural analogue of the recursive 2D grid algorithm, and beats the $\widetilde{O}(k^{3 \omega})$ runtime of direct Clifford simulation.

\section{Stabilizer operations in matrix multiplication time}
\label{sec:matrix_mult}

The goal of this section is to show that fast matrix multiplication techniques can be used in a wide range of stabilizer operations. These speedups are shown in \Cref{table:clifford_asymptotics}, where $\omega < 2.372859$ \cite{alman+williams:2020} is the matrix multiplication constant.  In particular, we will show that given the tableaux for Clifford operations on $n$ and $m$ qubits, the tableau for their composition can be computed in $O(\min(n m^{\omega-1}, m n^{\omega-1}))$ time.  Furthermore, given the tableau for a Clifford state on $n$ qubits, we simulate a measurement on $k$ qubits in $O(k^{\omega - 2} n^2 \log (2n/k))$ time. 

Most of the results in this section follow from the work of Dehaene and De Moor \cite{dehaene+demoor:2003} which describes the composition of Clifford operations as the product of binary matrices. We begin with a short summary of their techniques.  First, recall that every Pauli $P$ can be represented uniquely as\footnote{For the purposes of this section, we will use a separate font to write $\pfont{X}$, $\pfont{Y}$, and $\pfont{Z}$ Pauli operators to avoid confusion with the many other matrices that will appear.}
$$
i^\alpha (-1)^\beta \pfont{X}^a \pfont Z^b
$$
for $\alpha, \beta \in \{0,1\}$ and $a, b \in \{0,1\}^n$.

\begin{lemma}[Dehaene, De Moor \cite{dehaene+demoor:2003}]
\label{lem:pauli_mult}
Let $P_1 = i^{\alpha_1} (-1)^{\beta_1} \pfont X^{a_1} \pfont Z^{b_1}$ and $P_2 = i^{\alpha_2} (-1)^{\beta_2} \pfont X^{a_2} \pfont Z^{b_2}$.  Then $P_1 P_2 = i^{\alpha_{12}}(-1)^{\beta_{12}} \pfont X^{a_{12}} \pfont Z^{b_{12}}$ where
\begin{center}
\vspace{-20pt}
\begin{minipage}{.35\linewidth}
\begin{align*}
\alpha_{12} &= \alpha_1 + \alpha_2 \\
\beta_{12} &= \beta_1 + \beta_2 + \alpha_1\alpha_2 + b_1 \cdot a_2
\end{align*}
\end{minipage}
\begin{minipage}{.35\linewidth}
\begin{align*}
a_{12} &= a_1 + a_2 \\
b_{12} &= b_1 + b_2
\end{align*}
\end{minipage}
\end{center}
where all operations are over $\mathbb F_2$.
\end{lemma}
\begin{proof}
We start by noticing that $i^{\alpha_1 + \alpha_2} = i^{\alpha_1 \oplus \alpha_2} (-1)^{\alpha_1\alpha_2}$ when $\alpha_1, \alpha_2 \in \{0,1\}$.  Then
\begin{align*}
P_1 P_2 &= i^{\alpha_1 \oplus \alpha_2} (-1)^{\beta_1 + \beta_2 + \alpha_1\alpha_2 } \pfont X^{a_1} \pfont Z^{b_1} \pfont X^{a_2} \pfont Z^{b_2}\\
&= i^{\alpha_1 \oplus \alpha_2} (-1)^{\beta_1 + \beta_2 + \alpha_1\alpha_2 + b_1 \cdot a_2} \pfont X^{a_1}  \pfont X^{a_2} \pfont Z^{b_1} \pfont Z^{b_2}\\
&= i^{\alpha_1 \oplus \alpha_2} (-1)^{\beta_1 + \beta_2 + \alpha_1\alpha_2 + b_1 \cdot a_2} \pfont X^{a_1 \oplus a_2}  \pfont Z^{b_1 \oplus b_2},
\end{align*}
which gives the desired result.
\end{proof}

\begin{corollary}
\label{cor:multi_pauli_mult}
Let $P_i = i^{\alpha_i} (-1)^{\beta_i} \pfont X^{a_i} \pfont Z^{b_i}$ for $i \in \{1, \ldots, \ell\}$.  Then, $P_1 \cdots P_\ell = i^{\alpha} (-1)^{\beta} \pfont X^{a} \pfont Z^{b}$ where
\begin{center}
\vspace{-20pt}
\begin{minipage}{.4\linewidth}
\begin{align*}
\alpha &= \sum_{i =1}^\ell \alpha_i \\
\beta &= \sum_{i=1}^\ell \beta_i + \sum_{i < j}^\ell \alpha_i \alpha_j + \sum_{i < j}^\ell b_i \cdot a_j
\end{align*}
\end{minipage}
\begin{minipage}{.35\linewidth}
\begin{align*}
a &= \sum_{i=1}^\ell a_i \\
b &= \sum_{i=1}^\ell b_i
\end{align*}
\end{minipage}
\end{center}
where all operations are over $\mathbb F_2$.
\end{corollary}
\begin{proof}
Repeatedly apply \Cref{lem:pauli_mult}.
\end{proof}

Let us now fix our tableau representation for an $n$-qubit Clifford circuit $Q$.\footnote{We will mostly adopt the conventions of the Aaronson-Gottesman tableau representation \cite{aaronson+gottesman:2004}, but will keep track of the phase bits according to Dehaene and De Moor \cite{dehaene+demoor:2003}.}  The tableau is simply a list of the images of $\pfont X^{e_j}$ and $\pfont Z^{e_j}$ under conjugation by $Q$ for all $j$, where $e_j \in \{0,1\}^n$ is the $n$-bit vector with a 1 in position $j$ and 0's everywhere else.  Let $Q \pfont X^{e_j} Q^\dag = i^{\alpha_j} (-1)^{\beta_j} \pfont X^{a_j} \pfont Z^{b_j}$ and $Q \pfont Z^{e_j} Q^\dag = i^{\delta_j} (-1)^{\epsilon_j} \pfont X^{c_j} \pfont Z^{d_j}$.  We arrange this conjugation information into a $2n \times 2n$ binary matrix $M$, a phase vector $p \in \{0,1\}^{2n}$, and a sign vector $s \in \{0,1\}^{2n}$:
$$
M = \left[\begin{array}{c | c}
A & B \\ \hline C & D
\end{array}\right]
\hspace{10pt}
p = \left[\begin{array}{c} \alpha \\ \hline \delta \end{array}\right]
\hspace{10pt}
s = \left[\begin{array}{c} \beta \\ \hline \epsilon \end{array}\right]
$$
where $A = \{a_{i,j}\}$, $B = \{b_{i,j}\}$, $C = \{c_{i,j}\}$, and $D = \{d_{i,j}\}$ are the $n \times n$ binary matrices whose elements correspond to the Pauli decompositions. That is, $a_i$ is the $i$th row of $A$, and so on.
\begin{fact}[e.g., \cite{dehaene+demoor:2003}]
\label{fact:tableau_facts}
$M$, $p$, and $s$ correspond to a valid Clifford unitary iff
\begin{itemize}
\item $M$ is an element of the binary symplectic group:  $M^T \Omega M = \Omega$ where $\Omega = \left(\begin{smallmatrix} 0 & I \\ I & 0 \end{smallmatrix}\right)$;
\item $p = \diag(M \Lambda M^T)$ where $\Lambda = \left(\begin{smallmatrix} 0 & I \\ 0 & 0 \end{smallmatrix}\right)$.
\end{itemize}
\end{fact}
As you can see from the above fact, we don't actually need to store $p$ since it is encoded into $M$, but it will helpful to keep it around anyway.  

Let us now turn to the computational question of how much time is required to perform various Clifford operations.  Since an $n$-qubit state is represented by $O(n^2)$ many bits, it will often be computationally inefficient to write out an entirely new tableau resulting from a Clifford operation.  For this reason, we will adopt the conventional view of Clifford operations as updating the existing tableau (e.g. modifying entries and appending rows/columns). Therefore, operations such as composing a single-qubit gate with an $n$-qubit Clifford circuit only require $O(n)$ time if the tableau for the $n$-qubit Clifford is given as input (see \Cref{sec:composition} for a detailed explanation of composition).

We conclude this section with a simple fact that tensor products of Clifford operations are represented by the direct sum of their tableau:
\begin{fact}
Let $(M_1, p_1, s_1)$ and $(M_2, p_2, s_2)$ be the tableaux for Clifford operations $Q_1$ and $Q_2$, respectively.  Then $Q_1 \otimes Q_2$ is represented by tableau
$$
M_{12} = \left[\begin{array}{ c c | c c }
A_{1} & 0 & B_{1} & 0 \\
0  & A_{2}  & 0 & B_{2} \\ \hline
C_{1} & 0 & D_{1} & 0 \\
0 & C_{2} & 0 & D_{2}
\end{array}\right]
\hspace{20pt}
p_{12} = \left[\begin{array}{c} \alpha_1 \\ \alpha_2 \\ \hline \delta_1 \\ \delta_2 \end{array}\right]
\hspace{20pt}
s_{12} = \left[\begin{array}{c} \beta_1 \\ \beta_2 \\ \hline \epsilon_1 \\ \epsilon_2 \end{array}\right]
$$
\end{fact}
\begin{proof}
For Pauli operators $P_1$ and $P_2$, we have 
\[
(Q_1 \otimes Q_2) (P_1 \otimes P_2) (Q_1^\dag \otimes Q_2^\dag) = 
Q_1 P_1 Q_1^\dag \otimes Q_2 P_2 Q_2^\dag.
\]
The fact follows by setting $P_1 \otimes P_2$ to be $\pfont X^{e_j}$ (for the destabilizers) and $\pfont Z^{e_j}$ (for the stabilizers). For example, the first row of the tableau comes from setting $P_1 = \pfont{X}^{e_1}$ and $P_2 = I$.
\end{proof}
Therefore, we get that appending a single-qubit stabilizer state to an $n$-qubit stabilizer state takes time $O(n)$ and appending an $\ell$-qubit state takes time $O((n+\ell)^2)$.

\subsection{Composition}
\label{sec:composition}
This section is dedicated to showing fast matrix multiplication speedups for the composition of Clifford operations.  Before we do so, let us first describe the action of a single Clifford operation on a Pauli string.  For any square matrix $A$, we will define $\lows(A)$ to be the strictly lower triangular matrix of $A$ (i.e., with the diagonal and above diagonal entries replaced with $0$).

\begin{lemma}[Dehaene, De Moor \cite{dehaene+demoor:2003}]
\label{lem:tableau_with_pauli}
Let $Q$ be an $n$-qubit Clifford operation represented by tableau $(M, p, s)$.  Let $P := i^\alpha (-1)^\beta \pfont X^a \pfont Z^b$ be a Pauli where $a \in \{0,1\}^n$ and $b \in \{0,1\}^n$ form the $2n$-bit row vector $c = [ a \;\; b ]$.  We have $Q P Q^\dag = i^{\alpha'} (-1)^{\beta'} X^{a'} Z^{b'}$ with $c' = [ a' \;\; b'] = c M$, $\alpha' = \alpha + c \cdot p$, and
$$
\beta' = \beta + \alpha c \cdot p + c \cdot s + c \lows(p p^T + M \Lambda M^T) c^T
$$
\end{lemma}
\begin{proof}

Once again let $M = \left[\begin{smallmatrix} A & B \\ C & D \end{smallmatrix}\right]$.  We will write $M_i$ for the $i$th row of matrix $M$ (and similarly for the submatrices $A$, $B$, $C$, and $D$). Notice that the inner products between the $\pfont X$ and $\pfont Z$ parts of $M$ can be expressed with the $\Lambda = \left(\begin{smallmatrix} 0 & I \\ 0 & 0 \end{smallmatrix}\right)$ operator.  For example, $A_i \cdot B_j = M_i \Lambda M_j^T$ for $i, j \in \{1, \ldots, n\}$.  Ignoring the phase for the moment, we have
\begin{align*}
Q \pfont X^a \pfont Z^b Q^\dag &= \prod_{j : a_j = 1} Q \pfont X^{e_j} Q^\dag \prod_{j : b_j = 1} Q \pfont Z^{e_j} Q^\dag \\
&= \prod_{j:a_j = 1} i^{p_j} (-1)^{s_j} \pfont X^{A_j} \pfont Z^{B_j} \prod_{j:b_j = 1} i^{p_{j+n}} (-1)^{s_{j+n}} \pfont X^{C_j} \pfont Z^{D_j} \\
&= i^{\alpha'} (-1)^{\beta'} \pfont X^{a'} \pfont Z^{b'}.
\end{align*}
That is, $Q \pfont X^{a} \pfont Z^{b} Q^{\dag}$ is simply the product of the Pauli elements indexed by the rows of the tableau, so by \Cref{cor:multi_pauli_mult} (there are $\ell := 2n$ such Pauli operators) we get
\begin{align*}
\alpha' &= \sum_{j = 1}^n a_j p_j + \sum_{j = 1}^n b_j p_{j+n} = \sum_{j = 1}^n c_j p_j + \sum_{j = 1}^n c_{j+n} p_{j+n} = c \cdot p \\
\beta' &= \sum_{j = 1}^{2n} c_j s_j + \sum_{i < j}^{2n} c_i c_j p_i p_j + \sum_{i < j}^{2n} c_i c_j M_i \Lambda M_j^T \\
&= c \cdot s + c \lows(p p^T) c^T + c \lows(M \Lambda M^T) c^T
\end{align*}
and $c' =[ a' \;\; b' ] = c M$.  Adding the global phase back in gives the lemma.
\end{proof}

\begin{theorem}[Tableau composition \cite{dehaene+demoor:2003}]
\label{thm:tableaux_composition}
Let $Q_1, Q_2$ be $n$-qubit Clifford unitaries represented by tableaux $(M_1$, $p_1$, $s_1)$ and $(M_2$, $p_2$, $s_2)$ respectively.  The composition $Q_2 Q_1$ has tableau $(M_{12}, p_{12}, s_{12})$ with $M_{12} = M_1 M_2$, $p_{12} = p_1 + M_1 p_2$, and 
$$
s_{12} = s_1 + p_1 \circ M_1 p_2 + M_1 s_2 + \diag (M_1 \lows(p_2 p_2^T + M_2 \Lambda M_2^T) M_1^T)
$$
where $\circ$ denotes the Hadamard product\footnote{The Hadamard product of two $n \times m$ matrices $A = \{a_{i,j} \}_{i,j}$ and $B = \{b_{i,j} \}_{i,j}$ is their entrywise product $A \circ B = \{ a_{i,j} b_{i,j} \}_{i,j}$.} and ``diag'' selects the elements along the diagonal.
\end{theorem}
\begin{proof}
Notice that $Q_2 Q_1 \pfont X^{e_j} Q_1^\dag Q_2^\dag = Q_2 (i^{(p_1)_j} (-1)^{(s_1)_j} \pfont X^{(A_1)_j} \pfont Z^{(B_1)_j}) Q_2^\dag$, and similarly for the images of $\pfont Z^{e_j}$.  Thus, we arrive at the theorem statement by applying \Cref{lem:tableau_with_pauli} with $Q_2$ and every Pauli specified by the rows of the tableau for $Q_1$.
\end{proof}

From the tableaux composition theorem, it should be clear that the tableau for the product of two Clifford operations on $n$-qubits can be computed in time $O(n^\omega)$.  Furthermore, when the two Clifford operations are of different sizes, it may be possible to obtain a significant savings.  We know from standard techniques that if $Q_2$ is a one or two-qubit gate, then the tableau for $Q_1$ can be updated in $O(n)$ time.

The theorem below interpolates between these extremes. 

Concretely, composing an $n$-qubit operation $Q_1$ with an $m$-qubit operation $Q_2$ for $m \le n$, means computing the product $(Q_2 \otimes I^{n-m}) Q_1$, and similarly $Q_2 (Q_1  \otimes I^{m-n})$ when $m > n$. 

\begin{theorem}
\label{thm:fast_composition}
Let $Q_1, Q_2$ be $n$ and $m$-qubit Clifford operators, respectively, given by their tableaux.  The tableau of their product can be computed in time $O(\min(m^{\omega-1}n, n^{\omega-1}m))$.
\end{theorem}
\begin{proof}
When $m = \Theta(n)$, the theorem is trivial.  Let's consider the case where $m = o(n)$.  Without loss of generality, let's assume that $Q_2$ is applied to the first $m$ qubits of $Q_1$.  On $n$ qubits, the tableau of $Q_2$ is
$$
M_2 = \left[\begin{array}{ c c | c c }
A & 0 & B & 0 \\
0 & I & 0 & 0 \\ \hline
C & 0 & D & 0 \\
0 & 0 & 0 & I
\end{array}\right]
\hspace{20pt}
p_2 = \left[\begin{array}{c} \alpha \\ 0 \\ \hline \delta \\ 0 \end{array}\right]
\hspace{20pt}
s_2 = \left[\begin{array}{c} \beta \\ 0 \\ \hline \epsilon \\ 0 \end{array}\right]
$$
where $M_2 \in \{0,1\}^{2n \times 2n}$ and $A, B, C, D \in \{0,1\}^{m \times m}$.  Matching the dimensions of the matrices above, let's write the tableau of $Q_1$ as follows:
$$
M_1 = \left[\begin{array}{ c c | c c }
A_{11} & A_{12} & B_{11} & B_{12} \\
A_{21}  & A_{22}  & B_{21} & B_{22} \\ \hline
C_{11} & C_{12} & D_{11} & D_{12} \\
C_{21} & C_{22} & D_{21} & D_{22}
\end{array}\right]
\hspace{20pt}
p_1 = \left[\begin{array}{c} \alpha_1 \\ \alpha_2 \\ \hline \delta_1 \\ \delta_2 \end{array}\right]
\hspace{20pt}
s_1 = \left[\begin{array}{c} \beta_1 \\ \beta_2 \\ \hline \epsilon_1 \\ \epsilon_2 \end{array}\right]
$$
Therefore, according to \Cref{thm:tableaux_composition}, we have that the matrix $M_{12}$ for the tableaux of $Q_2 Q_1$ is
$$
\left[\begin{array}{ c c | c c }
A_{11} & A_{12} & B_{11} & B_{12} \\
A_{21}  & A_{22}  & B_{21} & B_{22} \\ \hline
C_{11} & C_{12} & D_{11} & D_{12} \\
C_{21} & C_{22} & D_{21} & D_{22}
\end{array}\right]
\left[\begin{array}{ c c | c c }
A & 0 & B & 0 \\
0 & I & 0 & 0 \\ \hline
C & 0 & D & 0 \\
0 & 0 & 0 & I
\end{array}\right]
=
\left[\begin{array}{ c c | c c }
A_{11}A + B_{11}C & A_{12}& A_{11}B + B_{11}D & B_{12}  \\
A_{21}A + B_{21}C  & A_{22}  & A_{21}B+ B_{21}D & B_{22} \\ \hline
C_{11}A + D_{11}C & C_{12} & C_{11}B + D_{11}D & D_{12} \\
C_{21}A + D_{21}C & C_{22} & C_{21}D + D_{21}D & D_{22}
\end{array}\right]
$$
There are two types of matrix products we require.  Terms such as $A_{11}A + B_{11}C$ consist entirely of $m \times m$ matrix products, so can be computed in time $O(m^\omega)$.  Terms such as $A_{21}A + B_{21}C$ require multiplication of an $(n-m) \times m$ matrix by an $m \times m$ matrix.  Dividing these matrices into blocks of size $m$, we get new block matrices of dimensions roughly $n/m \times 1$ and $1 \times 1$ where multiplication of a single element (i.e., block) takes time $O(m^\omega)$.  There are $n/m$ many block multiplications in the product, so it can be computed in time $O((n/m) m^\omega)$.  Therefore, the entire procedure costs $O(n m^{\omega-1})$ time. 

The calculation for the phase vector is similar, so we omit it.  Let us now turn to the calculation of the sign vector. First, notice that the diagonal of $M_1 \lows(p_2 p_2^T + M_2 \Lambda M_2^T) M_1^T$ is equal to
$$
\diag(M_1 \lows(p_2 p_2^T) M_1^T) + \diag(M_1 \lows(M_2 \Lambda M_2^T) M_1^T),
$$
so we will focus on computing the second more complicated term.  We have
$$
M_2 \Lambda M_2^T = 
\left[\begin{array}{ c c | c c }
A & 0 & B & 0 \\
0 & I & 0 & 0 \\ \hline
C & 0 & D & 0 \\
0 & 0 & 0 & I
\end{array}\right]
\left[\begin{array}{ c c | c c }
0 & 0 & I & 0 \\
0 & 0 & 0 & I \\ \hline
0 & 0 & 0 & 0 \\
0 & 0 & 0 & 0
\end{array}\right]
\left[\begin{array}{ c c | c c }
A^T & 0 & C^T & 0 \\
0 & I & 0 & 0 \\ \hline
B^T & 0 & D^T & 0 \\
0 & 0 & 0 & I
\end{array}\right]
=
\left[\begin{array}{ c c | c c }
AB^T & 0 & AD^T & 0 \\
0 & 0 & 0 & I \\ \hline
CB^T & 0 & CD^T & 0 \\
0 & 0 & 0 & 0
\end{array}\right]
$$
from which, we get
$$
\lows(M_2 \Lambda M_2^T) = 
\left[\begin{array}{ c c | c c }
\lows(AB^T) & 0 & 0 & 0 \\
0 & 0 & 0 & 0 \\ \hline
CB^T & 0 & \lows(CD^T) & 0 \\
0 & 0 & 0 & 0
\end{array}\right]
$$
from which we can derive the condition along the diagonal---that is, we have that the diagonal of $M_1\lows(M_2 \Lambda M_2^T)M_1^T$ is equal to
\begin{align*}
&\diag(A_{11} \lows(AB^T) A_{11}^T + B_{11} CB^T A_{11}^T + B_{11} \lows(CD^T) C_{11}^T) \\
\oplus \;\; & \diag(A_{21} \lows(AB^T) A_{21}^T + B_{21} CB^T A_{21}^T + B_{21} \lows(CD^T) C_{21}^T) \\
\oplus \;\; & \diag(C_{11} \lows(AB^T) B_{11}^T + D_{11} CB^T B_{11}^T + D_{11} \lows(CD^T) D_{11}^T) \\
\oplus \;\; & \diag(C_{21} \lows(AB^T) B_{21}^T + D_{21} CB^T B_{21}^T + D_{21} \lows(CD^T) D_{21}^T)
\end{align*}
where $\oplus$ is used here to indicate direct sum.  Since $A, B, C, D$ are $m \times m$ matrices, we can explicitly compute $\lows(AB^T)$ and $\lows(CD^T)$ in $O(m^\omega)$ time.  Consider a term such as $A_{21} \lows(AB^T) A_{21}^T$ where we are multiplying matrices with dimension $(n-m) \times m$, $m \times m$, and $m \times (n-m)$.  Dividing each matrix into blocks of size $m$, we get matrices of dimensions roughly $n/m \times 1$, $1 \times 1$, and $1 \times n/m$ where multiplication of blocks takes time $O(m^\omega)$.  Since we only need to calculate the $n/m$ blocks along the diagonal of the product matrix, we can compute the entire diagonal in time $O((n/m) m^\omega)$.  Since all other terms follow in a similar manner, we can calculate the entire sign vector in time $O(n m^{\omega-1})$.  Thus, for $m = o(n)$, the entire composition takes time $O(n m^{\omega -1})$.

When $n = o(m)$, the calculation is extremely similar to the one above, so we omit it.
\end{proof}

The above theorem is more general than needed for the purposes of this paper.  Nevertheless, we can use it to show how to multiply any $n$-qubit stabilizer state by an arbitrary collection of CZ gates in time $O(n^\omega)$ as stated in \Cref{table:clifford_asymptotics}.  First, let us assume that no two CZ gates in our collection are identical since they would simply cancel.  Then, let $A$ be the $n \times n$ binary symmetric matrix whose $(i,j)$ entry is $1$ iff there is a CZ gate between qubits $i$ and $j$ in our collection.  The tableau for the entire collection of CZ gates is
$$
M = \left[\begin{array}{c | c}
I & A \\ \hline 0 & I
\end{array}\right]
\hspace{10pt}
p = \left[\begin{array}{c} 0 \\ \hline 0 \end{array}\right]
\hspace{10pt}
s = \left[\begin{array}{c} 0 \\ \hline 0 \end{array}\right]
$$
Since we can explicitly write down the tableau for the collection of CZ gates, we can straightforwardly apply \Cref{thm:fast_composition}.

\subsection{Measurement}
\label{sec:fast_measurement}

The previous sections describe tableaux as representing Clifford unitary operations.  However, tableaux are equally suitable for representing states. Given tableau $M = \left[\begin{smallmatrix} A & B \\ C & D \end{smallmatrix}\right]$, $p = \left[\begin{smallmatrix} \alpha \\ \delta \end{smallmatrix}\right]$, and $s = \left[\begin{smallmatrix} \beta \\ \epsilon \end{smallmatrix}\right]$, the stabilizer state defined by the tableau is the unique state $\ket{\psi}$ such that
$$
i^{\delta_j} (-1)^{\epsilon_j} \pfont X^{C_j} \pfont Z^{D_j} \ket{\psi} = \ket{\psi}
$$
for all $j \in \{1, \ldots, n\}$.  That is, the lower half of the tableau describes the Pauli operations that fix or ``stabilize'' the state, which is why these Pauli operations are often referred to as stabilizers.  Nevertheless, the upper half of the tableau (consisting of the ``destabilizers'' of the state $\ket{\psi}$) is still useful in speeding up measurement computations \cite{aaronson+gottesman:2004}.  In this section, we show how to combine those techniques with fast matrix multiplication.  

An important primitive for the measurement protocol is choosing new stabilizer generators for the tableau by multiplying existing stabilizer generators together. Therefore, we will need a fast way to multiply several Pauli operators at once.  We will encounter two situations in particular:
\begin{enumerate}
    \item There are several Pauli operators, but each operator has  few non-identity elements;
    \item There are few Pauli operators, but each operator has many non-identity elements.
\end{enumerate} 
The following technical lemma handles both cases.
\begin{lemma}
\label{lem:tableau_with_many_paulis}
Let $Q$ be an $n$-qubit Clifford unitary represented by tableau $(M, p, s)$. Suppose we wish to compute $Q P_j Q^\dag$ for $\ell$ Pauli operators $\{P_j\}_{j=1}^\ell$ that only act non-trivially on the first $k$ qubits. Precisely, we define $P_j := \pfont X^{[A_j \; 0]} \pfont Z^{[B_j \; 0]}$ to be an $n$-qubit Pauli where $A_j,B_j \in \{0,1\}^k$ form the rows of the $\ell \times 2n$ matrix $C = [A \;0 \mid B \; 0]$.  We have $Q P_j Q^\dag = i^{\alpha'_j} (-1)^{\beta'_j} \pfont{X}^{A_j'} \pfont{Z}^{B_j'}$ with $C' = [ A' \mid B' ] = C M$, $\alpha' = C p$, and
$$
\beta' =  C s + \diag(C \lows(p p^T + M \Lambda M^T) C^T)
$$
Furthermore, $C'$, $\alpha'$, and $\beta'$ can be computed in time
\begin{itemize}
\item $O(n^2 k^{\omega -2})$ when $\ell = n$, i.e., several Pauli operators with few non-identity elements.
\item $O(n^2 \ell^{\omega -2} \log (2n/\ell))$ when $\ell \le k = n$, i.e., few Pauli operators with many non-identity elements 
\end{itemize}
\end{lemma}
\begin{proof}
Correctness follows from \Cref{lem:tableau_with_pauli}, so we focus only on running time.  Once again, we will only show how to compute the most complicated term---$\diag(C \lows(M \Lambda M^T) C^T)$---of the sign vector $\beta'$, and the rest of the terms follow by similar (or simpler) arguments. We will write $M$ in terms of the blocks $A_{11} \in \{0,1\}^{k \times k}$, $A_{12} \in \{0,1\}^{k \times (n-k)}$, $A_{22} \in \{0,1\}^{(n-k) \times (n-k)}$, and so on as in \Cref{thm:fast_composition}.

First, let's consider the case with many Pauli operators (i.e., $\ell = n$).  Following the same algebra as in \Cref{thm:fast_composition}, we arrive at the following expression for $C \lows(M \Lambda M^T)C^T$:
$$
A \lows(A_{11} B_{11}^T + A_{12} B_{12}^T) A^T + B (C_{11} B_{11}^T + C_{12} B_{12}^T) A^T + B \lows(C_{11} D_{11}^T + C_{12} D_{12}^T) B^T.
$$ 
 To compute $\lows( \cdot )$, we need to compute a $k \times (n-k)$ by $(n-k) \times k$ matrix product.  Dividing into $k \times k$ blocks, such a multiplication takes time $O(n k^{\omega-1})$, after which we can trivially compute the lower triangle matrix in $O(k^2)$ time.  Finally, we must multiply matrices with dimensions $n \times k$, $k \times k$, and $k \times n$.  Once again, we divide into $k \times k$ blocks, to compute the entire product in $O((n/k)^2 k^\omega)$ time. This concludes the first case.

Now consider the case where the Pauli operators can act non-trivially on all qubits (i.e., $k = n$). Since $k=n$, we can eliminate many of the terms from the expression we are trying to compute:
\[
A \lows(A_{11} B_{11}^T) A^T + B (C_{11} B_{11}^T) A^T + B \lows(C_{11} D_{11}^T) B^T.
\]
Nevertheless, we unfortunately encounter a new problem. Now notice that computing $\lows( \cdot )$ now requires $n \times n$ matrix products, so there is not obvious way to compute it in $O(n^2 \ell^{\omega - 2})$ time.  For the $B (C_{11} B_{11}^T) A^T$ term, we can choose the order in which we compose the matrices:  first compute $B C_{11}$ which takes time $O(n^2 \ell^{\omega - 2})$ and leaves another $\ell \times n$ matrix, and so on.  Unfortunately, such a strategy does not automatically work for the other terms.  

Since the two $\lows(\cdot)$ terms have identical structure, let's focus on $A \lows(A_{11} B_{11}^T) A^T$. To reiterate,  $A$ is an $\ell \times n$ matrix, and $A_{11}$ and $B_{11}^T$ are $n \times n$ matrices. Let $T(n,\ell)$ be the time it takes to compute $A \lows(A_{11} B_{11}^T) A^T$ for dimensions $n$ and $\ell$.

We give a recursive strategy to speed up the brute force calculation of terms of this type. Let's decompose these matrices further, writing $A = [A_1 A_2]$ for $A_1, A_2 \in \{0,1\}^{\ell \times n/2}$, $A_{11} = [\begin{smallmatrix} X_{11} & X_{12} \\ X_{21} & X_{22} \end{smallmatrix}]$, with $X_{ij} \in \{0,1\}^{n/2 \times n/2}$, and $B_{11}^T = [\begin{smallmatrix} Z_{11} & Z_{12} \\ Z_{21} & Z_{22} \end{smallmatrix}]$, with $Z_{ij} \in \{0,1\}^{n/2 \times n/2}$.  Together, we can expand $A \lows(A_{11} B_{11}^T) A^T$ as
$$
A_1 \lows(X_{11} Z_{11} + X_{12} Z_{21}) A_1^T + A_2 (X_{21} Z_{11} + X_{22} Z_{21}) A_1^T + A_2 \lows(X_{21} Z_{12} + X_{22} Z_{22}) A_2^T.
$$
  By the above expression, we have
$$
T(n,\ell) \le 4 T(n/2, \ell) + c_1 n^2 \ell^{\omega-2}
$$
where $c_1$ is some constant.  We now claim that $T(n, \ell) \le c_2 n^2 \ell^{\omega -2} \log (2n/\ell)$.  This is clearly true when $\ell = n$, since we can compute $\ell \times \ell$ matrix products in $O(\ell^\omega)$ time.  Otherwise, by induction we have
\begin{align*}
T(n,\ell) &\le 4 c_2 (n/2)^2 \ell^{\omega -2} \log (n/\ell) + c_1 n^2 \ell^{\omega-2} \\
&=  c_2 n^2 \ell^{\omega -2} \log (2n/\ell) - c_2 n^2 \ell^{\omega -2} + c_1 n^2 \ell^{\omega-2} \\
&\le c_2 n^2 \ell^{\omega -2} \log (2n/\ell)
\end{align*}
for $c_2 \ge c_1$.
\end{proof}

We are now ready to prove the main measurement theorem.  We state the theorem in terms of a $Z$-basis measurement on $k$ qubits.  That said, the cost of this measurement dominates the cost of applying an arbitrary $k$-qubit Clifford unitary, so the theorem still holds for any choice of basis.

\begin{theorem}
\label{thm:fast_measurement}
Let any $n$-qubit stabilizer state (represented by its tableau) and a subset of $k$ qubits be given.  There exists an $O(\min(n^2 k^{\omega - 2} \log(2n/k), k n^2))$ time algorithm that returns
\begin{enumerate}[itemsep = 0pt]
\item $Z$-basis measurement outcomes for each qubit in the subset, and
\item The Clifford tableau for the $(n-k)$-qubit post-measurement state.
\end{enumerate}
Furthermore, the same algorithm can be used to postselect\footnote{In addition to the subset of $k$ qubits to be measured, a string of $k$ preferred measurement outcomes is given as input.  The algorithm either returns that the postselected event has probability $0$ or returns the $(n-k)$-qubit state assuming the postselection was successful.} on a particular $k$-bit outcome.
\end{theorem}
\begin{proof}
Without loss of generality, let us assume we measure the first $k$ qubits.  We will also only be concerned with the $O(n^2 k^{\omega - 2} \log(2n/k))$ measurement algorithm since one can trivially obtain a $O(k n^2)$ algorithm by applying the single-qubit measurement scheme of Aaronson and Gottesman $k$ times \cite{aaronson+gottesman:2004}.

Let us write the tableau for our state represented by Clifford $Q$ as
$$
M = \left[\begin{array}{ c c | c c }
A_{11} & A_{12} & B_{11} & B_{12} \\
A_{21}  & A_{22}  & B_{21} & B_{22} \\ \hline
C_{11} & C_{12} & D_{11} & D_{12} \\
C_{21} & C_{22} & D_{21} & D_{22}
\end{array}\right]
\hspace{20pt}
p = \left[\begin{array}{c} \alpha_1 \\ \alpha_2 \\ \hline \delta_1 \\ \delta_2 \end{array}\right]
\hspace{20pt}
s = \left[\begin{array}{c} \beta_1 \\ \beta_2 \\ \hline \epsilon_1 \\ \epsilon_2 \end{array}\right]
$$
where $A_{11} \in \{0,1\}^{k \times k}$, $A_{22} \in \{0,1\}^{(n-k) \times (n-k)}$, and so on.  We will perform the measurement in two steps, following the Aaronson-Gottesman measurement procedure \cite{aaronson+gottesman:2004}---first focusing on those outcomes which are random and then on those which are determinate.\footnote{Note that the choice of which measurements are random and which are determinate is not unique.  For example, consider the Bell state $\frac{\ket{00} + \ket{11}}{\sqrt 2}$.  Measuring the first qubit, determines the second, and vice versa.}  We will assume correctness from the Aaronson-Gottesman procedure and will focus primarily on running times.

\noindent\textit{Step 1: Random measurement outcomes}\\
\noindent 
To determine the random measurements, we first need the Pauli-$\pfont X$ parts of the stabilizers, i.e., the $n \times k$ submatrix $C := \left[\begin{smallmatrix} C_{11} \\ C_{21} \end{smallmatrix}\right]$, to be in reduced row echelon form.  To do this, we compute the $LSP$ factorization (see Ibarra, Moran, and Hui \cite{ibarra+moran+hui:1982}) of $C^T$ in $O(k^{\omega -1} n)$ time.   That is, $C^T$ can be written as the product of three matrices $LSP$, where $L$ is an $k \times k$ lower triangular matrix with 1's on the main diagonal, $S$ is an $k \times n$ matrix which reduces to an upper triangular matrix with 1's along the diagonal when the zero rows are deleted, and $P$ is an $n \times n$ permutation matrix.  Therefore, we get $C = P^T S^T L^T$.

We assume without loss of generality that $P$ is the identity transformation, since we can rearrange of the rows of the tableau according to $P$ without affecting the underlying state.  
 
Now, let us rearrange the columns of $S^T$ by permutation $\Pi$ such that only the first $r := \textrm{rank}(C)$ columns are non-zero.  That is, we write $C = (S^T \Pi) (\Pi^T L^T \Pi) \Pi^T$.  

The final matrix in this decomposition, $\Pi^T$, is a permutation on the first $k$ qubits of the state, so let us drop it in the following discussion since we can simply permute those columns of the tableau accordingly.  We can now write $C = S' U'$, where $S'$ is the $n \times r$ matrix obtained by removing the last $k - r$ columns from $S^T \Pi$ and $U'$ is the $r \times k$ matrix obtained by removing the last $k - r$ rows from $\Pi^T L^T \Pi$. Notice that the first $r \times r$ submatrix of $U'$ is upper triangular with $1$'s on the main diagonal.  Let's also write $S' = \left[\begin{smallmatrix} L' \\ A' \end{smallmatrix}\right]$ where $L'$ is an $r \times r$ lower triangular matrix with $1$'s on the diagonal, and $A'$ is arbitrary.

Before proceeding, let us recap the current state of the tableau. Until this point, we have only swapped its rows and columns, so nothing has fundamentally changed about its structure. We can either think of these swaps as occurring virtually, by keeping track of the permutation of the rows and columns, or explicitly, by physically moving the data around in the tableau as each row/column is only swapped constantly-many times. Because of this procedure, we can assume our original $n \times k$ matrix $C$ is of the form
\[
C = S' U' = \begin{bmatrix} L'U' \\ A'U' \end{bmatrix}
\]
We now define $\mathcal L$ to be the $n \times n$ matrix that reduces $C$ to row echelon form, i.e., $\mathcal{L} C = \left[\begin{smallmatrix} U' \\ 0 \end{smallmatrix}\right]$.  One can check that
$
\mathcal L = \left(\begin{smallmatrix} (L')^{-1} & 0 \\ A' (L')^{-1} & I \end{smallmatrix}\right)
$
has the desired behavior. 

First notice that we can construct $\mathcal L$ efficiently since computing the inverse of the $r \times r$ matrix $L'$ takes time $O(r^\omega)$ and computing the product between the $(n-r)\times r$ matrix $A'$ and $L'$ takes time at most $O(n r^{\omega -1})$.

We will view the application of $\mathcal L$ as directly multiplying the stabilizer generators of $Q$.  For example, if the $i$th row of $\mathcal L$ has $1$'s in the first, second, and fourth columns, we will combine the first, second, and fourth stabilizers in the tableau (i.e., compute $Q \pfont Z^{e_1} \pfont Z^{e_2}  \pfont Z^{e_4} Q^\dag$).  Since $\mathcal L$ is only dense in the first $r$ columns, we can use \Cref{lem:tableau_with_many_paulis} to construct the resulting tableau in $O(n^2 r^{\omega - 2})$ time (i.e., the variables ``$\ell$'' and ``$k$'' in the statement of \Cref{lem:tableau_with_many_paulis} are set to $n$ and $r$, respectively).  Since we are only taking linear combinations of generators, the underlying state does not change. To preserve the fact that the tableau is symplectic, we can apply \Cref{lem:tableau_with_many_paulis} again with $\mathcal L^{-T} = \left(\begin{smallmatrix} (L')^T & (A')^T \\ 0 & I \end{smallmatrix}\right)$ applied to the destabilizers.\footnote{In other words, the row operations on the stabilizers/destabilizers result from multiplying the following matrix on the entire tableau:
\scalebox{.75}{$\left[\begin{array}{c | c} \mathcal L^{-T} & 0 \\ \hline  0 & \mathcal L \end{array}\right].$}
Since this matrix is symplectic, it corresponds to a valid Clifford tableau with 0's for phase bits (\Cref{fact:tableau_facts}). Therefore, the resulting product is also a valid Clifford tableau. Here and elsewhere, we use $-T$ to denote both inverse and transpose.}  Because $\mathcal L^{-T}$ is now upper triangular, we must use the other setting of \Cref{lem:tableau_with_many_paulis} (i.e., with the variable ``$\ell$'' set to $r$ and ``$k$'' set to $n$), which takes $O(n^2 r^{\omega-2} \log(2n/r))$ time.

To recap, the original submatrix $C$ of our tableau is now $\left[\begin{smallmatrix} U' \\ 0 \end{smallmatrix}\right]$.  Also recall that $U' = \left[\begin{smallmatrix} U'' A'' \end{smallmatrix}\right]$ where $U''$ is an $r \times r$ upper triangular matrix with ones along the diagonal.  Therefore, we can compute $(U'')^{-1}$ explicitly in time $O(r^\omega)$, and once again apply \Cref{lem:tableau_with_many_paulis} so that each of the first $r$ stabilizers has exactly one Pauli-$\pfont X$ component (i.e., the submatrix is in reduced row echelon form).  We then apply $(U'')^T$ to the destabilizers in the same fashion.

In the next step, we multiply the stabilizer generators into the destabilizer generators, so that no destabilizer has a Pauli-$\pfont X$ component within the first $r$ qubits.\footnote{A note of warning: this is the first operation that we apply to the tableau which is not symplectic. For a more in-depth justification of this procedure, we refer the reader to the work of Aaronson/Gottesman \cite{aaronson+gottesman:2004}.} Once again, this can be accomplished by \Cref{lem:tableau_with_many_paulis}. Then, we replace the first $r$ destabilizers with the first $r$ stabilizers. 

Since these first $r$ qubits will be the qubits for which the measurement outcome is random, we can now output a uniformly random measurement for each (or, in the case of postselection, the specified outcome).  Let this measurement be $z \in \{0,1\}^r$.  For $i \in \{1, \ldots, r\}$, we replace the $i$th stabilizer with $(-1)^{z_i} \pfont Z^{e_i}$.

Unfortunately, the tableau we just constructed is only valid for the post-measurement state in tensor product with the measured qubits.  Since we would like a tableau for just the $(n-k)$-qubit post-measurement state, additional work is required.  To recap, the state of the tableau is now
$$
M = \left[\begin{array}{ c c | c c }
I & A_{12} & B_{11} & B_{12} \\
0  & A_{22}  & B_{21} & B_{22} \\ \hline
0 & 0 & I & 0 \\
0 & C_{22} & D_{21} & D_{22}
\end{array}\right]
\hspace{20pt}
p = \left[\begin{array}{c} \alpha_1 \\ \alpha_2 \\ \hline 0 \\ \delta_2 \end{array}\right]
\hspace{20pt}
s = \left[\begin{array}{c} \beta_1 \\ \beta_2 \\ \hline z \\ \epsilon_2 \end{array}\right]
$$
where $A_{12} \in \{0,1\}^{n \times n-r}$, and so on.  To clarify, this is a slight abuse of notation since we've subdivided the matrix in terms of $r$ instead of $k$ (and we've chosen to redefine the submatrices instead of giving them new names).   In order to safely remove the first $r$ qubits from the tableau, we need to ensure that the tableau will still be symplectic afterwards. In some sense, we need $A_{12}, B_{11}, B_{12}, B_{21}$ and $D_{21}$ all to be $0$.  As before, we can zero out $D_{21}$ using \Cref{lem:tableau_with_many_paulis} with the transformation $L = \left(\begin{smallmatrix} I & 0 \\ D_{21} & I \end{smallmatrix}\right)$ on the stabilizers and $L^{-T} = \left(\begin{smallmatrix} I & D_{21}^T \\ 0 & I \end{smallmatrix}\right)$ on the destabilizers. This results in the tableau:
$$
M' = \left[\begin{array}{ c c | c c }
I & A_{12}' & B_{11}' & B_{12}' \\
0  & A_{22}'  & B_{21}' & B_{22}' \\ \hline
0 & 0 & I & 0 \\
0 & C_{22} & 0 & D_{22}
\end{array}\right]
\hspace{20pt}
p' = \left[\begin{array}{c} \alpha_1' \\ \alpha_2' \\ \hline 0 \\ \delta_2 \end{array}\right]
\hspace{20pt}
s' = \left[\begin{array}{c} \beta_1' \\ \beta_2' \\ \hline z \\ \epsilon_2' \end{array}\right]
$$
We can now construct the tableau for the state on $n - r$ qubits by simply removing the first $r$ rows/columns from each submatrix:
$$
M'' = \left[\begin{array}{ c | c }
A_{22}'  & B_{22}' \\ \hline
C_{22} & D_{22}
\end{array}\right]
\hspace{20pt}
p'' = \left[\begin{array}{c} \alpha_2' \\ \hline \delta_2 \end{array}\right]
\hspace{20pt}
s'' = \left[\begin{array}{c} \beta_2' \\ \hline \epsilon_2' \end{array}\right]
$$
It may not be the case that $A_{12}', B_{11}', B_{12}',$ and $B_{21}'$ are all zero, but one can check that the tableau is symplectic using the fact that it was symplectic before the rows/columns were removed.

\noindent\textit{Step 2: Determinate measurement outcomes}\\
\noindent 
Since we have already removed the $r$ qubits which had random outcomes in Step 1, let us assume without loss of generality that we are measuring the first $k$ qubits, all of which have determinate outcomes.  That is, for $i \in \{1, \ldots, k\}$, we have $\epsilon_i \pfont Z^{e_i}$ is a stabilizer of the state for some $\epsilon_i \in \{\pm1\}$.  The goal is to determine the $\epsilon_i$. 

To start, let us put the Pauli $\pfont X$ part of the destabilizers---i.e., the submatrix $\left[\begin{smallmatrix} A_{11} \\ A_{21} \end{smallmatrix}\right]$ ---into reduced row echelon form.  Once again, this will require a fast $LSP$ decomposition coupled with several applications of \Cref{lem:tableau_with_many_paulis}.  Recall that their are no Pauli-$\pfont X$ terms in the stabilizers, so the tableau of the resulting transformation is (once again, redefining the submatrices instead of giving them new names)
$$
M = \left[\begin{array}{ c c | c c }
I & A_{12} & B_{11} & B_{12} \\
0  & A_{22}  & B_{21} & B_{22} \\ \hline
0 & C_{12} & D_{11} & D_{12} \\
0 & C_{22} & D_{21} & D_{22}
\end{array}\right]
\hspace{20pt}
p = \left[\begin{array}{c} \alpha_1 \\ \alpha_2 \\ \hline \delta_1 \\ \delta_2 \end{array}\right]
\hspace{20pt}
s = \left[\begin{array}{c} \beta_1 \\ \beta_2 \\ \hline \epsilon_1 \\ \epsilon_2 \end{array}\right].
$$
In fact, because $M$ is symplectic, we have that $C_{12} = 0$, $D_{11} = I$, and $D_{12} = 0$, so we can immediately report the measurement outcome of the $i$th qubit as $(\epsilon_1)_i$.  In the case of postselection, we check if these outcomes match the specified outcomes.  If so, then there is nothing to be done; otherwise, we report that the postselected event has probability 0.

Finally, we repeat the procedure from the end of Step 1:  row reducing so that $D_{21} = 0$, and then returning the tableau with the first $k$ rows/columns removed from each submatrix.  This completes the proof.
\end{proof}

\begin{reptheorem}{thm:fastmm}
Let $\ket{\psi}$ be any $n$-qubit stabilizer state given by its tableau, and let $z = z_1 \ldots z_k$ be any $k$-bit string of postselected outcomes.  There is an $O(n^\omega)$ algorithm to obtain a uniformly random measurement of the state $(\bra{z} \otimes I) \ket{\psi}$ in the computational basis.\footnote{If $\omega = 2$, we pick up an additional log factor in the solution of the recurrence, i.e., a running time of $O(n^2 \log n)$.}
\end{reptheorem}
\begin{proof}
The idea will be to invoke \Cref{thm:fast_measurement} twice---first, for the postselected qubits, and then for the remaining qubits. In other words, if we let $f(n,k) := n^2 k^{\omega -2} \log(2n/k)$ denote the bound given for the measurement protocol in \Cref{thm:fast_measurement}, then the running time of the proposed algorithm is bounded by 
\[
f(n,k) + f(n-k, n-k) = O(n^2 k^{\omega - 2} \log(2n/k) + (n-k)^\omega).
\]
It remains to show that gives a final bound of $O(n^\omega)$ for all $k$. We separate into two cases based on how $k$ grows with $n$. When $k = \Theta(n)$, we have $c_1 n \le k \le c_2 n$ for some constants $c_1, c_2 > 0$, and so
\begin{align*}
n^2 k^{\omega - 2} \log(2n/k) + (n-k)^\omega
\le c_2^{\omega - 2} n^\omega \log(2/c_1) + n^\omega = O(n^\omega).
\end{align*}
When $k = o(n)$, we have $k = n / g(n)$ for some $g(n) = \omega(1)$, and so
\begin{align*}
n^2 k^{\omega - 2} \log(2n/k) + (n-k)^\omega
\le  \frac{n^\omega}{g(n)^{\omega - 2}} \log(2 g(n)) + n^\omega = O(n^\omega)
\end{align*}
where we use that $\log(2 g(n)) / g(n)^\epsilon = o(1)$ for all constants $\epsilon > 0$.  Therefore, in the case of $\omega = 2$, we actually get a bound of $O(n^2 \log n)$.
\end{proof}

We note again that the choice of measurement bases in \Cref{thm:fastmm} is arbitrary since single-qubit gates can be applied in $O(n)$ time.

\section{Acknowledgments}
D.\ Gosset and D.\ Grier are supported in part by IBM Research. We also acknowledge the support of the Natural Sciences and
Engineering Research Council of Canada and the Canadian Institute for Advanced Research. We thank Sergey Bravyi and Robert K\"{o}nig for helpful discussions. We also thank an anonymous reviewer for detailed feedback on presentation and for identifying an incorrect statement regarding runtime in Claim~\ref{claim:graphgprime} in the previous version.

\apptocmd{\sloppy}{\hbadness 10000\relax}{}{}
\bibliographystyle{quantum}
\bibliography{bibliography}

\appendix

\section{Applications}
\label{app:applications}

In this section, we formally state and prove Theorems \ref{thm:clifford_tensor} and \ref{thm:clifsim} from the introduction concerning the application of our graph state sampling problem to more traditional circuit simulation tasks.  We will need the well-known fact that measurements of quantum states can induce certain unitary transformations on the remaining qubits. This is the basis of measurement-based quantum computation \cite{raussendorf2001one} and the heuristic algorithm for Clifford simulation due to Anders and Briegel \cite{anders2006fast}. Bravyi has established a linear-time reduction from constant-depth circuit simulation to graph state simulation using the tableau representation of stabilizer states \cite{bravyireduction}. To keep things relatively self-contained, here we shall use a ``Hadamard gadget'' \cite{bremner2011classical} which allows us to implement Hadamard gates using $X$-basis measurements (see \Cref{fig:identities}).  We note that a similar technique was used in the recent paper of Guan and Regan \cite{guan2019stabilizer}.

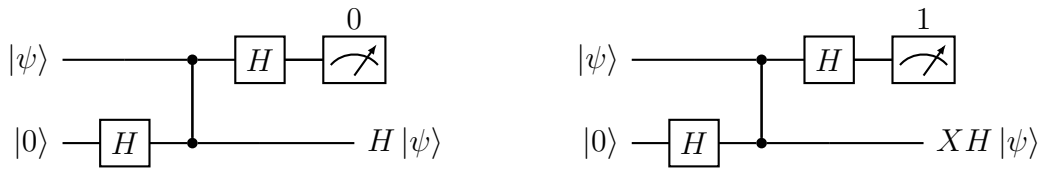
\begin{figure}[h!]
\centering
    \begin{subfigure}{0.45\textwidth}
    \begin{quantikz}
    \lstick{$\ket{\psi}$} & \qw & \ctrl{1} & \gate{H} & \meter{0} \qw\\
		\lstick{$\ket{0}$} & \gate{H} & \ctrl{-1} & \qw & \rstick{$H\ket{\psi}$} \qw 
    \end{quantikz}
\end{subfigure}
\begin{subfigure}{0.45\textwidth}
  \begin{quantikz}
    \lstick{$\ket{\psi}$} & \qw & \ctrl{1} & \gate{H} & \meter{1} \qw\\
		\lstick{$\ket{0}$} & \gate{H} & \ctrl{-1} & \qw & \rstick{$X H\ket{\psi}$} \qw 
    \end{quantikz}
\end{subfigure}
\caption{The Hadamard gadget implements either $H$ or $XH$ on the single-qubit input state $\psi$ depending on the measurement outcome. By linearity, the gadget can also be used to implement a Hadamard gate on a single-qubit of a multi-qubit state. \label{fig:identities}}
\end{figure}

\subsection{Clifford tensor network simulation}
\label{sec:clifford_tensor}
A \textit{Clifford tensor} $T$ of rank $k$ is an array of $2^k$ complex numbers $T_{i_1 i_2\ldots i_k}\in \mathbb{C}$ indexed by $k$ binary variables $i_1,i_2,\ldots, i_k \in \{0,1\}$, with the additional property that the $k$-qubit state
\begin{equation}
|T\rangle\equiv \sum_{i_1 i_2\ldots i_k} T_{i_1 i_2\ldots i_k} |i_1 i_2\ldots i_k\rangle
\label{eq:tensorstate}
\end{equation}
is proportional to a stabilizer state. Such a tensor can be represented pictorially as in \Cref{fig:tensordef} where each ``leg" represents one of the $k$ binary variables. A leg of tensor $T^{(1)}$ and a leg of tensor $T^{(2)}$ can be contracted by summing the corresponding binary variable as depicted in \Cref{fig:contraction}.

A Clifford tensor network consists of a set $T^{(1)}, T^{(2)}, \ldots , T^{(n)}$ of ``elementary" Clifford tensors along with a set $\mathcal{E}$ of pairs of legs that are to be contracted. Write $\mathcal{T}$ for the tensor obtained from the given collection of tensors by contracting all pairs of legs in $\mathcal{E}$.  We shall abuse notation and use $\mathcal{T}$ to refer both to the tensor network and the tensor that it evaluates to. In particular, $\mathcal{T}$ is a Clifford tensor with rank equal to the number of uncontracted legs. If it is nonzero, then it defines a stabilizer state (up to normalization) as in Eq.~\eqref{eq:tensorstate} and we consider the problem of simulating measurement of all qubits of this state.

The \emph{underlying graph} of the tensor network $\mathcal{T}$ is the bipartite graph with a vertex for each elementary tensor $T^{(i)}$, a vertex for each pair of contracted legs $e=\{\ell_1,\ell_2\}\in \mathcal{E}$ (i.e., an edge/wire of the network), and edges $\{e,\ell_1\}$ and $\{e,\ell_2\}$ between tensor vertices and contracted leg vertices which are incident in the network.

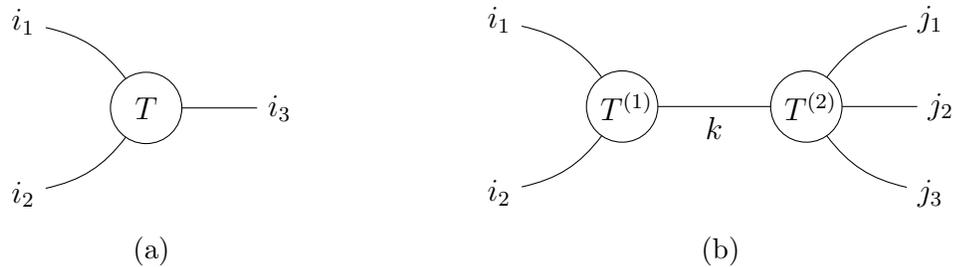
\begin{figure}[h!]
\centering
\begin{subfigure}[t]{0.45\textwidth}
\centering
\begin{tikzpicture}
\tikzstyle{tensor}=[draw, circle, minimum size=27pt, inner sep=0pt, outer sep=0pt]
    \node[tensor] (T) []  {$T$};
    \node (i1) [above left = .5 and 1 of T] {$i_1$};
    \node (i2) [below left = .5 and 1 of T] {$i_2$};
    \node (i3) [right = 1 of T] {$i_3$};

    \path[draw]
    (T) edge (i3)
    (T) edge[bend right=20] (i1)
    (T) edge[bend left=20] (i2);
\end{tikzpicture}
\subcaption{}
\label{fig:tensordef}
\end{subfigure}
\begin{subfigure}[t]{0.45\textwidth}
\centering
\begin{tikzpicture}
\tikzstyle{tensor}=[draw, circle, minimum size=27pt, inner sep=0pt, outer sep=0pt]
    \node[tensor] (T1) []  {$\;T^{(1)}$};
    \node (i1) [above left = .5 and 1 of T1] {$i_1$};
    \node (i2) [below left = .5 and 1 of T1] {$i_2$};
    \node[tensor] (T2) [right = 1.5 of T1]  {$\;T^{(2)}$};
    \node (j1) [above right = .5 and 1 of T2] {$j_1$};
    \node (j2) [right = 1 of T2] {$j_2$};
    \node (j3) [below right = .5 and 1 of T2] {$j_3$};

    \path[draw]
    (T1) edge[bend right=20] (i1)
    (T1) edge[bend left=20] (i2)
    (T1) edge node[below] {$k$} (T2)
    (T2) edge[bend left=20] (j1)
    (T2) edge (j2)
    (T2) edge[bend right=20] (j3);
\end{tikzpicture}
\subcaption{}
\label{fig:contraction}
\end{subfigure}
\caption{(a) Graphical representation of rank-$3$ tensor $T$ with indices $i_1$, $i_2$, and $i_3$. (b) Contracting rank-$3$ tensor $T^{(1)}$ with rank-$4$ tensor $T^{(2)}$ on index $k$ gives a rank-$5$ tensor $T_{i_1 i_2 j_1 j_2 j_3} := \sum_{k \in \{0,1\}} T^{(1)}_{i_1 i_2 k} T^{(2)}_{j_1 j_2 j_3 k}$.}
\end{figure}

\begin{PCTNSproblem}
The input is a Clifford tensor network $\mathcal{T}$ which has a planar underlying graph, and each elementary tensor has constant rank. If $\mathcal{T}$ is the zero tensor then the output is a flag indicating this; otherwise the output is a binary string sampled from the probability distribution obtained by measuring each qubit of $|\mathcal{T}\rangle$ in the standard basis.
\end{PCTNSproblem}

\begin{reptheorem}{thm:clifford_tensor}
There is a classical algorithm which solves the planar Clifford tensor network simulation problem with runtime $\widetilde{O}(n^{\omega/2})$, where $n$ is the number of elementary tensors.
\end{reptheorem}

\begin{proof}
Let $\mathcal{T}$ be specified by elementary tensors $T^{(1)}, T^{(2)},\ldots, T^{(n)}$ and pairs of contracted legs $\mathcal{E}$. Each elementary tensor defines a stabilizer state $|T^{(i)}\rangle$ via Eq.~\eqref{eq:tensorstate}. The  state defined by the tensor network $|\mathcal{T}\rangle$ is then given by
\[
|\mathcal{T}\rangle=\prod_{e\in \cal{E}} \left(\langle 00|+\langle 11|\right)_{e}|T^{(1)}\rangle\otimes |T^{(2)}\rangle\otimes\ldots \otimes |T^{(n)}\rangle
\]
The state $|T^{(1)}\rangle \otimes \ldots \otimes |T^{(n)}\rangle$ appearing on the right hand side has one qubit for each leg of each elementary tensor. It is an $m$-qubit state, where $m=O(n)$ is the total number of legs in the tensor network. We have then used the fact that contraction of two legs is equivalent to projecting the two corresponding qubits onto the Bell pair $\langle 00|+\langle 11|$.

We now use the following lemma to express each of the elementary tensors $|T^{(i)}\rangle$ as a graph state with local unitaries:
\begin{lemma}[Theorem 1 from Ref.~\cite{van2004graphical}]
A $k$-qubit stabilizer state $|\phi\rangle$ can be expressed as $|\phi\rangle=C_1\otimes C_2\otimes \ldots \otimes C_k|G\rangle$ where $G$ is a graph state and $C_1,C_2,\ldots, C_k$ are single-qubit Clifford unitaries.
\end{lemma}
Each elementary tensor has constant rank and therefore we may compute the single-qubit Cliffords and graph state for each elementary tensor in constant time on a classical computer.  In this way, we compute graphs $G_1,G_2,\ldots, G_n$ and single-qubit Cliffords $C_1,C_2,\ldots, C_m$ such that
\begin{equation}
|\mathcal{T}\rangle=\prod_{e\in \cal{E}} \left(\langle 00|+\langle 11|\right)_{e} C_1\otimes C_2\otimes \ldots \otimes C_m |G_1\rangle|G_2\rangle \ldots |G_n\rangle. 
\label{eq:tensorgraph}
\end{equation}
Our goal is to sample a binary string from the probability distribution obtained by measuring this state in the computational basis (or determine if $\mathcal{T}=0$). Below we show that this can be recast as an instance of the graph state simulation problem with a graph $G$ that admits an $O(1)$-coarse-graining to a planar graph, and then Theorem \ref{thm:coarse} completes the proof.

Consider a pair of qubits $e=\{i,j\}\in \mathcal{E}$ and the corresponding state
\[
C_i^{\dagger} \otimes C_j^{\dagger} \frac{1}{\sqrt{2}}(|00\rangle+|11\rangle)
\]
which appears (up to normalization) in Eq.~\eqref{eq:tensorgraph}. Noting that 
\[
\frac{1}{\sqrt{2}}\left(|00\rangle+|11\rangle\right)=(I\otimes H) \CZ \ket{++},
\]
we get
\begin{equation}
C_i^{\dagger} \otimes C_j^{\dagger} \frac{1}{\sqrt{2}}(|00\rangle+|11\rangle)=
I\otimes C_j^{\dagger} \overline{C}_i \frac{1}{\sqrt{2}}\left(|00\rangle+|11\rangle\right)= (I\otimes D) \CZ \ket{++}.
\label{eq:D}
\end{equation}
where $D=C_j^{\dagger} \overline{C}_i H$ is a single-qubit Clifford, and therefore can be written as a product of Hadamard and phase gates~\cite{nest2008classical}
\[
D\propto S^{a} H^{b} S^{c},
\]
for some $a,c\in \{0,1,2,3\}$ and $b\in \{0,1\}$. If $b=0$ then
\begin{equation}
C_i^{\dagger} \otimes C_j^{\dagger} \frac{1}{\sqrt{2}}(|00\rangle+|11\rangle) \propto  \CZ |\phi_i\otimes \phi_j\rangle
\label{eq:b0}
\end{equation}
for single-qubit stabilizer states $|\phi_i\rangle=|+\rangle$ and $|\phi_j\rangle=S^{a+c}|+\rangle$. On the other hand, if $b=1$ then we can use the Hadamard gadget from \Cref{fig:identities} to get
\begin{align}
C_i^{\dagger} \otimes C_j^{\dagger} \frac{1}{\sqrt{2}}(|00\rangle+|11\rangle) &\propto (I\otimes S^{a}HS^{c} )\CZ \ket{++}\nonumber\\
&= (I\otimes \langle+|\otimes I)\CZ_{ie}\CZ_{ej} |\phi_i\otimes \phi_e \otimes \phi_j\rangle,
\label{eq:b1}
\end{align}
where $|\phi_i\rangle=|+\rangle,|\phi_j\rangle=S^{a}|+\rangle,|\phi_e\rangle=S^{c}|+\rangle$ are single-qubit stabilizer states. 

For every pair of contracted legs $e=\{i,j\}\in \mathcal{E}$ in the tensor network we can compute $b\in \{0,1\}$ and the decomposition given by either Eq.~\eqref{eq:b0} or Eq.~\eqref{eq:b1}. Let us partition the elements of $\mathcal{E}=\mathcal{E}_0\cup \mathcal{E}_1$ where $\mathcal{E}_0, \mathcal{E}_1$ have $b=0,1$ respectively.  We obtain
\begin{equation}
|\mathcal{T}\rangle\propto \prod_{e=\{i,j\}\in \mathcal{E}_1} \left(\langle \phi_i \otimes \phi_e\otimes \phi_{j}|\CZ_{ie} \CZ_{ej} |+\rangle_e\right) \prod_{e=\{i,j\}\in \mathcal{E}_0} \left(\langle \phi_i \otimes \phi_j|\CZ\right)_{ij} |G_1\rangle|G_2\rangle \ldots |G_n\rangle.
\label{eq:psisample}
\end{equation}

In Eq.~\eqref{eq:psisample} we have a single-qubit stabilizer state $\langle \phi_i|$ for each contracted leg and one $\langle \phi_e|$ for each $e\in \mathcal{E}_1$. Letting $\mathcal{L}\subseteq [m]$ be the set of all contracted legs, we have
\begin{equation}
|\mathcal{T}\rangle =\left(\prod_{i\in \mathcal{L}} \langle \phi_i| \prod_{e\in \mathcal{E}_1} \langle \phi_e|\right)|G_{\mathcal{T}}\rangle 
\label{eq:tgss}
\end{equation}
where $G_\mathcal{T}$ is a graph with $m+|\mathcal{E}_1|$ vertices, defined as follows. Start with a graph with $m$ vertices that consists of a copy of each of the graphs $G_1,G_2, \ldots ,G_n$. Next add an edge between every pair $e=\{i,j\}\in \mathcal{E}_0$. Finally, add a new vertex labeled $e$ for each $e=\{i,j\}\in \mathcal{E}_1$ and edges $\{i,e\}$ and $\{e,j\}$.

In light of Eq.~\eqref{eq:tgss}, to solve the planar Clifford tensor network simulation problem defined by $\mathcal{T}$, it suffices to simulate single-qubit Pauli measurements on the graph state $|G_\mathcal{T}\rangle$ and postselect on the subset of measurement outcomes corresponding to the qubits in $\mathcal{L}\cup \mathcal{E}_1$. This is an instance of the graph state simulation problem with the graph $G_{\mathcal{T}}$.

To complete the proof, we construct an $r$-coarse-graining from $G_\mathcal{T}$ to a planar graph, derived from the underlying planar graph of the tensor network. For all $1 \leq i \leq n$, let the coarse-graining map all vertices of $G_i$ to the vertex for $T^{(i)}$ in the underlying graph. Map the additional vertex created for a pair of contracted legs $e=\{k, \ell\} \in \mathcal{E}_1$ to the vertex labelled $e$ in the underlying graph of the tensor network. If $G_\mathcal{T}$ has no vertex for a pair of contracted legs $e\in \mathcal{E}$, then let us remove that vertex from the underlying graph by contracting either edge incident to it; clearly the modified graph is still planar. At worst, the coarse-graining collapses all of $G_i$ to a single vertex (corresponding to $T^{(i)}$), but this is at most the rank of the tensors $T^{(i)}$, which is constant. 
\end{proof}

\subsection{Clifford circuit simulation}
\label{sec:planar_circuit_application}

Because Clifford circuit simulation is a such a common task, this section is devoted to proving an independent reduction from graph state simulation that does not require the use of postselection.  We now state a formal version of the theorem we'd like to prove:

\begin{reptheorem}{thm:clifsim}
Consider $n$ qubits arranged at the vertices of a planar graph $G$. Let $C$ be an $n$-qubit, depth $d$ quantum circuit composed of one- and two-qubit Clifford gates such that each two-qubit gate acts nontrivially on an edge of $G$. There is a classical algorithm with runtime upper bounded as $\widetilde{O}(n^{\omega/2} d^{\omega})$ which produces a sample $z\in \{0,1\}^n$ drawn from the output distribution of $C$.
\end{reptheorem}

Before we give a new proof of this theorem, let's sketch why it follows from \Cref{thm:clifford_tensor}.  Notice that an $n$-qubit Clifford circuit can be viewed as a Clifford tensor network obtained by contracting $n$ rank-$1$ tensors representing the initial single-qubit states $|0\rangle$, with rank-$2$ and rank-$4$ Clifford tensors corresponding to $1$- and $2$-qubit Clifford gates.  The underlying graph of this tensor network is not necessarily planar, but is $O(d)$-coarse-grained planar as can be seen by grouping all the gates acting on the same qubit.  This suffices to prove the claimed bound on the running time.

\begin{proof}[Proof of \Cref{thm:clifsim}]
Suppose without loss of generality that all gates are either $H,S$, or $\CZ$, since any one- or two-qubit Clifford gate can be rewritten as a circuit of size $O(1)$ using only these three gates.  We will massage the circuit into a certain canonical form along the lines of Ref. \cite{bremner2011classical}. Consider a circuit $C'$ which is obtained from $C$ by inserting two Hadamard gates acting on each qubit at the very beginning of the circuit, and also at the very end of the circuit. The unitaries $U_C, U_{C'}$ implemented by these circuits are the same, since 
\[
U_{C'}=(H^{\otimes n})^2U_{C}(H^{\otimes n})^2=U_C. 
\]
Moreover, the depth of $C'$ is $d+4$. Let us say that a Hadamard gate in $C'$ is a \textit{middle Hadamard gate} if it is in the subcircuit $H^{\otimes n}U_{C}H^{\otimes n}$ containing all gates except the initial and final layer of Hadamards. It will be convenient to write $h$ for the total number of middle Hadamard gates in $C'$. For any binary string $z\in \{0,1\}^h$ we also define $C'(z)$ to be the quantum circuit obtained from $C'$ by replacing the $j$th middle Hadamard gate by $X^{z_j} H$, for each $j\in [h]$ (here we fix some arbitrary ordering of the middle Hadamards).

Now the circuit $C'$ is composed of $O(nd)$ one- and two-qubit Clifford gates, and the circuit $C'(z)$ is obtained from $C'$ by inserting $O(nd)$ Pauli $X$ gates in locations determined by $z$. It follows that, given any binary string $z\in \{0,1\}^h$ we may in $O(nd)$-time compute an $n$-qubit Pauli $P(z)$ such that 
\begin{equation}
U_{C'(z)}=P(z)U_{C'}=P(z)U_C.
\label{eq:pofz}
\end{equation}
In particular, $P(z)$ is obtained by pushing all Pauli corrections to the end of the circuit using the fact that every gate in the circuit is Clifford and therefore maps Paulis to Paulis under conjugation.

We now construct a new quantum circuit by replacing all middle Hadamard gates in $C'$ with the Hadamard gadget depicted in \Cref{fig:identities}. Each Hadamard gadget introduces one ancilla qubit, so the resulting circuit has $n+h$ qubits in total. Let us now partition these $n+h$ qubits as $[n+h]=A_1\cup A_2\cup\ldots \cup A_n$ where subsets $A_j\subseteq [n+h]$ are defined as follows.  For each $j\in [n]$, the set $A_j$ contains qubit $j$ as well each ancilla qubit from a gadget replacing a Hadamard gate acting on the $j$th qubit in $C'$. Note that $|A_j|\leq d+4$ for all $j\in [n]$ since $C'$ has depth at most $d+4$.

Since all gates in $C'$ were either $H, S$ or $\CZ$ gates, after inserting the Hadamard gadgets we are left with a circuit of the form
\[
V=H^{\otimes {n+h}} \prod_{j\in M} S_j\prod_{\{i,j\}\in E(G')} \CZ_{ij} H^{\otimes {n+h}},
\]
for some subset $M\subseteq [n+h]$ and graph $G'$ with vertices labeled by $[n+h]$. Moreover, there exists a $(d+4)$-coarse-graining from $G'$ to $G$. Indeed, such a coarse-graining is defined by $\varphi \colon [n+h]\rightarrow [n]$ where $\varphi(k)$ is the unique $j\in [n]$ satisfying $k\in A_j$. As noted above, $|A_j|\leq d+4$ for all $j$. One can also straightforwardly verify that $\{u,v\}\in E(G')$ only if either (a) $u,v\in A_j$ for some $j$, or (b) $u\in A_j$ and $v\in A_i$ for some $\{i,j\}\in E(G)$.

We may therefore use the algorithm from Theorem \ref{thm:coarse} to obtain $x\in \{0,1\}^n$ and $z\in \{0,1\}^h$ such that $\langle x|\otimes \langle z| V|0^{n+h}\rangle \neq 0$. (Note that here we only need to solve an instance of the graph state simulation problem with $\mathcal{P}=\varnothing$.) Using the identities in Fig.~\ref{fig:identities} we then have
\[
\langle x|\otimes \langle z| V|0^{n+h}\rangle=\langle x|U_{C'(z)}|0^n\rangle 
\]
Finally, as noted above, we may compute (in $O(nd)$ time) an $n$-qubit Pauli $P(z)$ such that 
\[
\langle x|U_{C'(z)}|0^n\rangle =\langle x|P(z)U_{C}|0^n\rangle \neq 0.
\]
Letting $g\in \{0,1\}^n$ be such that $P(z)|x\rangle \propto |g\rangle$ we see that $\langle g|U_C|0^n\rangle\neq 0$. Finally, we select a random $s\in \{0,1\}^n$ and let 
\[
Q(s)=U_C Z(s) U_C^{\dagger}.
\]
Note that $Q(s)$ is a random stabilizer of $U_C|0^n\rangle$ and can be computed in time $O(nd)$, since $U_C$ has $O(nd)$ gates. By Claim \ref{claim:simplestab}, the binary string $y$ such that $|y\rangle \propto Q(s)|g\rangle$ is a random sample from the output distribution of $C$. This $y$ is the output of the algorithm. The total runtime is $\widetilde{O}(n^{\omega/2}d^{\omega})$ since all steps except for the graph state simulation take time $O(nd)$.
\end{proof}

\section{Software Implementation}
\label{sec:software_implementation}

This section is about our software implementation of the graph state simulation problem for the $n^{\frac{1}{2}} \times n^{\frac{1}{2}}$ grid graph.  The goal was to show that the asymptotic speedup we predict for the grid in \Cref{thm:grid_warmup} is actually practical, i.e., that it is feasible to implement and that the increased complexity of the algorithm does not hurt the runtime too much for small grids. We do not implement fast matrix multiplication for simplicity, so in the theorem we may set $\omega \to 3$ and avoid some extra log factors to obtain an $O(n^{3/2})$-time algorithm.

Recall that this running time is achieved by a divide-and-conquer recursive algorithm described in the example on page~\pageref{ex:grid_warmup}. We compare the running time of this algorithm with the na\"ive algorithm and the left-to-right sweep algorithm in \Cref{fig:grid_simulation}. In each case, for a given $n$, the algorithm was translated into a circuit (on $n$ qubits for the na\"{i}ve algorithm, $O(\sqrt{n})$ for the others), and then fed to the same tableau-based Clifford simulator
 (CHP++ \cite{gridCHPpp})
running on a modern desktop computer.\footnote{The experiments were conducted on a desktop computer with an AMD Ryzen 5 3600 processor and 16GB of RAM.} We make one concession to translate the divide-and-conquer algorithm to a circuit: each gate includes information about which qubits (represented as a interval) could be currently entangled with its inputs. In this way, we can represent the multiple subtableaux (demanded by a divide-and-conquer strategy) as separate intervals of qubits within the full tableau, without the performance hit that would come from scanning the entire tableau for each operation. Finally, in our tests, the measurement basis ($\pfont X$ or $\pfont Y$) for each qubit is chosen at random; we expect (and conjecture) that this is induces near worst-case performance in the three algorithms. 

In the log-log plot shown in \Cref{fig:loglogplot}, the line of best fit for all three algorithms suggests a better runtime than we predict.  Namely, the na\"ive algorithm scales as $n^{1.7}$ whereas the theoretical guarantee is $O(n^3)$; the left-to-right sweep algorithm scales as $n^{1.5}$ with theoretical guarantee only $O(n^2)$; and the recursive algorithm scales as $n^{1.2}$ with theoretical guarantee only $O(n^{1.5})$.  Nevertheless, we believe this is simply due to small-$n$ effects.

\begin{figure}
\centering
\input{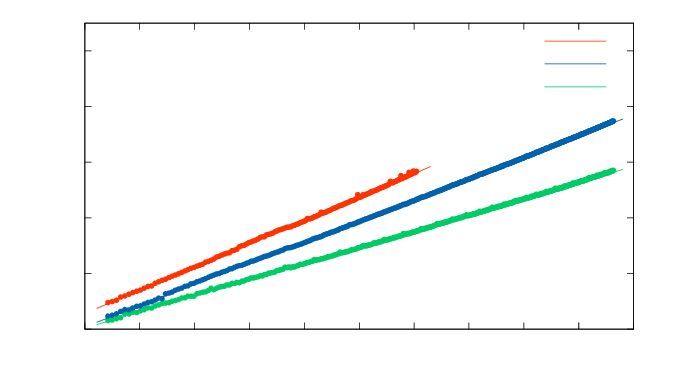}
\caption{Log-log plot of the simulation time (seconds) vs. $\ln(n)$ (i.e., log the total number of qubits in the $\sqrt{n} \times \sqrt{n}$ grid).  The slope is calculated as a linear regression of the data points, and shows the empirical scaling with number of qubits. The data depicted is the same as that in \Cref{fig:grid_simulation} except only for $\sqrt n$ at least 100. \label{fig:loglogplot}}
\end{figure}

In an effort to improve the performance of the recursive algorithm as much as possible, our software implementation differs from that given in the example in two important ways:
\begin{enumerate}
    \item In each recursive step, we split the grid into two (almost) identical halves rather than four quarters as in the example. Of course, in the next step we split those halves again in the opposite direction (the rule is to split the longer dimension),  so we end up with four quarters anyway.
    \item We eliminate the frame qubits between the two halves. The recursive calls return the state of the perimeter qubits for each half, connect the two perimeters directly with CZs at the seam (rather than to a separating frame), and measure qubits that are no longer part of the perimeter. This leads to fewer active qubits at the top-most levels of the recursion, which is important because that part of the computation dominates the runtime.
\end{enumerate}
Similarly, we use a tableau-based Clifford simulator loosely based on Aaronson's CHP \cite{aaronson+gottesman:2004}, but it differs in two important ways:
\begin{enumerate}
    \item When a qubit is measured, we update the tableau to be a direct sum, to reflect that the qubit is no longer entangled with the rest. This requires more work, but helps clearly separate the past interactions of the qubit from the future, facilitating reuse. 
    \item Additional gates have been added: CZ because it is essential to graph states, single-qubit gates to translate $Y$-basis measurement into standard basis measurement, and SWAP gates.
\end{enumerate}

\section{Correction Subroutine Example}\label{sec:affine_example}
In this section we work through an explicit example of the correction subroutine described in \Cref{sec:corgeneral}. The input is shown in \Cref{fig:affine_graph}. We shall consider the case in which all four data qubits are to be measured in the $\pfont X$ basis (i.e., $U_{\mathrm{bases}}=H^{\otimes 4}$). The quantum circuit $\EuScript{C}'$ is shown in  \Cref{fig:affine_example_circuit}. Note that vertices of the graph are labelled by uppercase letters while qubits of $\EuScript{C}'$ are labelled by lowercase letters, with the data qubit and merge ancilla for vertex $B$ being labelled $b$ and $b'$, respectively. Let us take $\pi = (*1**1,*****)$ as the desired pattern, where qubits are ordered $ab'bcd$. That is, we would like $P_\textrm{cor}$ to have nonzero $\pfont X$-components acting on qubits $b'$ and $d$.

\begin{figure}
\begin{subfigure}{0.27\textwidth}
\centering
\begin{tikzpicture}
\tikzstyle{vertex}=[draw, circle, minimum size=20pt, inner sep=0pt, outer sep=0pt]
\node (C) [vertex] at (1.299,-0.75) {$C$};
\node (B) [vertex] at (0,0) {$B$};
\node (A) [vertex] at (0,1.5) {$A$};
\node (D) [vertex] at (-1.299,-0.75) {$D$};
\node (ghost) at (0,-1.75) {};

\draw (A) -- (B);
\draw (C) -- (B);
\draw (B) -- (D);
\end{tikzpicture}
\end{subfigure}
\begin{subfigure}{0.72\textwidth}
\centering
\begin{tikzpicture}
\tikzstyle{bag}=[rectangle, rounded corners=8pt, minimum size=30pt, text width=30pt, align=center, fill opacity=0.1, text opacity=1]
\tikzstyle{introduce}=[bag, fill=introduce]
\tikzstyle{forget}=[bag, fill=forget]
\tikzstyle{merge}=[bag, fill=merge]
\matrix[row sep=17pt, column sep=23pt]{
\node (ab_int) [introduce]{\phantom{$AB$}}; \node (ab_int_text) {$AB$}; &  \node (ab_forg) [forget]{\phantom{$B$}};  \node (ab_forg_text) {$B$}; & & & \\
&&\node (bb) [merge] {\phantom{$B$}}; \node (bb_text) {$B$};& \node (bd_int) [introduce] {\phantom{$BD$}}; \node (bd_int_text) {$BD$};& \node (bd_forg) [forget] {\phantom{$\varnothing$}};  \node (bd_forg_text) {$\varnothing$}; &\\
\node (bc_int) [introduce]{\phantom{$BC$}}; \node (bc_int_text) {$BC$}; &  \node (bc_forg) [forget]{\phantom{$B$}};  \node (bc_forg_text){$B$}; & & &\\
};
\newcommand{\roundec}{2.5pt}
\draw (ab_int) edge (ab_forg);
\draw ($(ab_forg.south east) - (\roundec,-\roundec)$) edge ($(bb.north west) - (-\roundec,\roundec)$);
\draw ($(bc_forg.north east) - (\roundec,\roundec)$) edge ($(bb.south west) - (-\roundec,-\roundec)$);
\draw (bc_int) edge (bc_forg);
\draw (bb) edge (bd_int);
\draw (bd_int) edge (bd_forg);

\tikzstyle{type}=[rectangle, rounded corners=2pt, fill opacity=.1]
\tikzstyle{legend}=[]
\node (INTRODUCE) [type, fill=introduce, below = 30pt of bd_int]{};
\node (INTRODUCELABEL) [legend, right=5pt of INTRODUCE]{Introduce node};
\node (FORGET) [type, fill=forget, below=5pt of INTRODUCE]{};
\node (FORGETLABEL) [legend, right=5pt of FORGET]{Forget node};
\node (MERGE) [type, fill=merge, below=5pt of FORGET]{};
\node (MERGELABEL) [legend, right=5pt of MERGE]{Merge node};
\end{tikzpicture}
\end{subfigure}
\caption{\emph{Left:} Input graph to the graph state simulation problem. \emph{Right:} A nice tree decomposition.}
\label{fig:affine_graph}
\end{figure}
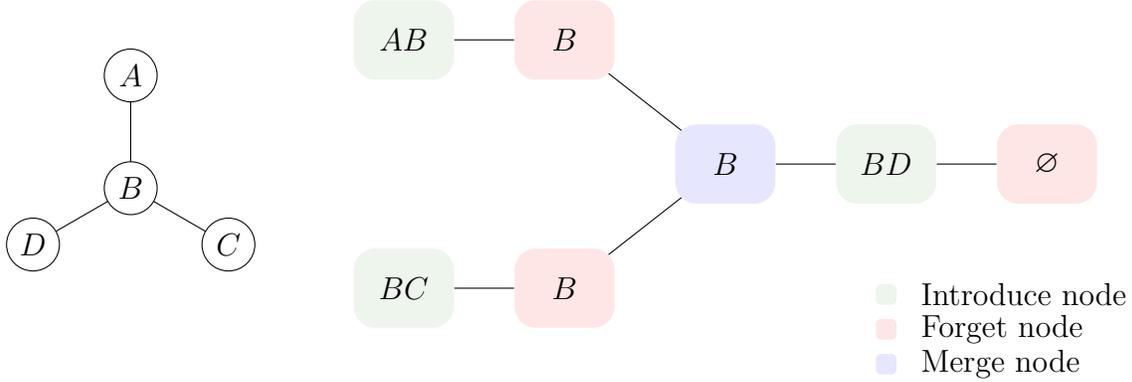

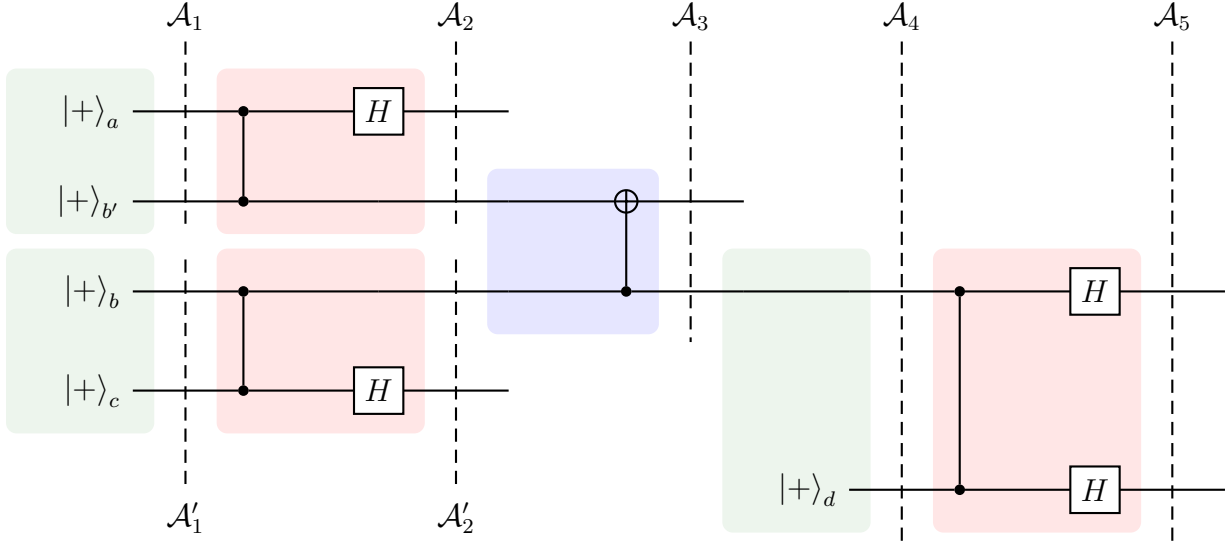
\begin{figure}
\begin{quantikz}[row sep=0.75cm, column sep = 1.4cm]
\gategroup[2,steps=2,style={draw = none, rounded corners,fill=introduce!10},
background]{}&\slice[style = {draw = black, shorten <= -0.5cm, shorten >= 4cm}, label style = {yshift = 0.5cm}]{$\mathcal{A}_1$} \slice[style = {draw = black, shorten <= 2.4cm, shorten >= 0.6cm}, label style = {yshift = -6.2cm}]{$\mathcal{A}'_1$} \lstick{$\ket{+}_a$} &\ctrl{1}\gategroup[2,steps=2,style={draw = none, rounded corners,fill=forget!10},
background]{} & \gate{H}\slice[style = {draw = black, shorten <= -0.5cm, shorten >= 4cm}, label style = {yshift = 0.5cm}]{$\mathcal{A}_2$}\slice[style = {draw = black, shorten <= 2.4cm, shorten >= 0.6cm}, label style = {yshift = -6.2cm}]{$\mathcal{A}'_2$}& \qw &           &   &                  &                      &  \\
&\lstick{$\ket{+}_{b'}$} & \control{} & \qw     & \qw\gategroup[2,steps=2,style={draw = none, rounded corners,fill=merge!10},
background]{} & \targ{}  \slice[style = {draw = black, shorten <= -0.5cm, shorten >= 2.5cm}, label style = {yshift = 0.5cm}]{$\mathcal{A}_3$}  &\qw&                  &                      &  \\
\gategroup[2,steps=2,style={draw = none, rounded corners,fill=introduce!10},
background]{}&\lstick{$\ket{+}_{b}$} & \ctrl{1}\gategroup[2,steps=2,style={draw = none, rounded corners,fill=forget!10},
background]{}  & \qw     & \qw & \ctrl{-1} &\qw \gategroup[3,steps=2,style={draw = none, rounded corners,fill=introduce!10},
background]{} & \qw   \slice[style = {draw = black, shorten <= -0.5cm, shorten >= -0.2cm}, label style = {yshift = 0.5cm}]{$\mathcal{A}_4$} &\ctrl{2}\gategroup[3,steps=2,style={draw = none, rounded corners,fill=forget!10},
background]{} & \gate{H}\slice[style = {draw = black, shorten <= -0.5cm, shorten >= -0.2cm}, label style = {yshift = 0.5cm}]{$\mathcal{A}_5$}& \qw  \\
&\lstick{$\ket{+}_c$} & \control{} & \gate{H}& \qw &           &   &                  &                      &   \\
&                   &            &         &     &           &   & \lstick{$\ket{+}_d$}& \control{}& \gate{H}&\qw   
\end{quantikz}
\caption{Example circuit $\circuit'$ constructed from the graph and tree decomposition in \Cref{fig:affine_graph}. Each affine space is associated with the bag immediately to the left of its dashed line. Qubits are ordered in the same way as $\pi$, from top to bottom.}
\label{fig:affine_example_circuit}
\end{figure}

We first compute the affine spaces $\mathcal{A}_1,\mathcal{A}_2,\ldots,\mathcal{A}_5$, which correspond to sections of $\circuit$ shown with dashed lines in \Cref{fig:affine_example_circuit}. We will omit the computation of $\mathcal{A}'_1$ and $\mathcal{A}'_2$ (which are also shown in \Cref{fig:affine_example_circuit}), because by symmetry they are isomorphic to $\mathcal{A}_1$ and $\mathcal{A}_2$, respectively. The leaf node containing $A$ and $B$ is an introduce node, so we set
\[
\mathcal{A}_1 = \left\{\begin{psmallmatrix}1&0\\0&1\\0&0\\0&0\end{psmallmatrix}x:x\in\mathbb{F}_2^2\right\},
\]
corresponding to the stabilizer generators $\pfont X_a$ and $\pfont X_{b'}$ of $\ket{+}_a\otimes\ket{+}_{b'}$. Throughout this example, we will order the rows of matrices first by grouping $\pfont X$- and $\pfont Z$-components together, and then in the same descending order as in $\pi$ and $\circuit'$. In the matrix above, the first and third rows correspond to qubit $a$ while the second and fourth correspond to $b'$.

Next we conjugate $\mathcal{A}_1$ by $(H\otimes I)\textrm{CZ}$. Since qubit $a$ is forgotten, we would in general need to enforce $\pi^{(X)}_a$ and $\pi^{(Z)}_a$, but since $\pi^{(X)}_a=\pi^{(Z)}_a=*$, there is nothing to be done. Therefore we have
\[
\mathcal{A}_2 = ((H\otimes I)\textrm{CZ})\mathcal{A}_1((H\otimes I)\textrm{CZ})^\dagger = \left\{\begin{psmallmatrix}0&1\\0&1\\1&0\\1&0\end{psmallmatrix}x:x\in\mathbb{F}_2^2\right\}.
\]

To compute $\mathcal{A}_3$, we first remove the rows of the matrix in $\mathcal{A}_2$ corresponding to $a$ (i.e. the first and third) to obtain $\{\begin{psmallmatrix} 0&1\\1&0\end{psmallmatrix}x:x\in\mathbb{F}_2^2\}$. By symmetry, we would obtain the same set after removing qubit $c$ from $\mathcal{A}'_2$. We then take the direct product
\[
\left\{\begin{psmallmatrix} 0&1\\1&0\end{psmallmatrix}x:x\in\mathbb{F}_2^2\right\}\oplus \left\{\begin{psmallmatrix} 0&1\\1&0\end{psmallmatrix}x:x\in\mathbb{F}_2^2\right\} \cong \left\{\begin{psmallmatrix} 0&1&0&0\\0&0&0&1\\1&0&0&0\\0&0&1&0\end{psmallmatrix}x:x\in\mathbb{F}_2^4\right\}.
\]
We have rearranged the right hand side matrix so that the $\pfont X$-components form the first two rows and the $\pfont Z$-components the last two, in contrast with Eq.~\eqref{eq:affine_product}. Next, we apply the CNOT gate to obtain
\begin{equation}\label{eq:affine_A3_intermediate}
\left\{\begin{psmallmatrix}0&1&0&1\\0&0&0&1\\1&0&0&0\\1&0&1&0 \end{psmallmatrix}x:x\in\mathbb{F}_2^4\right\}.
\end{equation}
Finally, we must enforce the pattern on the merge ancilla $b'$. Since $\pi_{b'}^{(X)} = 1$, this means finding the subset of elements in Eq.~\eqref{eq:affine_A3_intermediate} that have the form $(1,*,*,*)^T$. Basic linear algebra tells us that $x$ must take the form $x = e_4 + y_1e_1 + y_2(e_2+e_4) + y_3e_3$, where $e_j$ is the usual standard basis vector and $y_1,y_2,y_3\in\{0,1\}$ are free. The desired subset of Eq.~\eqref{eq:affine_A3_intermediate} therefore is
\[
\mathcal{A}_3 = \left\{\begin{psmallmatrix}0&1&0&1\\0&0&0&1\\1&0&0&0\\1&0&1&0 \end{psmallmatrix}(\begin{psmallmatrix}1&0&0\\0&1&0\\0&0&1\\0&1&0\end{psmallmatrix}y + \begin{psmallmatrix}0\\0\\0\\1\end{psmallmatrix}): y\in\mathbb{F}_2^3\right\}
=
\left\{
\begin{psmallmatrix}0&0&0\\0&1&0\\1&0&0\\1&0&1 \end{psmallmatrix} y + \begin{psmallmatrix}1\\1\\0\\0 \end{psmallmatrix} : y\in\mathbb{F}_2^3
\right\}.
\]
We then restrict $\mathcal{A}_3$ to coordinates corresponding to $b$ to obtain $\{\begin{psmallmatrix}0&1&0\\1&0&1 \end{psmallmatrix}x + \begin{psmallmatrix}
1\\0 \end{psmallmatrix}:x\in\mathbb{F}_2^3\}$. Since the matrix has more columns than rows, we search for a maximum linearly independent subset of columns. By inspection, we can just remove the third column to get $\{\begin{psmallmatrix}0&1\\1&0\end{psmallmatrix}x + \begin{psmallmatrix}
1\\0 \end{psmallmatrix}:x\in\mathbb{F}_2^2\}$. Taking a direct product with the affine space $\mathcal{I}$ for the stabilizer generator of $\ket{+}_d$ gives us
\[
\mathcal{A}_4 = \left\{\begin{psmallmatrix}0&1&0\\0&0&1\\1&0&0\\0&0&0\end{psmallmatrix}x + \begin{psmallmatrix}1\\0\\0\\0\end{psmallmatrix}:x\in\mathbb{F}_2^3\right\}.
\]

To compute $\mathcal{A}_5$, we conjugate the affine space by the CZ and Hadamard gates to obtain
\[
\left\{\begin{psmallmatrix}1&0&1\\0&1&0\\0&1&0\\0&0&1\end{psmallmatrix}x + \begin{psmallmatrix}0\\1\\1\\0\end{psmallmatrix}:x\in\mathbb{F}_2^3\right\},
\]
and since $\pi^{(X)}_d =1$, we search for all elements in this space of the form $(*,1,*,*)^T$. Using the same approach as in the computation of $\mathcal{A}_3$ we arrive at
\[
\mathcal{A}_5 = \left\{\begin{psmallmatrix}
1&1\\0&0\\0&0\\0&1 \end{psmallmatrix}y+ \begin{psmallmatrix}0\\1\\1\\0\end{psmallmatrix}:y\in\mathbb{F}_2^2\right\}.
\]

We are now ready to construct $P_\textrm{cor}$. We first choose a random element from $\mathcal{A}_5$, say, $(0,1,1,0)^T$. This fixes two tensor elements of the Pauli correction: $P_\textrm{cor}=\textrm{?}\otimes \textrm{?} \otimes \pfont Z\otimes \textrm{?} \otimes \pfont X$ (with the question marks to be determined below).

Since $((H\otimes H)CZ)^\dagger \pfont Z\otimes \pfont X ((H\otimes H)CZ) = \pfont X\otimes I$, pushing the chosen element of $\mathcal{A}_5$ backward through the $H$ and CZ gates gives us $(1,0,0,0)^T\in\mathcal{A}_4$. We then remove qubit $d$ to obtain $(1,0)^T$ and select a random element in $\mathcal{A}_3$ of the form $(*,1,*,0)^T$, say, $(1,1,0,0)^T$. Since $b'$ is measured immediately after the merge node, we can fill in another tensor element in the Pauli correction: $P_\textrm{cor} = \textrm{?}\otimes\pfont X \otimes\pfont Z\otimes \textrm{?} \otimes\pfont X$. 

Pushing $(1,1,0,0)^T$ back through the CNOT gate gives $(0,1,0,0)^T$. Removing $b$ gives us $(0,0)^T$, so we search for an element of $\mathcal{A}_2$ of the form $(*,0,*,0)^T$. The only such element is $(0,0,0,0)^T$, so we now have $P_\textrm{cor} = I\otimes\pfont X \otimes\pfont Z\otimes \textrm{?} \otimes\pfont X$. Similarly, we can reconstruct the tensor element corresponding to qubit $c$ and find that $P_\textrm{cor} = I\otimes\pfont X \otimes\pfont Z\otimes\pfont X \otimes\pfont X$ is a valid choice; indeed, we have $(01011,00100)\in\Pi$. We do not need to use $\mathcal{A}_1$ or $\mathcal{A}_1'$ as we now have all tensor elements of $P_\textrm{cor}$.

\section{Linear System Solving}
\label{sec:lss}

\linearsolver*

\begin{proof}
    Construct the circuit $\circuit$ associated with $G$ and its tree decomposition (in time linear in $\| T \|_1$). Claim~\ref{claim:graphgprime} says that $\circuit \ket{+^{n_t}}$ is a graph state $\ket{G'}$ for $G' = (V',E')$ which extends $G$ in the sense that $V' = V \cup M$ (where $M$ are the new vertices associated with the merge ancillas, i.e., $|M| = n_a$) and the adjacency matrix is 
    \begin{equation}
    A'=\begin{pmatrix} A & B \\ B^{T} & C\end{pmatrix}
    \label{eq:aprime}
    \end{equation}
    where $A$ is the adjacency matrix of $G$, and for some $n \times n_{a}$ matrix $B$, $n_a\times n_a$ matrix $C$.

    Suppose $x = x_V x_M, z = z_V z_M \in \mathbb{F}_2^{|V'|}$ are the $\pfont{X}$- and $\pfont{Z}$-components of a Pauli stabilizer for $\ket{G'}$. Since $\ket{G'}$ is a graph state, the components are related by the adjacency matrix: $z = A' x$. By Eq. \eqref{eq:aprime},
    \begin{equation}
    \begin{pmatrix} z_V \\ z_M \end{pmatrix} = \begin{pmatrix} A & B \\ B^{T} & C \end{pmatrix} \begin{pmatrix} x_V \\ x_M \end{pmatrix} = \begin{pmatrix} A x_V + B x_M \\ B^{T} x_V + C x_M \end{pmatrix}.
    \label{eq:linear}
    \end{equation}
    We claim there is a bijection between the set of solutions to $Ay = b$ and the set of solutions to Eq.~\eqref{eq:linear} which satisfy $z_V = b$ and $x_M = 0$ (i.e., stabilizers of $\ket{G'}$ of the form $(*_V 0_{M}, b_{V} *_{M})$). We will prove this claim by describing injective maps from each set to the other.
    
    First, suppose we are given $y$ such that $Ay = b$. Then setting $x_V = y$, $x_M = 0$, $z_V = b$, and $z_M = B^{T} y$ satisfies Eq.~\eqref{eq:linear}. Because we have set $x_V = y$, for two distinct solutions $y,y'$ we will find distinct stabilizers, and so this map is injective. Conversely, if we have a solution to Eq.~\eqref{eq:linear}, then the first equation, $z_V = A x_V + B x_M$ immediately implies $A x_V = b$ since $z_V = b$ and $x_M = 0$. If $(x_V 0_M, b_Vz_M)$ and $(x_V 0_M, b_V z'_M)$ represent stabilizers of $\ket{G'}$, then their product, represented by $(0_V 0_M, 0_V (z_M \oplus z'_M))$, is also a stabilizer. However because $\ket{G'}$ is a graph state, its only stabilizer with trivial $\pfont X$-component is the identity, which implies that $z_M = z'_M$, and thus this map is also injective.

    Therefore, we may apply Theorem~\ref{thm:powerful_correction} to sample a uniformly random stabilizer of $\ket{G'} = \circuit \ket{+^{n_t}}$ respecting the pattern $(*_V 0_M, b_V *_M)$ (or detect that no such stabilizer exists) in $O(\|T\|_{\omega}^{\omega})$ time. If a solution exists, then by the previous paragraph, it immediately yields a uniformly random solution to $Ay = b$.
\end{proof}
\end{document}